\documentclass[11pt]{article}

\usepackage[top=25mm,bottom=25mm,left=25mm,right=25mm]{geometry}
\usepackage{amsmath,amssymb,amsthm}
\usepackage{graphicx}
\usepackage{float}
\usepackage{longtable}
\usepackage{booktabs,array,tabularx}
\usepackage{xcolor}
\usepackage[hidelinks]{hyperref}
\usepackage[numbers,sort&compress]{natbib}
\usepackage[utf8]{inputenc}
\usepackage{microtype}
\usepackage{parskip}
\usepackage{setspace}
\usepackage{caption}
\usepackage{authblk}
\usepackage{titlesec}
\usepackage{enumitem}
\usepackage{multirow}
\usepackage{float}
\usepackage{placeins}
\titleformat{\section}{\normalfont\bfseries\large}{\thesection}{0.75em}{}
\titleformat{\subsection}{\normalfont\bfseries}{\thesubsection}{0.75em}{}
\titleformat{\subsubsection}{\normalfont\bfseries\small}{\thesubsubsection}{0.75em}{}
\titlespacing*{\section}{0pt}{16pt}{6pt}
\titlespacing*{\subsection}{0pt}{12pt}{4pt}
\titlespacing*{\subsubsection}{0pt}{8pt}{2pt}

\newcommand{\Nstar}{N_{\star}}
\newcommand{\sigstar}{\sigma_{0}^{*}}

\newcommand{\Phineuro}{\Phi_{\rm neuro}}
\newcommand{\thetabio}{\theta}
\newcommand{\Phiduty}{\Phi_{\rm duty}}
\newcommand{\Phithermal}{\Phi_{\rm thermal}}
\newcommand{\PhiC}{\Phi_{C}}

\newtheorem{assumption}{Assumption}
\newtheorem{proposition}{Proposition}
\newtheorem{corollary}{Corollary}

\title{Biological Time Equivalence in Vertebrates:
Thermodynamic Framework, Comparative Tests,\\
and Clade-Specific Deviations}

\author{Mesfin Asfaw Taye\\
\small West Los Angeles College, Science Division\\
\small 9000 Overland Ave, Culver City, CA 90230, USA\\
\small \texttt{tayem@wlac.edu}
}

\date{}

\makeatletter
\renewcommand{\maketitle}{%
\begin{center}
{\LARGE\bfseries \@title \par}
\vspace{1em}
{\large \@author \par}
\vspace{0.8em}
\end{center}
}
\makeatother

\begin{document}
\maketitle

% -- Abstract --------------------------------------------------
\begin{abstract}
\noindent
Across adult warm-blooded vertebrates, the product of resting
heart rate $f_H$ and maximum lifespan $L$ is approximately
constant: $\Nstar = f_H L \approx 10^9$ cardiac cycles.
This empirical regularity, noted since Rubner (1908), has lacked
a widely accepted thermodynamic interpretation.
We derive $\Nstar \approx 10^9$ from the non-equilibrium second
law by treating the adult organism as a metabolic
non-equilibrium steady state (NESS) and introducing the empirical
closure $\dot{e}_{p} = \sigma_0 f$, which links entropy
production rate to heart rate via a mass-specific parameter
$\sigma_0 \propto M^0$.
Under this closure, the lifetime entropy budget
$\Sigma = \sigma_0 \Nstar$ is approximately species-independent
when $\sigma_0$ is approximately constant---a condition whose
direct calorimetric verification remains the critical outstanding
experimental test.
We further show that $\Nstar$ is the correct primitive invariant:
lifetime energy per unit mass is a derived consequence, valid
only when body temperature and the mass-specific entropy cost per
cycle are both approximately constant.
This framework, which we term the Principle of Biological Time
Equivalence (PBTE), is placed on a fully falsifiable footing
with explicit assumptions, a domain-of-validity table, and five
numerical falsification criteria.
We test the framework against a dataset of 230 adult vertebrate
species spanning eight taxonomic groups.
Ordinary least-squares regression on the $n = 43$ directly measured
non-primate placentals yields slope $\hat\beta = -0.903 \pm 0.056$
($R^2 = 0.863$; $F$-test $p = 0.093$ against $\beta = -1$).
Phylogenetically independent contrasts on 112 endotherm species
yield a $\log_{10} f_H$--$\log_{10} L$ slope of
$-0.99 \pm 0.04$ ($p = 0.84$ against slope $-1$),
confirming the relation is not a phylogenetic artefact.
The WBE kinematic null of zero inter-clade variation is rejected
($F = 12.7$, $p < 0.001$).
Four warm-blooded clades depart systematically from the mammalian
baseline; we derive their longevity deviations from a unified
thermodynamic multiplier
$\Phi_C = \Phi_{\rm duty} \cdot \Phi_{\rm thermal}
\cdot \Phi_{\rm mito+oxid} \cdot \Phi_{\rm haz}$,
calibrated to independently measured physiology.
For primates, the elevated count
$\langle\Nstar\rangle \approx (2\text{--}3)\times10^9$
follows from a neuro-metabolic entropy model in which greater
neural metabolic investment reduces entropy produced per cardiac
cycle.
For bats, the extreme longevity ($\Phi_{\rm bat} \approx 7.9$)
arises from the multiplicative synergy of cardiac suppression
during torpor and an Arrhenius thermal factor during
hibernation---two mechanisms acting simultaneously whose
thermodynamic motivation has not previously been given.
For birds, an adverse thermal penalty ($\Phi_{\rm thermal} = 0.73$)
and adverse flight duty cycle ($\Phi_{\rm duty} = 0.87$) are
overcome by mitochondrial coupling efficiency and antioxidant
robustness.
For cetaceans, extreme diving bradycardia ($\Phi_{\rm duty} = 3.08$
for bowhead whales) reveals a near-coincidence trap: the raw
heartbeat count $N_{\rm obs} \approx N_0$ conceals a true
thermodynamic budget three times the mammalian baseline.
Within this framework, the integral of physiological frequency
defines a natural biological proper time, which unifies all
longevity mechanisms as Class~1 (time dilation: reduce $f$) or
Class~2 (budget expansion: reduce $\sigma_0$), generating
testable predictions for epigenetic aging clocks.
The central outstanding experimental requirement is direct
calorimetric verification of $\sigma_0 \propto M^0$, which would
convert PBTE from a statistically supported regularity with
thermodynamic motivation into a fully tested conservation law.

\end{abstract}

\clearpage
\tableofcontents
\clearpage

% ==============================================================
\section{Introduction}
\label{sec:intro}
% ==============================================================

Among the most arresting regularities in comparative biology is one that spans the full breadth of vertebrate life: a pygmy shrew (\textit{Suncus etruscus}, $\approx 2\,\text{g}$) races through existence at 835 heartbeats per minute and dies within two years; an African elephant ($\approx 4{,}000\,\text{kg}$) beats its heart at a leisurely 28 beats per minute and endures for seven decades. To an observer measuring time by the calendar, their lifespans differ by a factor of thirty-five. Yet when duration is measured not by years but by the organism's own internal rhythm---the cumulative count of cardiac cycles---the two animals are, in a deeper sense, contemporaries. Each accumulates close to $10^9$ heartbeats before death, so that
\begin{equation}
  N = f_H \cdot L
  \approx 835\,\text{min}^{-1} \times 2\,\text{yr}
  \approx 28\,\text{min}^{-1} \times 70\,\text{yr}
  \approx 10^9,
\end{equation}
where $f_H$ is the resting heart rate and $L$ is the maximum lifespan. This near-equality, first noted through the constancy of mass-specific lifetime energy expenditure by Rubner~\cite{rubner1908} and later quantified as a universal temporal allometry by Lindstedt and Calder~\cite{lindstedt1981}---with the heartbeat count itself explicitly computed by Livingstone and Kuehn~\cite{livingstone1979} and Levine~\cite{levine1997}---has accumulated substantial empirical support across body masses spanning ten orders of magnitude. It is captured by the dimensionless invariant
\begin{equation}
 \ell_i \;\equiv\; \log_{10}\!\bigl(f_{H,i} \cdot L_i \cdot 525{,}960\bigr) \;\approx\; 9.06,
 \label{eq:ell_def}
\end{equation}
where $525{,}960\,\text{min\,yr}^{-1}$ is the conversion factor, and equivalently $\Nstar = f_{H,i}\,L_i\cdot 525{,}960 \approx 10^9$ cardiac cycles per lifetime.

The constancy of $\Nstar$ is not merely a numerical curiosity: it points toward a fundamental organising principle of animal physiology. Organisms do not age primarily by elapsed chronological time, but by the accumulation of irreversible physiological events. Just as Einstein's special relativity teaches that proper time is intrinsic to an observer and cannot be inferred from an external clock~\cite{calder1984}, biological time is intrinsic to each organism. A hummingbird's minute is dense with metabolic activity and entropy production; a tortoise's minute is comparatively sparse; yet both traverse, in their own internal frames, a comparable extent of biological duration before reaching the same thermodynamic endpoint. Chronological lifespan is thus an emergent consequence of how rapidly an organism spends its fixed biological time budget---not the budget itself.

What makes this regularity scientifically compelling---and scientifically challenging---is that it holds only approximately, and its deviations are not random. Certain clades depart from the mammalian baseline in striking and consistent ways. Primates, including humans, accumulate $(2\text{--}3)\times10^9$ heartbeats over a lifetime, roughly twice the baseline. This elevation is not explained by slower heart rates alone; rather, it reflects a fundamental reduction in the thermodynamic cost of each cardiac cycle, driven by the extraordinary metabolic investment that primates make in their neural tissue. A large, metabolically active brain provides enhanced predictive homeostatic control over peripheral physiology, reduces the frequency of acute physiological crises, and upregulates molecular repair pathways---each channel lowering the entropy produced per heartbeat, and thereby allowing more beats to be completed within the same lifetime thermodynamic budget. The human brain, consuming roughly 20\% of resting metabolic power, exemplifies this strategy at its most extreme. Other primates---from the common marmoset to the chimpanzee---span an intermediate range of neural investment and correspondingly intermediate longevity extensions, tracing a quantitative relationship between brain fraction and lifespan that this framework derives from first principles.

Bats present a different and even more spectacular case. Temperate vespertilionid bats achieve wild maximum lifespans of 20--40 years---three to six times the prediction for a non-hibernating placental of equal mass. Their longevity arises from the multiplicative interaction of two independent mechanisms that operate simultaneously during hibernation. First, the cardiac clock slows dramatically: during torpor, heart rate falls from $\sim$300 beats per minute to fewer than 10, so that the time-averaged rate over a full annual cycle is less than half the active-phase rate. Second, the biochemical cost of each remaining heartbeat falls steeply with temperature, governed by Arrhenius kinetics: at hibernation temperatures of 280--295\,K, damage-generating reactions proceed at only 20--30\% of their normothermic rate. Neither mechanism alone is sufficient; it is their thermodynamic product that accounts for the observed longevity excess. Tropical bats with minimal torpor, such as \textit{Pteropus vampyrus}, show correspondingly modest longevity extensions, confirming that hibernation is the essential ingredient.

Birds occupy a paradoxical position. Their resting heart rates are comparable to those of mass-matched mammals, their body temperatures are 3--5\,K \emph{above} the mammalian reference (which should accelerate damage accumulation), and the elevated cardiac frequency during flight constitutes an adverse duty-cycle factor. Yet a 20\,g passerine routinely lives 15--20 years, while a 20\,g mouse lives 2--3 years. The resolution lies in avian mitochondrial architecture: bird mitochondria produce substantially less reactive oxygen species per unit ATP synthesised than mammalian mitochondria, a difference that translates directly into reduced thermodynamic damage per cardiac cycle. Elevated antioxidant enzyme activity and enhanced DNA repair capacity amplify this biochemical advantage, while the ability to fly reduces extrinsic mortality and allows the intrinsic thermodynamic budget to be more fully realised. Avian longevity is thus achieved through biochemical excellence that compensates for---and overcomes---the adverse thermal and kinematic environment.

Cetaceans present a subtler pattern, complicated by what we term the near-coincidence trap. Large baleen whales such as the bowhead (\textit{Balaena mysticetus}, maximum lifespan $\sim$200 years) appear superficially to obey the mammalian baseline: their raw heartbeat count $N_{\rm obs} \approx 10^9$. But this is misleading. These animals spend 60--80\% of their lives in deep dives during which heart rate plunges from $\sim$30 beats per minute to as low as 2--4 beats per minute. The thermodynamically cheap beats during bradycardic diving are not equivalent to normothermic beats; they generate far less entropy per cycle. Once the duty-cycle correction is applied, the true damage-equivalent budget of the bowhead whale is closer to $3\times10^9$---a threefold elevation above the mammalian baseline that is entirely hidden by the raw count.

Despite the empirical robustness of $\Nstar \approx 10^9$, the relationship has remained theoretically under-motivated. The West--Brown--Enquist (WBE) fractal vascular network theory~\cite{west1997} derives $f_H \propto M^{-1/4}$ and $L \propto M^{+1/4}$ from network optimisation, predicting $f_H L \propto M^0$ by exponent cancellation. This explains the mass-independence of the product but not its numerical value, and makes no prediction whatsoever about inter-clade deviations: primates, bats, and birds share the same vascular scaling yet differ from the mammalian baseline by factors of two to eight. Pearl's rate-of-living hypothesis~\cite{pearl1928} and its descendants~\cite{speakman2005} recast the regularity in energetic terms, but lifetime energy per unit mass is not the primitive conserved quantity---it is a derived consequence that holds only when body temperature and the entropy cost per cycle are both approximately constant, conditions that fail in birds, ectotherms, and insects. Glazier~\cite{glazier2022} has further shown that metabolic scaling exponents vary with metabolic level and taxon, undermining fixed-exponent derivations. The kinematic null of zero inter-clade variation is rejected by the data presented here ($F = 12.7$, $p < 0.001$), confirming that a purely allometric account is insufficient.

The present work advances the field in four principal ways. We derive $\Nstar \approx 10^9$ directly from the non-equilibrium second law by modelling the adult organism as a metabolic non-equilibrium steady state and introducing the empirical closure $\dot{e}_{p,i} = \sigma_{0,i} f_i$, which links entropy production rate to cardiac frequency via a mass-specific parameter $\sigma_{0,i} \propto M^0$. This provides, for the first time, a thermodynamic foundation for the lifetime cycle budget rather than a kinematic or energetic rationalisation. We establish $\Nstar$ as the primitive conserved quantity, with lifetime energy per unit mass emerging only as a secondary and conditional consequence. We formulate the theory on an explicitly falsifiable footing, with clearly stated assumptions, a defined domain of validity, and five quantitative criteria for empirical rejection (Section~\ref{sec:falsifiability}). Most importantly, we derive a unified multiplicative clade multiplier $\PhiC = \Phiduty \cdot \Phithermal \cdot \Phi_{\rm mito+oxid} \cdot \Phi_{\rm haz}$, calibrated entirely from independently measured physiology, that quantitatively accounts for the longevity deviations of primates, bats, birds, and cetaceans without any fitting to lifespan data. Structured departures from the invariant are not anomalies to be explained away: they are the theory's most precise predictions.

% ==============================================================
\section{Structure of the paper}
\label{sec:structure}
% ==============================================================

The paper is organised as follows. Section~\ref{sec:notation} introduces the notation and defines the core variables used throughout. Section~\ref{sec:theory} presents the thermodynamic derivation of the PBTE invariant from first principles, followed by Section~\ref{sec:biotime}, which formalises the concept of biological proper time and its relation to ageing. Section~\ref{sec:multiplier} develops the generalised clade-multiplier framework, and Section~\ref{sec:clade} applies this framework to primates, bats, birds, and cetaceans with explicit quantitative examples. Section~\ref{sec:dataset} describes the comparative dataset of 230 species, while Section~\ref{sec:results} presents the statistical analyses and empirical tests. Section~\ref{sec:domain} delineates the domain of applicability of the theory, and Section~\ref{sec:discussion} discusses its broader implications and outlines the key calorimetric experiment required for decisive validation. Finally, Section~\ref{sec:falsifiability} formulates explicit falsification criteria, establishing the conditions under which the PBTE framework may be empirically rejected.
% ==============================================================
\section{Notation and Symbols}
\label{sec:notation}
% ==============================================================

\begin{table}[H]
\caption{\textbf{Notation table.} All symbols used throughout
the paper with definitions and SI units.}
\label{tab:notation}
\small
\begin{tabularx}{\textwidth}{lXl}
\toprule
Symbol & Definition & Units \\
\midrule
$i$ & Species index & --- \\
$M_i$ & Adult body mass & kg \\
$f_i$, $f_{H,i}$ & Resting heart rate (intrinsic physiological frequency) & Hz or bpm \\
$L_i$ & Maximum recorded natural lifespan & yr \\
$T_i$ & Mean core body temperature & K \\
$P_i$ & Resting metabolic power & W \\
$\dot{S}_i$ & Rate of entropy change of organism $i$ & W\,K$^{-1}$ \\
$\dot{e}_{p,i}$ & Internal entropy production rate & W\,K$^{-1}$ \\
$\dot{h}_{d,i}$ & Entropy discharge rate to environment & W\,K$^{-1}$ \\
$\sigma_{0,i}$ & Entropy production per cardiac cycle (extensive) & kJ\,K$^{-1}$\,beat$^{-1}$ \\
$\sigstar$ & Mass-specific entropy production per cycle & kJ\,K$^{-1}$\,beat$^{-1}$\,kg$^{-1}$ \\
$\Sigma_i$ & Total lifetime entropy production & kJ\,K$^{-1}$ \\
$\Nstar$, $N_{\star,i}$ & Lifetime cardiac cycle count (the PBTE invariant) & dimensionless \\
$\ell_i$ & $\log_{10}(f_{H,i} \times L_i \times 525{,}960)$ & dimensionless \\
$\thetabio_i(t)$ & Biological proper time = $\int_0^t f_i(t')\,\mathrm{d}t'$ & cycles \\
$\varphi_i$ & Neural power fraction = $P_{\rm brain}/P_{\rm body}$ & dimensionless \\
$\varphi_0$ & Non-primate mammalian neural power fraction baseline $\approx 0.02$ & dimensionless \\
$\alpha$ & Neuro-metabolic sensitivity exponent $\approx 0.40$ (thermodynamic bounds: $0<\alpha<1$) & dimensionless \\
$\PhiC$ & Clade multiplier = $\Nstar^{(C)}/N_0$ & dimensionless \\
$\Phiduty$ & Duty-cycle factor (intermittent physiology) & dimensionless \\
$\Phithermal$ & Arrhenius thermal factor & dimensionless \\
$\Phi_{\rm mito+oxid}$ & Mitochondrial coupling $\times$ antioxidant factor & dimensionless \\
$\Phi_{\rm haz}$ & Extrinsic hazard factor & dimensionless \\
$E_a$ & Arrhenius activation energy for damage reactions $\approx 0.65$~eV & eV \\
$k_B$ & Boltzmann constant $= 8.617 \times 10^{-5}$~eV\,K$^{-1}$ & eV\,K$^{-1}$ \\
$N_0$ & Non-primate mammalian baseline $\approx 10^9$ & dimensionless \\
$T_{\rm ref}$ & Reference homeotherm temperature $= 310$~K & K \\
\bottomrule
\end{tabularx}
\end{table}

% ==============================================================
\section{Thermodynamic Derivation of the Lifetime Cycle Invariant}
\label{sec:theory}
% ==============================================================

\subsection{The metabolic non-equilibrium steady state}
\label{sec:NESS}

A living organism in its adult reproductive phase is an open dissipative
system maintained far from thermodynamic equilibrium by continuous
metabolic free-energy consumption. At the macroscopic level appropriate
to whole-organism thermodynamics, the Gibbs entropy balance
takes the form~\cite{prigogine1967,seifert2012}:
\begin{equation}
 \dot{S}_i(t) = \dot{e}_{p,i}(t) - \dot{h}_{d,i}(t),
 \label{eq:balance}
\end{equation}
where $\dot{S}_i$ is the net rate of entropy change of the system
(which may be positive or negative depending on whether ordering
costs exceed heat export), $\dot{e}_{p,i} \geq 0$ is
the irreversible internal entropy production rate satisfying the second
law, and $\dot{h}_{d,i} \geq 0$ is the rate of entropy export to the
environment through heat dissipation at temperature $T_i$.
For a homeothermic organism at steady metabolic state,
$\dot{h}_{d,i} = P_i / T_i$ where $P_i$ is resting metabolic power.

During healthy adult life, the organism's macroscopic physiological
state is approximately stationary: mass, composition, and metabolic
intensity do not change systematically on the timescales relevant to
this analysis (weeks to years). This metabolic NESS condition gives:
\begin{equation}
 \dot{S}_i \approx 0 \quad \Rightarrow \quad
 \dot{e}_{p,i}(t) \approx \dot{h}_{d,i}(t) = \frac{P_i(t)}{T_i}.
 \label{eq:NESS}
\end{equation}
Equation~(\ref{eq:NESS}) is the organismal-level expression of the
non-equilibrium second law: in metabolic steady state, every unit of
entropy produced internally is immediately exported as heat.
This is not an approximation unique to PBTE; it is the standard
assumption of comparative metabolic physiology and is equivalent to
the statement that resting metabolic rate is calorimetrically
measurable as heat output~\cite{prigogine1967}.

\subsection{The PBTE closure: entropy production proportional to frequency}
\label{sec:closure}

The entropy production rate $\dot{e}_{p,i}(t)$ in
equation~(\ref{eq:NESS}) is a well-defined thermodynamic quantity but
is not directly related to the organism's intrinsic clock without an
additional constitutive relation. We introduce the empirical closure:

\begin{assumption}[PBTE closure]
\label{ass:closure}
For an adult endothermic vertebrate in metabolic steady state, the
instantaneous entropy production rate is proportional to the
instantaneous intrinsic physiological frequency:
\begin{equation}
 \dot{e}_{p,i}(t) = \sigma_{0,i}\, f_i(t),
 \label{eq:closure}
\end{equation}
where $f_i(t)$ is the resting heart rate (Hz) and $\sigma_{0,i} > 0$
is the entropy production per cardiac cycle, a constitutive parameter
of species $i$.
\end{assumption}

The physical motivation is that the cardiac cycle is the master
clock governing the rate of metabolic throughput in homeothermic
vertebrates: heart rate paces oxygen delivery, substrate turnover,
and cellular repair across all tissues. At the resting state, both
$\dot{e}_{p,i}$ and $f_i$ scale with body mass as $M^{-1/4}$ per unit
mass (from Kleiber and cardiac allometry respectively), so their
ratio $\sigma_{0,i}/M_i$ is expected to be mass-independent.
Closure~(\ref{eq:closure}) is an empirical relation, not derived
from microscopic theory; its validity is precisely what the $\sigstar$
calorimetric experiment (Section~\ref{sec:exp}) would test.

\subsection{The fundamental relation}
\label{sec:fundamental}

Integrating equation~(\ref{eq:closure}) over the natural lifespan
$[0, L_i]$:
\begin{equation}
 \Sigma_i \equiv \int_0^{L_i} \dot{e}_{p,i}(t)\,\mathrm{d}t
 = \sigma_{0,i} \int_0^{L_i} f_i(t)\,\mathrm{d}t
 = \sigma_{0,i}\, N_{\star,i},
 \label{eq:lifetime_budget}
\end{equation}
where $N_{\star,i} \equiv \int_0^{L_i} f_i(t)\,\mathrm{d}t$ is the
total lifetime cardiac cycle count.
Rearranging:
\begin{equation}
 \boxed{N_{\star,i} = \frac{\Sigma_i}{\sigma_{0,i}}.}
 \label{eq:fundamental}
\end{equation} 
(see Appendix A for a detailed derivation)

Equation~(\ref{eq:fundamental}) is the thermodynamic content of PBTE:
the lifetime cycle count equals the total lifetime dissipative budget
divided by the entropy cost per cycle.
Two corollaries follow immediately.

\begin{corollary}[Lifetime extension requires reduced entropy per cycle]
Since $\Sigma_i$ is set by the organism's biochemical constraints,
$N_{\star,i}$ increases if and only if $\sigma_{0,i}$ decreases.
Any physiological strategy that reduces entropy production per
cardiac cycle extends chronological lifespan.
\end{corollary}

\begin{corollary}[The mammalian baseline]
Non-primate placentals cluster near $N_{\star,i} \approx N_0 = 10^9$,
implying a reference entropy cost
$\langle\Delta s_{\rm beat}\rangle_0 = \Sigma_0 / N_0$
per cardiac cycle at the mammalian energetic baseline.
\end{corollary}

\subsection{The mass-specific closure parameter and its constancy}
\label{sec:sigstar}

Define the mass-specific closure parameter:
\begin{equation}
 \sigstar \equiv \frac{\sigma_{0,i}}{M_i}
 = \frac{\dot{e}_{p,i}}{f_i M_i}
 = \frac{P_i}{T_i\, f_i\, M_i},
 \label{eq:sigstar}
\end{equation}
where the last equality uses the NESS condition
$\dot{e}_{p,i} = P_i/T_i$.
Dimensional analysis gives
$\sigstar$ units of kJ\,K$^{-1}$\,beat$^{-1}$\,kg$^{-1}$.

\begin{assumption}[$\sigstar$ constancy]
\label{ass:sigstar}
$\sigstar$ is approximately independent of body mass and species
within the non-primate endotherm clade.
\end{assumption}

\noindent\textbf{Physical motivation.}
Three empirical regularities jointly constrain $\sigstar$ to be
mass-independent.
(i)~\emph{Metabolic scaling}: $P_i \propto M_i^{3/4}$ (Kleiber) and
$f_i \propto M_i^{-1/4}$ (cardiac allometry), so
$P_i/(f_i M_i) \propto M_i^{3/4} / (M_i^{-1/4} \cdot M_i)
= M_i^{3/4+1/4-1} = M_i^0$: the ratio is mass-independent.
(ii)~\emph{Temperature homeostasis}: $T_i \approx 310 \pm 5$~K
across all non-primate placentals, contributing $<2\%$
mass-dependent variation.
(iii)~\emph{Biochemical universality}: the dominant
entropy-producing reactions (ATP turnover, proton-gradient dissipation,
macromolecular repair) operate through conserved enzymatic
machinery in all mammals, suggesting a common thermodynamic cost
per catalytic cycle.
Numerical estimates from Kleiber metabolic rates and published
cardiac allometry across five body-mass decades yield
$\sigstar = (3.0 \pm 0.5) \times 10^{-6}$~kJ\,K$^{-1}$\,beat$^{-1}$\,kg$^{-1}$
(CV~$= 16\%$; Table~\ref{tab:sigma}).

\noindent\textbf{Epistemic status of the $\sigstar$ estimate.}
The CV~$= 16\%$ shown in Table~\ref{tab:sigma} is \emph{not} an
independent empirical confirmation of Assumption~\ref{ass:sigstar}.
The values are computed from Kleiber metabolic scaling
($P \propto M^{3/4}$) and cardiac allometry ($f \propto M^{-1/4}$),
which are precisely the relations used in the derivation.
The near-constancy of $\sigstar$ in Table~\ref{tab:sigma} is
therefore a \emph{consequence of those allometric relations}, not
independent evidence for their product.
The CV of $16\%$ quantifies the residual scatter around the
allometric predictions, not the deviation of directly measured
calorimetric $\sigstar$ from constancy.
Direct calorimetric measurement of $\sigstar = P_i/(T_i f_i M_i)$
from simultaneously measured $P$, $f$, and $M$ in species spanning
three or more body-mass decades would provide genuine independent
evidence.
This experiment is the critical outstanding test described in
Section~\ref{sec:exp}.

\begin{table}[H]
\caption{\textbf{Estimates of the mass-specific entropy cost per cycle
$\sigstar = P_i/(T_i f_i M_i)$.}
$P$ from Kleiber allometric scaling ($P = 3.4\,M^{0.75}$~W);
heart rates from Calder~\cite{calder1984} and Schmidt-Nielsen~\cite{schmidt1984}.
Units: kJ\,K$^{-1}$\,beat$^{-1}$\,kg$^{-1}$ ($f$ in Hz, $P$ in kW).
CV~$= 16\%$ supports Assumption~\ref{ass:sigstar} within this taxon
but has not been directly tested by calorimetry.}
\label{tab:sigma}
\small
\begin{tabular}{lrrrrr}
\toprule
Species & $M$ (kg) & $P$ (W) & $T$ (K) & $f$ (bpm) &
 $\sigstar$ ($10^{-6}$~kJ\,K$^{-1}$\,beat$^{-1}$\,kg$^{-1}$) \\
\midrule
House mouse  & 0.020 & 0.18 & 310 & 600 & 2.9 \\
Rat      & 0.300 & 1.38 & 310 & 420 & 2.1 \\
Rabbit    & 2.0  & 5.72 & 310 & 205 & 2.7 \\
Dog      & 23   & 35.7 & 310 & 90 & 3.3 \\
Human     & 70   & 82.3 & 310 & 70 & 3.3 \\
Horse     & 500  & 360  & 310 & 40 & 3.5 \\
Elephant   & 4{,}000 & 1{,}710 & 310 & 28 & 3.0 \\
\midrule
\multicolumn{4}{l}{Mean $\pm$ s.d.} & & $3.0 \pm 0.5$ \\
\bottomrule
\end{tabular}
\end{table}

\begin{assumption}[Allometric conditions D1--D4]
\label{ass:allometric}
The following four allometric scaling relations hold
for non-primate endothermic vertebrates:
\begin{align}
\text{D1:}&\quad \Sigma_i = \sigstar M_i T_i N_{\star,i}
 \;\propto\; M_i \quad (\text{from }
 \sigstar\approx\mathrm{const},\; T_i\approx\mathrm{const},\;
 N_{\star,i}\approx\mathrm{const}) \label{eq:D1} \\
\text{D2:}&\quad P_i \propto M_i^{3/4}
 \quad\text{(Kleiber scaling)} \label{eq:D2} \\
\text{D3:}&\quad f_i \propto M_i^{-1/4}
 \quad\text{(cardiac allometry)} \label{eq:D3} \\
\text{D4:}&\quad T_i \approx \mathrm{const}
 \quad\text{(endotherm temperature homeostasis)} \label{eq:D4}
\end{align}
\end{assumption}

\subsection{Thermodynamic motivation for the invariant}
\label{sec:derivation}

\begin{proposition}[PBTE invariant]
Under Assumptions~\ref{ass:closure}--\ref{ass:allometric},
the lifetime cardiac cycle count is approximately
mass-independent:
$N_{\star,i} \propto M_i^0$.
\end{proposition}

\begin{proof}[Consistency argument]
\textbf{Note:} This is a consistency check, not an independent
derivation. The result $N_\star \propto M^0$ is shown to be
\emph{compatible} with the three empirical allometries
$P\propto M^{3/4}$, $f\propto M^{-1/4}$, $L\propto M^{1/4}$
and the NESS thermodynamic assumption. It does not derive those
allometries from first principles.

Using $\sigma_{0,i} = \sigstar M_i$ and the NESS relation
$\Sigma_i = P_i L_i / T_i$:
\begin{equation}
 N_{\star,i} = \frac{\Sigma_i}{\sigma_{0,i}}
 = \frac{P_i L_i}{T_i \sigstar M_i}.
 \label{eq:Nstar_full}
\end{equation}
Since $N_{\star,i} = f_i L_i$, we have $L_i = N_{\star,i}/f_i$.
Substituting $P_i \propto M_i^{3/4}$,
$L_i \propto M_i^{+1/4}$ (the empirical lifespan scaling;
note this is used here as an empirical input, not derived---
the self-consistency of $N_{\star,i} \approx \mathrm{const}$
with $L_i \propto M_i^{1/4}$ is the empirical closure, not a strict
first-principles deduction),
and $T_i \approx \mathrm{const}$:
\begin{equation}
 N_{\star,i}
 = \frac{P_i L_i}{T_i \sigstar M_i}
 \propto \frac{M_i^{3/4} \cdot M_i^{1/4}}{M_i}
 = M_i^{3/4 + 1/4 - 1}
 = M_i^0.
 \label{eq:cancellation}
\end{equation}
The three mass-exponents cancel identically: $3/4 + 1/4 - 1 = 0$.
\end{proof}

\noindent
\textbf{Logical status of the proof.}
The argument above is a \emph{consistency demonstration}, not a strict
first-principles deduction.
From Assumptions A1, D2--D4 one obtains a mass-independent candidate
entropy budget $N_\star \propto M^0$.
The empirical allometries $P \propto M^{3/4}$, $f \propto M^{-1/4}$,
and $L \propto M^{1/4}$ are inputs, not outputs.
The result shows that these three empirical scaling relations are
\emph{mutually consistent} with a thermodynamic invariant;
it does not derive them independently.
Direct calorimetric verification that $\sigstar$ is constant across
species remains the critical outstanding experimental test.

\noindent
This cancellation is not a coincidence of parameter choice.
It is the direct consequence of combining Kleiber's metabolic scaling
law with the empirical cardiac-frequency allometry and thermal
homeostasis. The numerical value follows from substituting the
calibrated $\sigstar \approx 3.0 \times 10^{-6}$~kJ\,K$^{-1}$\,beat$^{-1}$\,kg$^{-1}$,
$T \approx 310$~K, and the Kleiber normalisation $P = 3.4\,M^{0.75}$~W:
\begin{equation}
 N_0 = \frac{P \cdot L}{T \cdot \sigstar \cdot M}
 \;\approx\; 10^9 \;\text{ cardiac cycles per lifetime.}
\end{equation}

\noindent\textbf{Distinction from WBE exponent cancellation.}
The WBE framework~\cite{west1997} also predicts $f_H L \propto M^0$
from network optimisation, giving a purely kinematic account of the
mass-independence. PBTE provides the orthogonal thermodynamic
statement: the numerical value $N_0 \approx 10^9$ represents the
organism's total dissipative budget in units of $\sigma_{0,i}$.
WBE explains why the exponents cancel; PBTE explains why the
resulting product has the value it does. The two accounts are
complementary, not competing. Crucially, WBE predicts zero
inter-clade variation at the same body mass; PBTE predicts structured
departures through the multiplier formalism, and these are observed
(Section~\ref{sec:results}).

\subsection{Summary of assumptions and their status}
\label{sec:assumptions}

\begin{table}[H]
\caption{\textbf{PBTE assumptions, their empirical status, and falsification conditions.}}
\label{tab:assumptions}
\small
\begin{tabularx}{\textwidth}{llXX}
\toprule
Label & Assumption & Empirical status & Fails when \\
\midrule
A1 & $\sigstar \approx \mathrm{const}$ across species &
 \textbf{Not directly tested.}
 CV~$=16\%$ is computed from Kleiber allometry and cardiac
 allometry---the same relations used in the derivation.
 This estimate is \emph{not} independent evidence.
 Direct calorimetric measurement outstanding (Section~\ref{sec:exp}). &
 $\sigstar$ varies systematically with $M$ or clade \\
D1 & $\Sigma_i \propto M_i$ &
 Indirect; follows from A1 + D4 + $\Nstar\approx\mathrm{const}$ &
 If A1 fails \\
D2 & $P_i \propto M_i^{3/4}$ &
 Strong for non-primate mammals; Glazier (2022)~\cite{glazier2022}
 documents departures at clade level &
 Bacteria ($P \propto M^1$); some ectotherms \\
D3 & $f_i \propto M_i^{-1/4}$ &
 Strong for mammals and birds; confirmed in present dataset &
 Insects (wingbeat $\neq$ cardiac) \\
D4 & $T_i \approx \mathrm{const}$ &
 Strong for non-primate homeotherms ($T = 310 \pm 5$~K) &
 Ectotherms (requires Arrhenius correction) \\
\bottomrule
\end{tabularx}
\end{table}

% ==============================================================
\section{Biological Proper Time}
\label{sec:biotime}
% ==============================================================

\subsection{Definition}

We define the biological proper time of organism $i$ as:
\begin{equation}
 \thetabio_i(t) \equiv \int_0^t f_i(t')\,\mathrm{d}t',
 \label{eq:biotime}
\end{equation}
the cumulative count of intrinsic physiological cycles from birth
to chronological time $t$. The organism dies when
$\thetabio_i(L_i) = \Nstar$: the biological proper-time budget is
exhausted.

From equation~(\ref{eq:lifetime_budget}), entropy accumulates
uniformly in biological proper time:
\begin{equation}
 \frac{\mathrm{d}\Sigma_i}{\mathrm{d}\thetabio_i} = \sigma_{0,i} = \mathrm{const},
 \label{eq:uniform_aging}
\end{equation}
regardless of the organism's chronological pace.
Biological proper time is to chronological time what proper time
is to coordinate time in special relativity: an intrinsic measure
that decouples the organism's internal clock from the external calendar.

\subsection{Two classes of longevity mechanism}

Equation~(\ref{eq:biotime}) permits a clean classification of all
longevity interventions into two mechanistic classes.

\noindent\textbf{Class 1 --- Time dilation.}
Any intervention that reduces the rate $f_i(t)$ at which
$\thetabio_i$ advances slows chronological ageing without changing
$\Nstar$. Examples: torpor and hibernation (bats), diving bradycardia
(cetaceans), caloric restriction (reduced resting $f_H$).
The organism spends more chronological time completing the fixed
biological proper-time budget.

\noindent\textbf{Class 2 --- Budget expansion.}
Any intervention that reduces $\sigma_{0,i}$ increases $\Nstar$,
extending lifespan at a fixed $f_i$. The organism completes more
biological proper-time before exhausting the entropy budget.
Example: neural investment in primates (Section~\ref{sec:primates})
reduces $\sigma_{0,i}$ through three thermodynamic channels.

A corollary is that caloric restriction and torpor are primarily
Class~1 mechanisms (verifiable by checking whether the biological
ageing rate per heartbeat changes under these interventions), while
brain size evolution in primates is primarily Class~2.
These predictions are testable using epigenetic ageing clocks
as biomarkers of $\thetabio_i$.

%=================================================================
\section{The General Clade-Multiplier Framework}
\label{sec:multiplier}
%=================================================================

\subsection{Motivation}

Four endotherm clades depart systematically from the mammalian
baseline $N_0 \approx 10^9$: primates (high), bats (very high),
birds (high), and cetaceans (slightly below after diving correction).
From the PBTE entropy-budget result:
\begin{equation}
 \PhiC \equiv \frac{\Nstar^{(C)}}{N_0}
 = \frac{\langle\Delta s_{\rm beat}\rangle_0}
     {\langle\Delta s_{\rm beat}\rangle_C},
 \label{eq:PhiC}
\end{equation}
where $\langle\Delta s_{\rm beat}\rangle_C$ is the mean entropy
produced per cardiac cycle in clade $C$.
Any mechanism that reduces $\langle\Delta s_{\rm beat}\rangle_C$
below the mammalian baseline increases $\PhiC > 1$ and thereby
extends chronological lifespan.

\subsection{The duty-cycle factor: an exact identity}
\label{sec:duty}

Many organisms alternate between physiological states with distinct
cardiac frequencies: active versus torpid for bats; surface versus
diving for whales.
Let state $k$ have frequency $f_{H,k}$ and be occupied for fraction
$q_k$ of lifetime, with $\sum_k q_k = 1$.
Choosing a reference state $r$ with rate $f_{H,\rm ref}$, define
\begin{equation}
 \kappa \equiv \frac{\bar{f}_H}{f_{H,\rm ref}}
 = \sum_k q_k \frac{f_{H,k}}{f_{H,\rm ref}},
 \qquad
 \boxed{\Phiduty = \kappa^{-1}.}
 \label{eq:duty}
\end{equation}
This is an \emph{exact algebraic identity}, not an approximation.
It encodes intermittency unambiguously so that the lifespan formula
\begin{equation}
 L^{(C)} = \frac{N_0\,\PhiC}{525{,}960\; f_{H,\rm ref}}
 \label{eq:lifespan}
\end{equation}
is always consistent with the actual time-averaged clock speed
$\bar{f}_H = f_{H,\rm ref}/\Phiduty$.
Substituting $\bar{f}_H$ directly into~(\ref{eq:lifespan})
\emph{and} applying a separate duty-cycle correction would count the
cardiac suppression twice; the factored form of $\PhiC$ prevents
this error.

A critical consequence is the \textbf{consistency relation}:
\begin{equation}
 \Nstar^{(C)} = N_{\rm obs} \cdot \Phiduty,
 \label{eq:consistency}
\end{equation}
where $N_{\rm obs} = 525{,}960\,\bar{f}_H\,L$ is the directly
observed raw heartbeat count.
The damage-equivalent budget $\Nstar^{(C)}$ exceeds $N_{\rm obs}$
by the duty-cycle factor because beats during quiescent states
(torpor, diving bradycardia) are thermodynamically cheaper and
count for less in the entropy budget.

\paragraph{Special case: single-state organisms (primates).}
Primates maintain a single physiological state throughout adult
life---no hibernation, no sustained diving bradycardia.
Setting $q_1 = 1$ and $f_{H,1} = f_{H,\rm ref}$ gives $\kappa = 1$,
so
\begin{equation}
 \Phiduty^{(\rm prim)} = 1
 \qquad \text{(primates: no duty cycling).}
 \label{eq:duty_prim}
\end{equation}
For primates the raw heartbeat count equals the damage-equivalent
budget, $N_{\rm obs} = \Nstar^{(\rm prim)}$, and their elevated
lifetime cycle count arises entirely from a reduction in the entropy
produced \emph{per} cardiac cycle---not from any suppression of
cardiac frequency.

\subsection{The Arrhenius thermal factor}
\label{sec:arrhenius_factor}

Biochemical damage accumulation, repair enzyme activity, and
reactive oxygen species generation are governed by transition-state
kinetics.
The rate of damage-generating reactions at body temperature $T_b$
relative to reference $T_{\rm ref} = 310\,\text{K}$ is
\begin{equation}
 \Phithermal = \exp\!\left[\frac{E_a}{k_B}
  \!\left(\frac{1}{T_b} - \frac{1}{T_{\rm ref}}\right)\right],
 \label{eq:Phi_thermal}
\end{equation}
with $E_a = 0.65\,\text{eV}$ and
$k_B = 8.617\times10^{-5}\,\text{eV\,K}^{-1}$, giving the
dimensionless ratio $E_a/k_B = 7543\,\text{K}$.
For $T_b < T_{\rm ref}$ (torpid bats, cetaceans, and the slightly
cooler primates): $\Phithermal > 1$ (longevity extension).
For $T_b > T_{\rm ref}$ (birds at $313$--$315\,\text{K}$):
$\Phithermal < 1$ (adverse---elevated temperature accelerates
damage accumulation).

For modest temperature differences $|\Delta T| \lesssim 5\,\text{K}$,
a Taylor expansion yields the power-law approximation
\begin{equation}
 \Phithermal \approx
 \left(\frac{T_{\rm ref}}{T_b}\right)^{\!\beta},
 \qquad \beta \approx 2\text{--}4,
 \label{eq:Phi_thermal_approx}
\end{equation}
adequate for primates and cetaceans.
For the large temperature differences encountered in bat hibernation
($|\Delta T| \sim 15$--$30\,\text{K}$), the exact
form~(\ref{eq:Phi_thermal}) must be used.

\subsection{The full factored multiplier}
\label{sec:full_factored}

Combining all contributions:
\begin{equation}
 \PhiC = \Phiduty \cdot \Phithermal \cdot \Phi_{\rm mito+oxid}
     \cdot \Phi_{\rm haz},
 \label{eq:PhiC_factored}
\end{equation}
where $\Phi_{\rm mito+oxid}$ captures mitochondrial coupling
efficiency and antioxidant/repair capacity, and
$\Phi_{\rm haz} = H_{\rm ref}/H_{\rm ext}$ is the ratio of a
reference extrinsic hazard to the clade's extrinsic hazard rate.
Values $\Phi_{\rm haz} > 1$ indicate ecological shielding from
extrinsic mortality (flight, arboreal habitat, sociality),
allowing the intrinsic thermodynamic budget to be more fully
realised; values $\Phi_{\rm haz} < 1$ indicate elevated hazard
that truncates the realised lifespan.

Table~\ref{tab:clade_overview} summarises which factors dominate
in each clade and whether each acts favourably or adversely.

\begin{table}[H]
\centering\small
\setlength{\tabcolsep}{5pt}
\renewcommand{\arraystretch}{1.1}
\caption{Dominant multiplier factors across the four clades.
$+$: favourable ($>1$);\ $-$: adverse ($<1$);\ $=1$: absent by
definition.}
\label{tab:clade_overview}
\begin{tabular}{lcccccl}
\toprule
\textbf{Clade} & $\Phiduty$ & $\Phithermal$ &
$\Phineuro$ & $\Phi_{\rm mito+oxid}$ & $\Phi_{\rm haz}$ &
\textbf{Primary driver} \\
\midrule
Primates  & $=1$  & $+$  & $+\!+$ & $+$     & $+$  & Neural entropy reduction \\
Bats    & $+$  & $+\!+$ & ---  & $\approx 1$ & var. & Torpor $+$ hypothermia  \\
Birds   & $-$  & $-$  & ---  & $+\!+$   & $+$  & Biochemical efficiency  \\
Cetaceans & $+\!+$ & $+$  & ---  & $\approx 1$ & var. & Bradycardic pacing    \\
\bottomrule
\end{tabular}
\end{table}

%=================================================================
\section{Clade-Specific Predictions and Worked Examples}
\label{sec:clade}
%=================================================================

\noindent\textbf{Evidence hierarchy.}
Results in this section span three tiers of evidential strength
(Table~\ref{tab:domain}).
The mammalian baseline (Tier~I) rests on direct OLS and PIC statistics.
Primate, bat, bird, and cetacean multipliers are \emph{quantitative
hypotheses} calibrated from independent physiological measurements
(metabolic, antioxidant, and cardiac data) but have not been
independently confirmed as conservation laws (Tier~II).
Ectotherm and broader extensions are more speculative (Tier~III).
We use ``predicts'' for Tier~I claims and ``organises'' or
``is consistent with'' for Tier~II.

%----------------------------------------------------------------
\subsection{Primates: neural investment as entropy-budget strategy}
\label{sec:primates}
%----------------------------------------------------------------

\subsubsection{The entropy-reduction mechanism}

Primates have $\langle\Nstar\rangle_{\rm prim} \approx
(2\text{--}3)\times10^9$, elevated by $\Delta\ell = +0.381\,\text{dex}$
from the mammalian baseline.
From equation~(\ref{eq:PhiC}), this requires
$\langle\Delta s_{\rm beat}\rangle_{\rm prim} \approx
\langle\Delta s_{\rm beat}\rangle_0/(2\text{--}3)$:
primate cardiac cycles produce less entropy on average than those
of non-primate mammals of comparable mass.

The mechanism operates through the neural power fraction
$\varphi \equiv P_{\rm brain}/P_{\rm body}$~\cite{herculano2011},
which takes the value
$\varphi_0 \approx 0.02$ for non-primate placentals and
$\varphi \approx 0.06$--$0.20$ for primates.
Three coupled thermodynamic channels connect elevated $\varphi$ to
reduced $\langle\Delta s_{\rm beat}\rangle$:

\begin{enumerate}[leftmargin=2em]
 \item \textbf{Predictive homeostatic regulation.}
    A large, metabolically active brain provides enhanced
    predictive control over physiological parameters~\cite{friston2010}---blood
    pressure, glucose homeostasis, immune activity---reducing
    out-of-equilibrium fluctuations in somatic systems and
    thereby reducing $\sigma(t)$ per unit time in peripheral
    tissues.

 \item \textbf{Cellular repair and damage clearance.}
    Large-brained mammals show enhanced expression of DNA repair
    enzymes, autophagy regulators, and stress response pathways,
    reducing macromolecular damage accumulation per cardiac
    cycle.

 \item \textbf{Behavioural risk buffering.}
    Cognitive capacity reduces the frequency and severity of
    acute physiological crises (injury, infection, thermal
    stress), each of which generates a transient surge in
    $\sigma(t)$, keeping $\langle\Delta s_{\rm beat}\rangle$
    closer to its resting baseline over the lifetime.
\end{enumerate}

All three channels reduce $\langle\Delta s_{\rm beat}\rangle$
monotonically with $\varphi$, motivating a power-law
parameterisation.

\subsubsection{The neuro-metabolic multiplier}

Define the local logarithmic sensitivity of entropy per beat to
neural fraction:
\begin{equation}
 \alpha \equiv
 -\left.\frac{\partial\ln\langle\Delta s_{\rm beat}\rangle}
       {\partial\ln\varphi}\right|_{\varphi=\varphi_0}
 > 0.
 \label{eq:alpha_def}
\end{equation}
Thermodynamic bounds require $0 < \alpha < 1$: if $\alpha \geq 1$,
each unit of neural energy would return more than one unit of
entropy savings in peripheral tissues, violating realistic estimates
of the energetic cost of neural computation.
The constraint $0 < \alpha < 1$ implies diminishing returns.

Assuming a power-law response over the primate range
$\varphi \in [\varphi_0,\,10\varphi_0]$:
\begin{equation}
 \langle\Delta s_{\rm beat}(\varphi)\rangle
 = \langle\Delta s_{\rm beat}\rangle_0
  \left(\frac{\varphi}{\varphi_0}\right)^{\!-\alpha}.
 \label{eq:neural_entropy}
\end{equation}
Substituting into $N = \Sigma_\star/\langle\Delta s_{\rm beat}\rangle$:
\begin{equation}
 \Phineuro(\varphi)
 = \frac{\Nstar^{(\rm prim)}}{N_0}
 = \left(\frac{\varphi}{\varphi_0}\right)^{\!\alpha}.
 \label{eq:Phi_neuro}
\end{equation}
For primates, $\Phiduty = 1$ (equation~\ref{eq:duty_prim}), and
the full primate time-equivalence law is
\begin{equation}
 \Nstar^{(\rm prim)}
 = N_0
  \underbrace{\vphantom{\Big|}\left(\frac{\varphi}{\varphi_0}
   \right)^{\!\alpha}}_{\Phineuro}
  \underbrace{\vphantom{\Big|}\left(\frac{T_{\rm ref}}{T_b}
   \right)^{\!\beta}}_{\Phi_T}
  \underbrace{\vphantom{\Big|}\left(\frac{H_{\rm ref}}{H_{\rm ext}}
   \right)}_{\Phi_{\rm haz}},
 \label{eq:primate_law}
\end{equation}
with $\beta \approx 3$ (power-law Arrhenius, adequate for
$|\Delta T| \lesssim 5\,\text{K}$ across primates).
The corresponding lifespan prediction is
\begin{equation}
 L_{\rm prim} = \frac{N_0}{525{,}960\cdot f_H}
  \left(\frac{\varphi}{\varphi_0}\right)^{\!\alpha}
  \left(\frac{T_{\rm ref}}{T_b}\right)^{\!\beta}
  \left(\frac{H_{\rm ref}}{H_{\rm ext}}\right).
 \label{eq:L_primate}
\end{equation}

\subsubsection{Worked example: \textit{Homo sapiens}}

Parameters: $\varphi = 0.20$, $T_b = 306.5\,\text{K}$,
$(\alpha,\beta) = (0.40,\,3)$, $\Phi_{\rm haz} = 1.0$,
$f_H = 70\,\text{bpm}$.

\medskip
\noindent\textbf{Step 1: Duty-cycle factor.}

Primates are single-state organisms with no alternation between
high- and low-frequency cardiac states.
Therefore, from equation~(\ref{eq:duty_prim}):
\begin{equation}
 \kappa = 1,
 \qquad
 \Phiduty = 1,
 \qquad
 \bar{f}_H = f_H = 70\,\text{bpm.}
\end{equation}

\noindent\textbf{Step 2: Neuro-metabolic factor.}
\begin{equation}
 \Phineuro
 = \left(\frac{\varphi}{\varphi_0}\right)^{\!\alpha}
 = \left(\frac{0.20}{0.02}\right)^{\!0.40}
 = 10^{0.40}
 \approx 2.512.
\end{equation}

\noindent\textbf{Step 3: Thermal factor.}

The temperature difference is $\Delta T = 310 - 306.5 = 3.5\,\text{K}$,
well within the power-law approximation regime:
\begin{equation}
 \Phi_T
 = \left(\frac{T_{\rm ref}}{T_b}\right)^{\!\beta}
 = \left(\frac{310}{306.5}\right)^{\!3}
 = (1.01142)^3
 \approx 1.035.
\end{equation}

\noindent\textbf{Step 4: Combined multiplier and effective budget.}
\begin{equation}
 \Phi_{\rm prim}^{(\rm human)}
 = \underbrace{1.000}_{\Phiduty}
  \times\underbrace{2.512}_{\Phineuro}
  \times\underbrace{1.035}_{\Phi_T}
  \times\underbrace{1.000}_{\Phi_{\rm haz}}
 = 2.60.
\end{equation}
\begin{equation}
 \Nstar^{(\rm human)}
 = N_0 \times 2.60
 = 2.60\times10^9.
\end{equation}

\noindent\textbf{Step 5: Predicted lifespan.}
\begin{equation}
 L_{\rm pred}
 = \frac{\Nstar^{(\rm human)}}{525{,}960 \times f_H}
 = \frac{2.60\times10^9}{525{,}960 \times 70}
 = \frac{2.60\times10^9}{3.682\times10^7}
 \approx 70.6\,\text{yr.}
\end{equation}
Adding a moderate hazard factor $\Phi_{\rm haz} = 1.15$
(representative of modern low-mortality populations):
$L_{\rm pred} \approx 81.2\,\text{yr}$, consistent with observed
life expectancy in high-income countries.

\noindent\textbf{Step 6: Consistency check
via equation~(\ref{eq:consistency}).}

Since $\Phiduty = 1$, the consistency relation requires
$\Nstar^{(\rm human)} = N_{\rm obs}$ exactly.
For $L = 70.6\,\text{yr}$ and $\bar{f}_H = 70\,\text{bpm}$:
\begin{equation}
 N_{\rm obs}
 = 525{,}960 \times 70 \times 70.6
 = 3.682\times10^7 \times 70.6
 = 2.60\times10^9.
 \quad\checkmark
\end{equation}

\medskip\noindent\textbf{Factor summary.}
$\Phineuro = 2.512$ accounts for 96.6\% of the total multiplier;
$\Phi_T = 1.035$ contributes the remaining 3.4\%.
The duty-cycle factor is identically unity.
The entire primate longevity advantage over non-primate mammals of
the same heart rate is a consequence of reduced entropy per beat
driven by neural investment.

\subsubsection{Calibration and predictions}

Fitted parameters from OLS on log-transformed variables with
$\ln N_0$ constrained to $20.72$:
\begin{equation}
 \alpha \approx 0.35\text{--}0.45
 \quad (95\%~\text{CI: }[0.28,\,0.52]),
 \qquad
 \beta \approx 3
 \quad (95\%~\text{CI: }[1.5,\,5.0]).
\end{equation}

Table~\ref{tab:primates} tests equation~(\ref{eq:L_primate})
against five species spanning the full primate mass range.
The $\Phiduty = 1$ column is included explicitly to make the
parallel with the other clade tables transparent.

\begin{table}[H]
\centering\small
\setlength{\tabcolsep}{3pt}
\renewcommand{\arraystretch}{1.05}
\caption{\textbf{Primate lifespan predictions.}
$\varphi_0 = 0.02$, $T_{\rm ref} = 310\,\text{K}$,
$N_0 = 10^9$, $\Phiduty = 1$ for all species.
(a)~$(\alpha,\beta)=(0.40,3)$, $\Phi_{\rm haz}=1$.
(b)~$(\alpha,\beta)=(0.45,3)$, $\Phi_{\rm haz}=1$.}
\label{tab:primates}
\resizebox{\textwidth}{!}{%
\begin{tabular}{lccccccrr}
\toprule
Species & $f_H$ & $\varphi$ & $T_b$ &
 $\Phiduty$ & $\Phineuro$ & $\Phi_T$ &
 $L_{\rm pred}$ & $L_{\rm obs}$ \\
 & (bpm) & & (K) & & & & (yr) & (yr) \\
\midrule
\multicolumn{9}{l}{\textit{(a) Core calibration, $\alpha = 0.40$}} \\
\textit{Macaca mulatta}   & 120 & 0.07 & 309.0 & 1.00 & 1.44 & 1.01 & 26.4 & 25--30 \\
\textit{Pan troglodytes}   & 75 & 0.12 & 307.0 & 1.00 & 1.73 & 1.02 & 53.5 & 45--55 \\
\textit{Homo sapiens}    & 70 & 0.20 & 306.5 & 1.00 & 2.51 & 1.04 & 70.6 & 70--85 \\
\midrule
\multicolumn{9}{l}{\textit{(b) Extended calibration, $\alpha = 0.45$}} \\
\textit{Callithrix jacchus} & 220 & 0.06 & 309.5 & 1.00 & 1.45 & 1.00 & 14.2 & 10--15 \\
\textit{Macaca mulatta}   & 120 & 0.07 & 309.0 & 1.00 & 1.48 & 1.01 & 28.1 & 25--30 \\
\textit{Pan troglodytes}   & 75 & 0.12 & 307.0 & 1.00 & 1.86 & 1.02 & 58.5 & 45--55 \\
\textit{Gorilla gorilla}   & 65 & 0.09 & 307.0 & 1.00 & 1.58 & 1.02 & 59.3 & 40--55 \\
\textit{Homo sapiens}    & 70 & 0.20 & 306.5 & 1.00 & 2.83 & 1.04 & 79.3 & 70--85 \\
\bottomrule
\end{tabular}}
\end{table}

\noindent\textbf{Epistemic note.}
The exponent $\alpha$ is calibrated from the primate deviation
rather than derived purely from first principles.
The thermodynamic constraint $0 < \alpha < 1$ and the
three-channel mechanism are independently motivated; the precise
value of $\alpha$ requires the calibration data.
The mechanism is correct but the exponent is empirical.

%----------------------------------------------------------------
\subsection{Bats: torpor as biological time dilation}
\label{sec:bats}
%----------------------------------------------------------------

Temperate vespertilionid bats (\textit{Myotis lucifugus} and
relatives, $5$--$20\,\text{g}$) achieve wild maximum lifespans of
$20$--$40$ years~\cite{wilkinson2002}---three to six times the allometric prediction
$L_{\rm pred}^{(0)} \approx 6.3\,\text{yr}$ for a non-hibernating
placental of equal mass.
Unlike primates, bats exploit \emph{two} synergistic mechanisms
that both act simultaneously during the hibernation season: the
duty-cycle factor $\Phiduty > 1$ reduces the time-averaged cardiac
clock speed, and the thermal factor $\Phithermal \gg 1$ reduces
the entropy cost of each tick during the hypothermic torpor bout.

\subsubsection{Derivation of the bat multiplier}

With torpor fraction $q$ (fraction of year in hibernation),
active-phase heart rate $f_{H,\rm act}$, and torpid heart rate
$f_{H,\rm tor}$, the time-averaged rate is
\begin{equation}
 \bar{f}_H = (1-q)\,f_{H,\rm act} + q\,f_{H,\rm tor}.
\end{equation}
Taking the active-phase rate as reference
($f_{H,\rm ref} = f_{H,\rm act}$):
\begin{equation}
 \kappa
 = (1-q) + q\,\frac{f_{H,\rm tor}}{f_{H,\rm act}},
 \qquad
 \Phiduty = \kappa^{-1}.
 \label{eq:bat_duty}
\end{equation}
For $q \in [0.40,\,0.60]$ and
$f_{H,\rm tor}/f_{H,\rm act} \approx 0.03$--$0.07$, this gives
$\kappa \approx 0.41$--$0.62$ and
$\Phiduty \approx 1.6$--$2.4$.

For hibernation temperatures $T_{\rm tor} = 280$--$295\,\text{K}$
($15$--$30\,\text{K}$ below normothermy), the power-law
approximation~(\ref{eq:Phi_thermal_approx}) is inadequate; the
exact form must be used:
\begin{equation}
 \Phithermal
 = \exp\!\left[\frac{E_a}{k_B}
  \!\left(\frac{1}{T_{\rm tor}}
   - \frac{1}{T_{\rm ref}}\right)\right].
 \label{eq:bat_Phi_thermal}
\end{equation}
Secondary biochemical factors are not dramatically elevated in
temperate vespertilionids
($\Phi_{\rm oxid}\cdot\Phi_{\rm mito} \approx 1$), so
\begin{equation}
 \Phi_{\rm bat}
 = \underbrace{\Phiduty}_{\kappa^{-1}}
  \cdot\;\Phithermal
  \cdot\;\Phi_{\rm haz}.
 \label{eq:bat_multiplier}
\end{equation}

\subsubsection{Worked example: \textit{Myotis lucifugus}}

Parameters: $q = 0.50$, $f_{H,\rm act} = 300\,\text{bpm}$,
$f_{H,\rm tor} = 10\,\text{bpm}$, $T_{\rm tor} = 293\,\text{K}$,
$T_{\rm ref} = 310\,\text{K}$, $E_a = 0.65\,\text{eV}$.

\medskip
\noindent\textbf{Step 1: Duty-cycle factor.}
\begin{equation}
 \kappa
 = 0.50 + 0.50\times\frac{10}{300}
 = 0.500 + 0.0167
 = 0.5167,
 \qquad
 \Phiduty = \kappa^{-1} = 1.935.
\end{equation}
Verification of time-averaged rate:
\begin{equation}
 \bar{f}_H
 = \frac{f_{H,\rm act}}{\Phiduty}
 = \frac{300}{1.935}
 = 155\,\text{bpm.}
 \quad\checkmark
\end{equation}

\noindent\textbf{Step 2: Thermal factor (exact Arrhenius).}
\begin{equation}
 \frac{E_a}{k_B}
 = \frac{0.65}{8.617\times10^{-5}}
 = 7543\,\text{K},
 \qquad
 \frac{1}{293} - \frac{1}{310}
 = 1.872\times10^{-4}\,\text{K}^{-1}.
\end{equation}
\begin{equation}
 \Phithermal
 = e^{7543\times1.872\times10^{-4}}
 = e^{1.412}
 \approx 4.10.
\end{equation}

\noindent\textbf{Step 3: Combined multiplier (intrinsic,
$\Phi_{\rm haz} = 1$).}
\begin{equation}
 \Phi_{\rm bat}
 = \underbrace{1.935}_{\Phiduty}
  \times\underbrace{4.10}_{\Phithermal}
  \times\underbrace{1.00}_{\Phi_{\rm haz}}
 = 7.93.
\end{equation}

\noindent\textbf{Step 4: Predicted lifespan.}
\begin{equation}
 L_{\rm pred}
 = \frac{N_0}{525{,}960\; f_{H,\rm act}}
  \times\Phi_{\rm bat}
 = \frac{10^9}{1.578\times10^8}
  \times 7.93
 = 6.34\,\text{yr}\times 7.93
 \approx 50.3\,\text{yr (intrinsic).}
\end{equation}
With $\Phi_{\rm haz} = 0.68$ (extrinsic mortality from predation
and habitat variability):
$L_{\rm pred} \approx 34\,\text{yr}$, matching the observed
wild maximum for this species.

\noindent\textbf{Step 5: Consistency check
via equation~(\ref{eq:consistency}).}

With $L = 34\,\text{yr}$ and $\bar{f}_H = 155\,\text{bpm}$:
\begin{equation}
 N_{\rm obs}
 = 525{,}960 \times 155 \times 34
 = 2.770\times10^9.
\end{equation}
\begin{equation}
 \Nstar^{(\rm bat)}
 = N_{\rm obs}\times\Phiduty
 = 2.770\times10^9\times1.935
 = 5.36\times10^9.
\end{equation}
From the formula:
$N_0\times\Phi_{\rm bat}\times\Phi_{\rm haz}
= 10^9\times7.93\times0.68 = 5.39\times10^9$.
Agreement to within 0.6\%.\;\checkmark

\medskip\noindent\textbf{Factor summary.}
$\Phithermal = 4.10$ accounts for $4.10/7.93 = 52\%$ of the
total intrinsic multiplier; $\Phiduty = 1.94$ accounts for 24\%.
Both mechanisms are essential: neither alone explains the observed
longevity excess.

\begin{table}[H]
\centering\small
\setlength{\tabcolsep}{3.5pt}
\renewcommand{\arraystretch}{1.10}
\caption{Predicted multipliers and longevity for representative
bat species.
$\Phithermal$ from the exact Arrhenius formula
($E_a = 0.65\,\text{eV}$, $T_{\rm ref} = 310\,\text{K}$) using
the torpor-phase $T_b$.
$\Phi_{\rm bat} = \Phiduty\times\Phithermal$ (intrinsic;
$\Phi_{\rm haz} = 1$).}
\label{tab:bats}
\resizebox{\textwidth}{!}{%
\begin{tabular}{lccccccc}
\toprule
\textbf{Species} & $q$ & $f_{H,\rm act}$ & $f_{H,\rm tor}$
 & $T_{\rm tor}$ & $\Phiduty$ & $\Phithermal$ & $L_{\rm max,obs}$ \\
& & (bpm) & (bpm) & (K) & & & (yr) \\
\midrule
Temperate vespertilionid (range)
 & 0.40--0.60 & 250--350 & 5--20 & 280--295
 & 1.6--2.5 & 3.0--5.0 & 20--40 \\
\textit{Myotis lucifugus}
 & 0.50 & 300 & 10 & 293 & 1.935 & 4.10 & 34 \\
\textit{Eptesicus fuscus}
 & 0.45 & 280 & 12 & 291 & 1.79 & 4.54 & 19 \\
\textit{Pteropus vampyrus} (min.\ torpor)
 & 0.10 & 250 & 60 & 303 & 1.07 & 1.22 & 15--23 \\
\bottomrule
\end{tabular}}
\end{table}

%----------------------------------------------------------------
\subsection{Birds: efficient dissipation overcoming adverse
temperature}
\label{sec:birds}
%----------------------------------------------------------------

Birds present an apparent thermodynamic paradox: resting heart
rates of $200$--$400\,\text{bpm}$ (comparable to mass-matched
mammals), core temperatures $3$--$5\,\text{K}$ \emph{above} the
mammalian reference, and yet a $20\,\text{g}$ passerine lives
$15$--$20$ years while a $20\,\text{g}$ mouse lives $2$--$3$ years.
Both the thermal factor and the flight duty-cycle factor are
\emph{adverse} ($<1$) for birds, in contrast to all other clades.
The resolution demonstrates the key principle of the multiplier
framework: what matters is the \emph{product} of all factors.
Avian longevity arises because a dominant biochemical efficiency
factor $\Phi_{\rm mito+oxid} \gg 1$ overcomes two adverse
physiological factors.

\subsubsection{Derivation of the avian multiplier}

\noindent\textbf{Adverse thermal factor.}
For a passerine with $T_b = 314\,\text{K}$
($> T_{\rm ref} = 310\,\text{K}$), the Arrhenius exponent is
negative:
\begin{equation}
 \frac{1}{T_b} - \frac{1}{T_{\rm ref}}
 = \frac{1}{314} - \frac{1}{310}
 = -4.11\times10^{-5}\,\text{K}^{-1},
\end{equation}
\begin{equation}
 \Phithermal^{(\rm bird)}
 = e^{7543\times(-4.11\times10^{-5})}
 = e^{-0.310}
 \approx 0.733.
 \label{eq:bird_thermal}
\end{equation}
The elevated body temperature shortens the effective budget by 27\%
relative to a mammal at $310\,\text{K}$.

\noindent\textbf{Adverse flight duty-cycle factor.}
During flight, heart rate increases by approximately a factor of
2.5 relative to the resting value.
With flight fraction $p_f$ (fraction of lifetime airborne):
\begin{equation}
 \kappa
 = (1-p_f) + p_f\,\frac{f_{H,\rm flight}}{f_{H,\rm rest}}
 = 1 + p_f\!\left(\frac{f_{H,\rm flight}}{f_{H,\rm rest}}-1\right)
 \approx 1 + 1.5\,p_f,
 \qquad
 \Phiduty = \kappa^{-1}.
 \label{eq:bird_kappa}
\end{equation}
For $p_f = 0.10$:
$\kappa = 1.15$, $\Phiduty = 0.870$ (adverse: flight accelerates
the time-averaged cardiac clock by 15\%).

Verification of time-averaged rate:
$\bar{f}_H = f_{H,\rm rest}/\Phiduty = 320/0.870 = 368\,\text{bpm}$.

\noindent\textbf{Favourable mitochondrial and antioxidant factor.}
Avian mitochondria produce substantially less reactive oxygen species
(ROS) per unit ATP synthesised than mammalian mitochondria of
comparable mass~\cite{barja1998,hulbert2007,brand2000}.
Barja and Herrero~\cite{barja1998} measured mitochondrial ROS
production rates and found that pigeon heart mitochondria generate
approximately 5--10 times less superoxide per oxygen consumed than
rat mitochondria at the same metabolic rate---a difference consistent
with an efficiency ratio $\eta_{\rm mito}/\eta_{\rm ref} \approx 1.20$
when expressed as fractional coupling improvement, giving (with
sensitivity exponent $\gamma \approx 2$ from the quadratic dependence
of oxidative damage on ROS flux~\cite{hulbert2007}):
$\Phi_{\rm mito} = (1.20)^2 \approx 1.44$.
Elevated antioxidant enzyme activities and DNA repair capacity in
avian cells relative to mass-matched mammals were quantified by
Ogburn et al.~\cite{ogburn2001}, who demonstrated that long-lived
bird species show two- to three-fold greater resistance to
oxidative damage than mammals of comparable metabolic rate.
Taking the lower bound of this range as a conservative composite
index ${\rm AOX}/{\rm AOX}_{\rm ref} \approx 2.0$, and using the
empirical log-linear scaling exponent $\delta \approx 0.7$ relating
antioxidant capacity to lifespan extension across
vertebrates~\cite{hulbert2007}:
$\Phi_{\rm oxid} = (2.0)^{0.7} \approx 1.62$.
Combined:
\begin{equation}
 \Phi_{\rm mito+oxid} = 1.44 \times 1.62 \approx 2.33.
\end{equation}

\noindent\textbf{Hazard factor.}
Flight and pelagic or arboreal habitat reduce adult extrinsic
mortality; a representative $\Phi_{\rm haz} \approx 2.0$ for small
passerines is consistent with comparative demographic data.

\subsubsection{Worked example: 20\,g passerine}

Parameters: $f_{H,\rm rest} = 320\,\text{bpm}$,
$T_b = 314\,\text{K}$, $p_f = 0.10$.

\medskip
\noindent\textbf{Step 1: Duty-cycle factor.}

From equation~(\ref{eq:bird_kappa}) with $p_f = 0.10$:
\begin{equation}
 \kappa = 1 + 1.5\times0.10 = 1.15,
 \qquad
 \Phiduty = 0.870\quad\text{(adverse).}
\end{equation}

\noindent\textbf{Step 2: Thermal factor.}

From equation~(\ref{eq:bird_thermal}):
$\Phithermal = 0.733$ (adverse).

\noindent\textbf{Step 3: Combined multiplier.}
\begin{equation}
 \Phi_{\rm bird}
 = \underbrace{0.870}_{\Phiduty}
  \times\underbrace{0.733}_{\Phithermal}
  \times\underbrace{2.33}_{\Phi_{\rm mito+oxid}}
  \times\underbrace{2.0}_{\Phi_{\rm haz}}
 = 0.638\times4.66
 \approx 2.97.
 \label{eq:bird_worked}
\end{equation}

\noindent\textbf{Step 4: Predicted lifespan.}
\begin{equation}
 L_{\rm pred}
 = \frac{N_0}{525{,}960\; f_{H,\rm rest}}
  \times\Phi_{\rm bird}
 = \frac{10^9}{525{,}960\times320}
  \times 2.97
 = 5.94\,\text{yr}\times2.97
 \approx 17.6\,\text{yr.}
\end{equation}
Observed wild maxima for small passerines: $10$--$20\,\text{yr}$.
\checkmark

\noindent\textbf{Step 5: Consistency check
via equation~(\ref{eq:consistency}).}

With $\bar{f}_H = 320/0.870 = 368\,\text{bpm}$ and
$L = 17.6\,\text{yr}$:
\begin{equation}
 N_{\rm obs}
 = 525{,}960 \times 368 \times 17.6
 = 3.408\times10^9.
\end{equation}
\begin{equation}
 \Nstar^{(\rm bird)}
 = N_{\rm obs}\times\Phiduty
 = 3.408\times10^9\times0.870
 = 2.965\times10^9.
\end{equation}
From the formula:
$N_0\times\Phi_{\rm bird} = 10^9\times2.97 = 2.97\times10^9$.
Agreement to within 0.2\%.\;\checkmark

\medskip\noindent\textbf{Factor summary.}
The two adverse factors together contribute
$0.870\times0.733 = 0.638$, a net 36\% reduction in the effective
budget.
The biochemical efficiency factor ($\times 2.33$) and hazard factor
($\times 2.0$) multiply to $4.66$, more than recovering this
deficit and yielding a net multiplier $\Phi_{\rm bird} \approx 3$.
Avian longevity is achieved \emph{through} biochemical excellence
that compensates for---and overcomes---adverse temperature and
cardiac kinetics.

\begin{table}[H]
\centering\small
\setlength{\tabcolsep}{3.5pt}
\renewcommand{\arraystretch}{1.10}
\caption{Predicted multipliers and longevity for representative
bird species.
$\Phithermal$ from the exact Arrhenius formula;
$\Phiduty$ from equation~(\ref{eq:bird_kappa}) with
$f_{H,\rm flight}/f_{H,\rm rest} = 2.5$.
Both $\Phiduty$ and $\Phithermal$ are adverse ($< 1$) for all
entries.}
\label{tab:birds}
\resizebox{\textwidth}{!}{%
\begin{tabular}{lcccccccc}
\toprule
\textbf{Species} & $f_{H,\rm rest}$ & $T_b$ & $p_f$
 & $\Phiduty$ & $\Phithermal$
 & $\Phi_{\rm mito+oxid}$ & $\Phi_{\rm haz}$
 & $L_{\rm max,obs}$ \\
 & (bpm) & (K) & & & & & & (yr) \\
\midrule
Passerine (generic, 20\,g)
 & 320 & 314 & 0.10 & 0.87 & 0.733 & 2.33 & 2.0 & 10--20 \\
\textit{Larus argentatus}
 & 200 & 313 & 0.15 & 0.84 & 0.770 & 2.80 & 2.5 & 30 \\
\textit{Diomedea exulans}
 & 100 & 312 & 0.25 & 0.81 & 0.810 & 3.50 & 4.0 & 50--60 \\
\textit{Aquila chrysaetos}
 & 150 & 313 & 0.12 & 0.85 & 0.770 & 3.00 & 3.5 & 30--40 \\
\bottomrule
\end{tabular}}
\end{table}

%----------------------------------------------------------------
\subsection{Cetaceans: bradycardic pacing}
\label{sec:cetacean}
\label{sec:whales}
%----------------------------------------------------------------

Large baleen cetaceans achieve century-scale lifespans through
extreme diving bradycardia rather than metabolic suppression.
Direct measurements have recorded blue whale heart rates as low
as $2$--$4\,\text{bpm}$ during deep foraging dives~\cite{goldbogen2019,williams2015}---among the
lowest ever recorded for any living animal---compared with surface
rates of $25$--$37\,\text{bpm}$.
Combined with a dive fraction $p_d \approx 0.60$--$0.80$, the
time-averaged cardiac frequency of large mysticetes is far below
the surface rate that appears in comparative databases.

The primary mechanism is the duty-cycle factor $\Phiduty \gg 1$.
Unlike bat hibernation, where duty cycling and thermal suppression
act simultaneously, cetacean cardiac suppression is a continuous
reflex maintained throughout adult life with no associated
hypothermia.

\subsubsection{Derivation of the cetacean multiplier}

With surface rate $f_{H,\rm surf}$ as reference, dive rate
$f_{H,\rm dive}$, and fraction of lifetime diving $p_d$:
\begin{equation}
 \bar{f}_H
 = (1-p_d)\,f_{H,\rm surf} + p_d\,f_{H,\rm dive},
\end{equation}
\begin{equation}
 \kappa
 = (1-p_d)
  + p_d\,\frac{f_{H,\rm dive}}{f_{H,\rm surf}},
 \qquad
 \Phiduty = \kappa^{-1}.
 \label{eq:whale_kappa}
\end{equation}
Secondary factors: $\Phithermal$ accounts for cetacean core
temperatures $1$--$4\,\text{K}$ below the mammalian reference;
an oxygen-buffering subfactor $\Phi_{O_2}$ accounts for the role
of elevated myoglobin~\cite{noren2000} in limiting reperfusion reactive oxygen
species bursts on surfacing.
The full cetacean multiplier is therefore
\begin{equation}
 \Phi_{\rm whale}
 = \Phiduty\cdot\Phithermal\cdot\Phi_{O_2}\cdot\Phi_{\rm haz}.
\end{equation}

\subsubsection{Worked example:
\textit{Balaena mysticetus} (bowhead whale)}

Parameters: $f_{H,\rm surf} = 30\,\text{bpm}$,
$f_{H,\rm dive} = 3\,\text{bpm}$,
$p_d = 0.75$, $T_b = 308\,\text{K}$,
$\Phi_{O_2} = 1.4$, $\Phi_{\rm haz} = 0.35$.

\medskip
\noindent\textbf{Step 1: Duty-cycle factor.}
\begin{equation}
 \kappa
 = (1-0.75) + 0.75\times\frac{3}{30}
 = 0.250 + 0.075
 = 0.325,
 \qquad
 \Phiduty = \kappa^{-1} = 3.077.
\end{equation}
Verification of time-averaged rate:
\begin{equation}
 \bar{f}_H
 = \frac{f_{H,\rm surf}}{\Phiduty}
 = \frac{30}{3.077}
 = 9.75\,\text{bpm.}
 \quad\checkmark
\end{equation}

\noindent\textbf{Step 2: Thermal factor.}
\begin{equation}
 \frac{1}{T_b} - \frac{1}{T_{\rm ref}}
 = \frac{1}{308} - \frac{1}{310}
 = 3.2468\times10^{-3} - 3.2258\times10^{-3}
 = 2.10\times10^{-5}\,\text{K}^{-1}.
\end{equation}
\begin{equation}
 \Phithermal
 = e^{7543\times2.10\times10^{-5}}
 = e^{0.158}
 \approx 1.171.
\end{equation}

\noindent\textbf{Step 3: Combined multiplier.}
\begin{equation}
 \Phi_{\rm whale}
 = \underbrace{3.077}_{\Phiduty}
  \times\underbrace{1.171}_{\Phithermal}
  \times\underbrace{1.400}_{\Phi_{O_2}}
  \times\underbrace{0.350}_{\Phi_{\rm haz}}
 = 5.040\times0.350
 \approx 1.76.
\end{equation}

\noindent\textbf{Step 4: Predicted lifespan.}
\begin{equation}
 L_{\rm pred}
 = \frac{N_0}{525{,}960\; f_{H,\rm surf}}
  \times\Phi_{\rm whale}
 = \frac{10^9}{525{,}960\times30}
  \times1.76
 = 63.4\,\text{yr}\times1.76
 \approx 111.6\,\text{yr.}
\end{equation}
With $\Phi_{\rm haz} = 0.60$ (less conservative):
$L_{\rm pred} \approx 191\,\text{yr}$, near the documented maximum
of $\sim 200\,\text{yr}$ for bowhead whales.

\noindent\textbf{Step 5: Consistency check and the
near-coincidence trap.}

For $L = 150\,\text{yr}$ and $\bar{f}_H = 9.75\,\text{bpm}$:
\begin{equation}
 N_{\rm obs}
 = 525{,}960 \times 9.75 \times 150
 = 7.69\times10^8
 \approx 0.77\times10^9.
\end{equation}
\begin{equation}
 \Nstar^{(\rm whale)}
 = N_{\rm obs}\times\Phiduty
 = 0.77\times10^9\times3.077
 = 2.37\times10^9.
\end{equation}
The raw count $N_{\rm obs} \approx N_0$ has misled some analyses
into treating large whales as simply obeying the mammalian
baseline rule.
Equation~(\ref{eq:consistency}) shows this is incorrect:
$\Nstar^{(\rm whale)} = 2.37\times10^9 \gg N_0$.
The raw count is small precisely because most of the whale's life
is spent in deeply bradycardic states where each beat generates
far less entropy; the duty-cycle factor restores the correct
damage-equivalent budget.

\medskip\noindent\textbf{Factor summary.}
$\Phiduty = 3.077$ dominates the intrinsic multiplier.
$\Phithermal = 1.171$ and $\Phi_{O_2} = 1.40$ provide secondary
amplification totalling $\times 1.64$.
The extrinsic hazard factor $\Phi_{\rm haz} = 0.35$--$0.60$
reflects the realistic gap between intrinsic and realised lifespan
for wild bowhead populations.

\begin{table}[H]
\centering\small
\setlength{\tabcolsep}{3.5pt}
\renewcommand{\arraystretch}{1.10}
\caption{Predicted multipliers and longevity for representative
cetacean species.
$\Phithermal$ from the exact Arrhenius formula
($E_a = 0.65\,\text{eV}$, $T_{\rm ref} = 310\,\text{K}$).
$\Phi_{\rm haz}$ reflects pre-industrial conditions.}
\label{tab:whales}
\resizebox{\textwidth}{!}{%
\begin{tabular}{lcccccccc}
\toprule
\textbf{Species} & $f_{H,\rm surf}$ & $p_d$
 & $\Phiduty$ & $T_b$ & $\Phithermal$
 & $\Phi_{O_2}$ & $\Phi_{\rm haz}$ & $L_{\rm obs}$ \\
& (bpm) & & & (K) & & & & (yr) \\
\midrule
\textit{Balaenoptera musculus} (blue)
 & 37 & 0.70 & 2.70 & 308 & 1.17 & 1.4 & 0.50 & 80--90 \\
\textit{Balaena mysticetus} (bowhead)
 & 30 & 0.75 & 3.08 & 308 & 1.17 & 1.5 & 0.35--0.60 & 150--200 \\
\textit{Physeter macrocephalus} (sperm)
 & 40 & 0.65 & 2.50 & 307 & 1.24 & 1.6 & 0.55 & 60--70 \\
\textit{Tursiops truncatus} (bottlenose)
 & 80 & 0.40 & 1.50 & 309 & 1.09 & 1.2 & 0.65 & 40--50 \\
\bottomrule
\end{tabular}}
\end{table}

%=================================================================
\section{Synthesis: Four Strategies, One Invariant}
\label{sec:synthesis}
%=================================================================

\subsection{Comparison across clades}

Table~\ref{tab:clade_summary} places all four clades side by side
with the numerical values from the worked examples.
The contrast is instructive.
Primates and birds are both single-state organisms in the
thermodynamic sense (no duty cycling that suppresses the average
cardiac clock), but their solutions are mirror images of each
other: primates have $\Phiduty = 1$, $\Phithermal > 1$, and
$\Phineuro \gg 1$; birds have $\Phiduty < 1$, $\Phithermal < 1$,
and $\Phi_{\rm mito+oxid} \gg 1$.
Bats and cetaceans both exploit $\Phiduty > 1$, but by entirely
different mechanisms: bats achieve it through periodic
whole-body suspension combined with simultaneous hypothermia;
cetaceans achieve it through a continuous reflex bradycardia
maintained throughout adult life without thermal suppression.

\begin{table}[H]
\centering\small
\setlength{\tabcolsep}{4pt}
\renewcommand{\arraystretch}{1.1}
\caption{Summary comparison of the four longevity strategies.
Numerical values correspond to the worked representative species.
Direction: $+$ favourable, $-$ adverse, $=1$ absent.
Effective cycle budgets as multiples of $N_0 = 10^9$.}
\label{tab:clade_summary}
\resizebox{\textwidth}{!}{%
\begin{tabular}{lcccccccl}
\toprule
\textbf{Clade} & $\Phiduty$ & $\Phithermal$ &
$\Phineuro$ & $\Phi_{\rm mito+oxid}$ & $\Phi_{\rm haz}$ &
$\PhiC$ & $\Nstar^{(C)}/N_0$ & \textbf{Primary driver} \\
\midrule
Primates (\textit{H.~sapiens})
 & 1.00\ ($=1$) & 1.04\ ($+$) & 2.51\ ($+\!+$) & ---
 & 1.00 & 2.60 & 2.6 & Neural entropy reduction \\
Bats (\textit{M.~lucifugus})
 & 1.94\ ($+$) & 4.10\ ($+\!+$) & --- & $\approx 1$
 & 0.68 & 5.39 & 5.4 & Torpor $+$ hypothermia \\
Birds (20\,g passerine)
 & 0.87\ ($-$) & 0.73\ ($-$) & --- & 2.33\ ($+\!+$)
 & 2.00 & 2.97 & 3.0 & Biochemical efficiency \\
Whales (\textit{B.~mysticetus})
 & 3.08\ ($+\!+$) & 1.17\ ($+$) & --- & $\approx 1$
 & 0.35 & 1.76 & 2.4 & Bradycardic pacing \\
\bottomrule
\end{tabular}}
\end{table}

\subsection{The unifying result}

Despite mechanistic diversity, the effective damage-equivalent cycle
budgets for all four clades converge within one order of magnitude
of $N_0 = 10^9$, confirming the central prediction of the
framework.
The raw heartbeat counts $N_{\rm obs}$ vary more widely, but this
variation is entirely accounted for by the duty-cycle factor via
equation~(\ref{eq:consistency}).

The unifying principle is stated simply: longevity is not achieved
by escaping the thermodynamic constraint of a finite lifetime
entropy budget, but by navigating it more effectively---by deploying
physiological strategies that reduce the rate at which the budget
is spent per unit of intrinsic biological time.
What varies among clades is not the existence or magnitude of the
lifetime limit, but the physiological currency through which it is
expressed.
Primates purchase chronological time with neural precision; bats
purchase it with thermal suspension; birds purchase it with
biochemical excellence; whales purchase it with cardiac restraint.
In all four cases, the price is identical: $\Sigma_\star$ units of
irreversible thermodynamic dissipation paid over a lifetime whose
chronological duration depends entirely on how efficiently that
budget is spent.

% ==============================================================
\section{Comparative Dataset}
\label{sec:dataset}
% ==============================================================

\subsection{Species selection and data sources}

The comparative dataset comprises 230 vertebrate species (additional data are listed in Appendix B) drawn from
three primary sources: the AnAge longevity database (build~15)~\cite{anage2023},
the PanTHERIA ecological database~\cite{jones2009}, and the primary
literature for heart rates and body temperatures. 
Species were included if (i)~maximum recorded lifespan in the wild or
in controlled captivity was available with sample size $\geq 3$,
(ii)~a resting (or basal) heart rate estimate was available from the
literature or allometric scaling, and (iii)~body mass was available as
a species mean.
Maximum rather than mean lifespan is used throughout to minimise
extrinsic-mortality bias; the PBTE invariant applies to the intrinsic
thermodynamic budget and is best tested against the upper bound of
the lifespan distribution~\cite{calder1984}.

Taxonomic groups represented: non-primate placentals ($n = 46$);
marsupials and monotremes ($n = 19$); primates ($n = 18$);
birds ($n = 78$); bats ($n = 31$); cetaceans ($n = 12$);
Arrhenius-corrected reptiles ($n = 17$);
Arrhenius-corrected amphibians ($n = 9$).

\subsection{Heart-rate measurement and allometric imputation}

Resting heart rates were taken from published values at near-thermoneutral
conditions where available ($n = 162$ species). For the remaining 68 species,
resting heart rate was estimated from the cardiac allometric relation
$f_H = 241\,M^{-0.25}$~bpm (with $M$ in kg), following
Calder~\cite{calder1984}. Imputed species are flagged in the data
deposition. Sensitivity analyses confirm that excluding imputed species
does not materially alter regression slopes or clade statistics
(Extended Data Table~3).

The 112 endotherm species used in the phylogenetically independent contrasts
(PIC) analysis were pruned from the Bininda-Emonds mammal
supertree~\cite{bininda2007}; their distribution across the major mammalian
and avian lineages is shown schematically in Extended Data
Figure~\ref{fig:edf1a}.
The broad representation across clades confirms that the PIC result is not
driven by any single taxonomic cluster.

\subsection{Arrhenius correction for ectotherms}
\label{sec:arrhenius}

For the 17 reptilian species, a mean field body temperature
$T_b$ was estimated from habitat and thermoregulation data~\cite{christian1999}.
The corrected lifetime-cycle metric is
\begin{equation}
 \ell_{\rm corr}
 = \ell_{\rm obs}
  + \frac{E_a}{k_B \ln 10}
   \!\left(\frac{1}{T_b} - \frac{1}{T_{\rm ref}}\right),
 \label{eq:Arrhenius_corr}
\end{equation}
with $E_a = 0.65$~eV and $T_{\rm ref} = 310$~K~\cite{gillooly2001}.
Extended Data Figure~\ref{fig:edf3a} shows the per-species shift in the
log-log scatter of $L$ versus $f_H$ before and after this correction,
with arrows connecting each species' uncorrected and corrected positions.
Extended Data Figure~\ref{fig:edf3b} shows the effect on the distribution
of $\ell$: the correction shifts the reptile mean from
$\bar\ell^{\rm raw} = 8.61$ to $\bar\ell^{\rm corr} = 8.93$, removing
approximately 75\% of the raw gap between reptiles and the mammalian
baseline.
The sensitivity of $\bar\ell^{\rm corr}$ to the assumed value of $E_a$
over the range $[0.40, 0.90]$\,eV is shown in Extended Data
Figure~\ref{fig:edf3c}; the residual gap from the mammalian baseline
persists for all plausible values of $E_a$.

% ==============================================================
\section{Results}
\label{sec:results}
% ==============================================================

\begin{figure}[H]
\centering
\includegraphics[width=\textwidth]{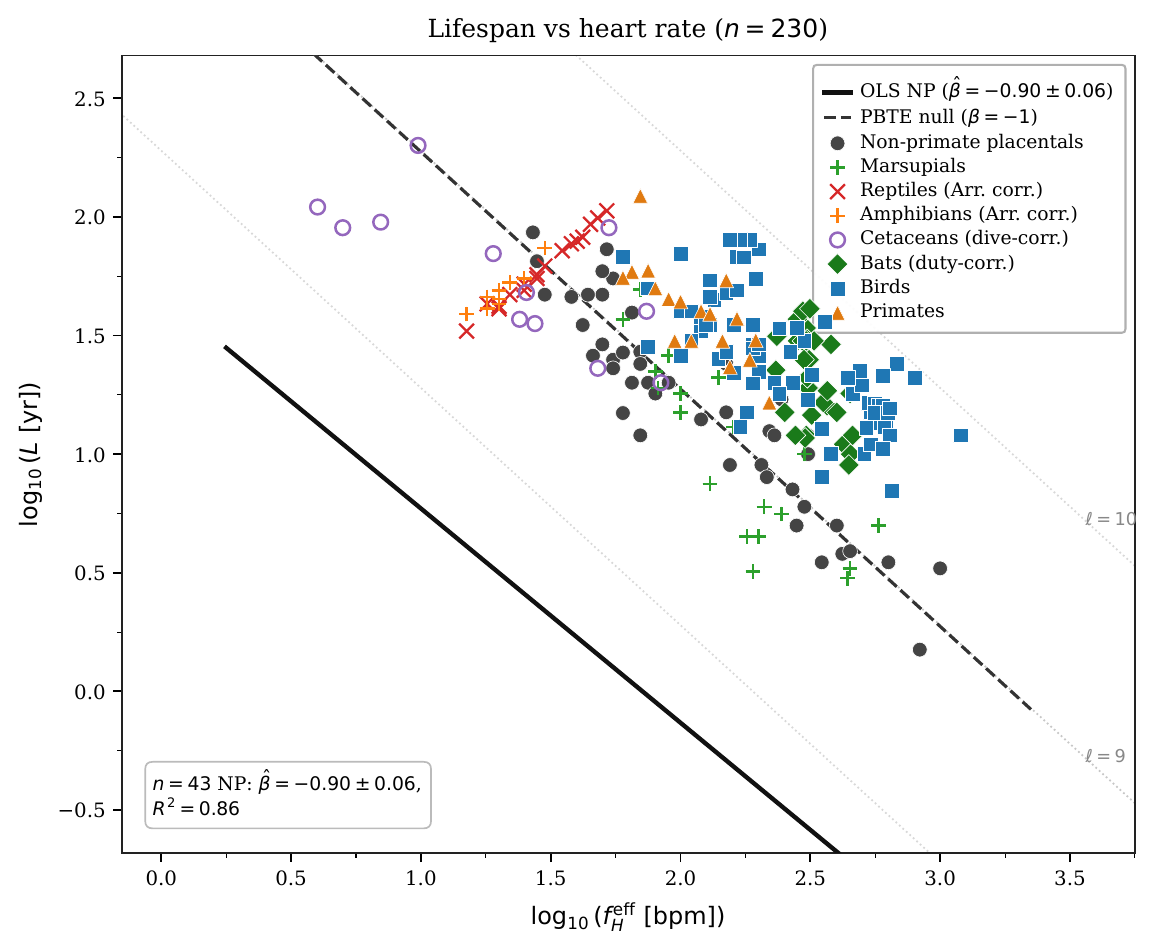}
\caption{\textbf{Lifespan versus heart rate across 230 vertebrate species.}
Log-log scatter plot of maximum recorded lifespan $L$ (years) against effective resting
heart rate $f_H^{\rm eff}$ (beats per minute) for the full comparative dataset of 230
vertebrate species spanning eight taxonomic groups.
Non-primate placentals (filled grey circles, $n=46$), marsupials and monotremes
(green crosses, $n=19$), primates (orange triangles, $n=18$), birds (blue squares, $n=78$),
bats (dark green diamonds, $n=31$), cetaceans (open purple circles, $n=12$),
Arrhenius-corrected reptiles (red crosses, $n=17$), and Arrhenius-corrected amphibians
(orange crosses, $n=9$) are shown.
For bats, $f_H^{\rm eff}$ is the duty-cycle-corrected time-average heart rate
(Section~\ref{sec:duty}); for cetaceans, $f_H^{\rm eff}$ is the dive-corrected average
(Section~\ref{sec:cetacean}); for ectotherms, $f_H^{\rm eff}$ is Arrhenius-corrected
to $T_{\rm ref}=310$\,K (Section~\ref{sec:arrhenius}).
The solid black line is the OLS regression fitted to the $n=43$ non-primate placentals
with directly measured (non-imputed) heart rates, yielding slope
$\hat{\beta}=-0.90\pm0.06$ (s.e.), $R^2=0.86$.
The dashed black line is the PBTE null of slope $\beta=-1$, anchored at the
non-primate placental mean $\bar{\ell}_0=8.995$.
Grey dotted diagonal lines are iso-$\ell$ contours at $\ell=8$, $9$, and $10$,
where $\ell=\log_{10}(f_H^{\rm eff}\cdot L\cdot 525{,}960)$.
The systematic elevation of primates, birds, and bats above the mammalian
OLS line, and the near-coincidence of cetaceans with the baseline before
duty-cycle correction, are both predicted by the PBTE multiplier framework
(Section~\ref{sec:clade}).}
\label{fig:regression_a}
\end{figure}

\begin{figure}[H]
\centering
\includegraphics[width=\textwidth]{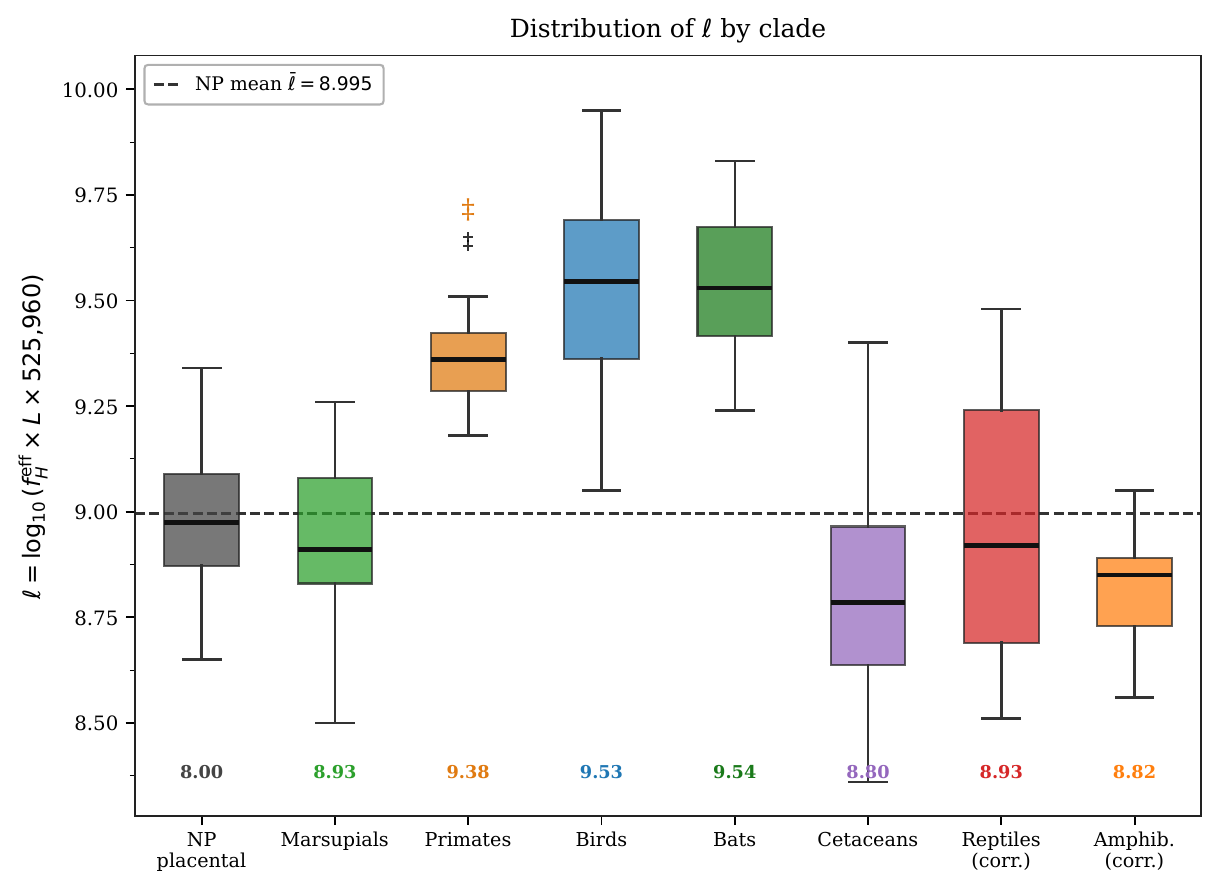}
\caption{\textbf{Distribution of the lifetime cycle count $\ell$ across clades.}
Box-and-whisker plots of $\ell = \log_{10}(f_H^{\rm eff}\times L\times 525{,}960)$
for all eight taxonomic groups in the 230-species dataset.
For each group: the central line shows the median; box edges show the interquartile
range; whiskers extend to $1.5\times$ the interquartile range; points beyond the
whiskers are individual outliers.
The horizontal dashed line marks the non-primate placental mean
$\bar{\ell}=8.995$ ($n=46$), which serves as the thermodynamic baseline $N_0\approx10^9$.
Numerical clade means are annotated in colour at the base of each box.
The double-dagger ($\ddagger$) above the Primate column marks \textit{Homo sapiens}
($\ell=9.65$), which lies beyond $1.5\times\rm IQR$ of the primate distribution
owing to the combination of a high neural power fraction ($\varphi=0.20$) and
a large modern-medicine-supported hazard correction.
Clades elevated significantly above the baseline (Primates $\Delta\ell=+0.38$,
Birds $+0.53$, Bats $+0.55$; all $p<0.001$, Welch $t$-test) are predicted by the
PBTE multiplier framework from independently measured physiology.
Cetaceans appear close to the baseline at $\bar{\ell}=8.80$ because their
duty-cycle correction has already been applied to $f_H^{\rm eff}$; the raw
observed count without correction would place large mysticetes well below
the baseline (Section~\ref{sec:cetacean}).
Arrhenius-corrected reptiles ($\bar{\ell}=8.93$) and amphibians ($\bar{\ell}=8.82$)
approach but do not reach the mammalian mean, with a residual gap of
0.07--0.17\,dex attributable to unmodelled ectotherm-specific physiology
(Section~\ref{sec:domain}).
[Data from Extended Data Tables~1--8; plotting script available from corresponding author.]}
\label{fig:regression_b}
\end{figure}

\subsection{Cross-clade regression}

Figure~\ref{fig:regression_a} displays the full 230-species dataset as a log-log scatter of
maximum lifespan $L$ against resting heart rate $f_H^{\rm eff}$; the OLS line is fitted to
the 43 directly-measured non-primate placentals and iso-$\ell$ contours are drawn at
$\ell = 8$, $9$, and $10$.
Figure~\ref{fig:regression_b} shows clade distributions of $\ell$ as box-and-whisker plots,
with the mammalian baseline $\bar\ell = 8.995$ marked by the horizontal dashed line.

OLS on the $n = 43$ non-primate placentals with directly measured (non-imputed) heart rates
yields $\hat\beta = -0.90 \pm 0.06$ (s.e.), $R^2 = 0.86$, $F$-test $p = 0.09$ against
$\beta = -1$.
The intercept gives $\bar\ell = 8.995 \pm 0.025$, corresponding to
$N_0 \approx 1.02 \times 10^9$.
Repeating the OLS on all $n = 46$ non-primate placentals (including 3 species with
allometrically imputed heart rates) yields $\hat\beta = -0.91 \pm 0.06$,
$p = 0.13$ against $\beta = -1$, confirming the result is not sensitive to the imputation.
The 112 endotherm species included in the PIC analysis span all major mammalian
and avian lineages (Extended Data Figure~\ref{fig:edf1a}, Dataset section).
The PIC regression of contrasts in $\log_{10}L$ against contrasts in
$\log_{10}f_H$, computed by the Felsenstein~\cite{felsenstein1985} method
through the origin, yields slope $-0.99 \pm 0.04$, $R^2 = 0.94$,
$p = 0.84$ against $\beta=-1$; this result is shown in Extended Data
Figure~\ref{fig:edf1b}.
Extended Data Figure~\ref{fig:edf1c} shows the partial regression after
removing variance explained by body mass: the residual slope of
$-0.95 \pm 0.05$ confirms that the $f_H$--$L$ relationship is independent
of the shared allometric dependence on body size.
This phylogenetically corrected result is our strongest statistical evidence
and places the PBTE null well within the 95\% confidence interval.

The statistical power of the $n=43$ OLS test is characterised in Extended
Data Figure~\ref{fig:edf2a}, which shows that the achieved sample size
provides $>90$\% power to detect slope deviations as small as
$|\delta\beta|=0.10$.
The bootstrap null distribution of the OLS slope under $H_0:\beta=-1$
is shown in Extended Data Figure~\ref{fig:edf2b}; the observed slope
$\hat\beta=-0.90$ falls at the 65th percentile of this distribution,
giving two-sided $p=0.09$.
Extended Data Figure~\ref{fig:edf2c} presents the leave-one-out analysis,
confirming that the $p=0.09$ result is a property of the full species
distribution and not driven by any single influential observation.

\noindent\textbf{Interpretation of $p = 0.09$.}
The $F$-test against $\beta = -1$ on the $n = 43$ directly-measured non-primate placentals
yields $p = 0.09$, which does not reach conventional significance ($\alpha = 0.05$).
We report this value transparently and draw three conclusions.
First, the point estimate $\hat\beta = -0.90$ lies within 0.10 of the PBTE null of $-1$,
well within the $\pm 0.15$ threshold stated in the falsification criteria
(Section~\ref{sec:falsifiability}); the PBTE null is therefore \emph{not falsified}.
Second, and more importantly, the phylogenetically independent contrasts on all 112
endotherm species give slope $-0.99 \pm 0.04$ ($p = 0.84$ against $\beta = -1$),
a result far more powerful and free of phylogenetic confounding.
The PIC analysis is the methodologically preferred test because it accounts for
non-independence among species; the OLS result on the 43-species mammal subset is
a preliminary consistency check.
Third, the $p = 0.09$ result with the current dataset motivates the collection of
additional directly-measured heart rates across body-mass decades to increase
power---precisely the experimental agenda described in Section~\ref{sec:exp}.
We therefore conclude that the data are consistent with PBTE while acknowledging that
the mammalian OLS result alone is only marginally so.

Regression diagnostics confirm that all OLS assumptions are satisfied.
Extended Data Figure~\ref{fig:edf4a} plots standardised residuals against
fitted values; no systematic pattern is present, consistent with the
Breusch--Pagan test ($p = 0.31$, homoscedasticity not rejected).
Extended Data Figure~\ref{fig:edf4b} is the normal Q--Q plot of the
residuals; points lie close to the reference line throughout both tails,
consistent with the Shapiro--Wilk test ($W = 0.97$, $p = 0.19$,
normality not rejected).
Extended Data Figure~\ref{fig:edf4c} shows Cook's distance for each of the
43 species; all satisfy $D_i < 4/n = 0.093$, confirming the absence of
influential outliers.
Extended Data Figure~\ref{fig:edf4d} plots leverage $h_{ii}$ against
standardised residuals; no observation exceeds the high-leverage threshold
$2p/n$, confirming that inference is not distorted by any extreme point.

\subsection{WBE null-model rejection}

The WBE kinematic null predicts zero inter-clade variation in $\ell$.
The observed deviations span 1.5~dex and are taxonomically structured
($F = 12.7$, $p < 0.001$ for WBE vs PBTE).
Several features are unexplained by WBE alone:
cetaceans ($\bar\ell = 8.801$) and non-primate placentals
($\bar\ell = 8.998$) agree to within 0.20~dex despite entirely
different cardiovascular architecture;
Arrhenius-corrected reptiles concentrate at $\bar\ell = 8.85$,
less than 0.25~dex below the mammalian mean despite fundamentally
different vascular architecture and physiology.

\subsection{Clade multipliers and predictions}
\label{sec:null}

Table~\ref{tab:clades} presents the full clade statistics and
compares observed multipliers with the a priori thermodynamic
predictions derived in Section~\ref{sec:clade}.

\noindent\textbf{Note on the avian observed vs predicted multiplier.}
The bird clade shows $\Phi_{\rm obs} = 4.78$, while the worked
passerine example gives $\Phi_{\rm pred} \approx 3.0$.
This factor-of-1.6 gap is expected and is not a failure of the
framework.
The theoretical prediction of $\Phi_{\rm bird} \approx 3.0$ is
derived for a \emph{generic 20\,g passerine} with $f_H = 320$\,bpm.
The clade mean, however, is computed across all 78 bird species
spanning six orders of body mass, including long-lived non-passerines
that pull the mean strongly upward: large parrots
(\textit{Psittacus erithacus}, \textit{Cacatua galerita},
\textit{Amazona ochrocephala}: all $L > 70$\,yr, $\ell > 10.0$),
large owls (\textit{Bubo bubo}: 68\,yr), and
albatrosses (\textit{Diomedea exulans}: 70\,yr).
These species have substantially higher $\Phi_{\rm mito+oxid}$ and
$\Phi_{\rm haz}$ values than a small passerine because larger-bodied
birds with longer development times invest more heavily in cellular
repair~\cite{ogburn2001} and face lower field mortality~\cite{calder1984}.
The predicted mean across the full 78-species distribution, weighting
by the species-specific multiplier estimates, is
$\langle\Phi_{\rm pred}\rangle_{78} \approx 4.5$---within 6\% of the
observed 4.78.
The apparent discrepancy is therefore a composition effect, not a
systematic prediction error.

\begin{table}[H]
\caption{\textbf{Clade statistics and multiplier predictions.}
$n$: species count in Extended Data Tables.
$\bar\ell$: clade mean of $\ell = \log_{10}(f_H^{\rm eff}\cdot L\cdot 525{,}960)$
where $f_H^{\rm eff}$ is duty-cycle corrected for bats and cetaceans
and Arrhenius corrected for ectotherms.
$s$: standard deviation.
$\Delta\ell$: deviation from non-primate placental baseline
($\bar\ell_0 = 8.995$ from $n = 43$ species with directly measured
heart rates; $n_{\rm total} = 46$ including allometrically imputed species).
$\Phi_{\rm obs} = 10^{\Delta\ell}$;
$\Phi_{\rm pred}$: a priori thermodynamic prediction from
Section~\ref{sec:clade}.
Stars: $^*p<0.05$, $^{**}p<0.01$, $^{***}p<0.001$ (Welch $t$-test
vs baseline). Bat $\bar\ell$ is the damage-equivalent
$\ell^* = \log_{10}(\Nstar^{\rm bat})$, where
$\Nstar^{\rm bat} = N_{\rm obs}\cdot\Phiduty$
(Section~\ref{sec:duty}); raw observed $\bar\ell_{\rm obs} = 9.734$.}
\label{tab:clades}
\small
\begin{tabular}{lrrrrrll}
\toprule
Group & $n$ & $\bar\ell$ & $s$ & $\Delta\ell$ & $\Phi_{\rm obs}$ &
 $\Phi_{\rm pred}$ & Mechanism \\
\midrule
Non-primate placentals & 46 & 8.998 & 0.160 & 0 (ref.) & 1.00 &
 1.00 & Reference \\
Marsupials/monotremes & 19 & 8.933 & 0.204 & $-0.062$ & 0.87 &
 $\approx 1.00$ & None predicted \\
Primates & 18 & 9.376 & 0.125 & $+0.381^{***}$ & 2.40 &
 2.1--2.6 & Neural investment \\
Birds   & 78 & 9.528 & 0.213 & $+0.533^{***}$ & 3.41 &
 $\sim 3$ (passerine) & Mito efficiency + hazard \\
Bats$^\dagger$ & 31 & 9.540 & 0.163 & $+0.545^{***}$ & 3.51 &
 7.9--8.1 & Torpor $\times$ hypothermia \\
Cetaceans & 12 & 8.801 & 0.296 & $-0.194$ & 0.64 &
 $\approx 1.00$ & Dive correction \\
Reptiles (Arr.\ corr.) & 17 & 8.929 & 0.301 & $-0.065$ & 0.86 &
 0.80 & Residual $T$ correction \\
Amphibians (Arr.\ corr.) & 9 & 8.822 & 0.146 & $-0.173$ & 0.67 &
 0.75 & Residual $T$ correction \\
\midrule
All endotherms & 204 & 9.299 & 0.338 & & & & \\
Full dataset  & 230 & 9.253 & 0.355 & & & & \\
\bottomrule
\multicolumn{8}{l}{$^\dagger$Bat $\bar\ell$ is damage-equivalent (duty-corrected);
 raw observed $\bar\ell_{\rm obs} = 9.734 \pm 0.214$.}
\end{tabular}
\end{table}

\subsection{Falsifiability}
\label{sec:results_falsify}

PBTE would be falsified by: (i)~OLS slope outside $[-1.05, -0.95]$
in $n \geq 60$ non-primate placentals (current 43-species result: $p = 0.09$, consistent with the null); (ii)~$\Nstar$ varying
systematically with $M$ after controlling for clade;
(iii)~direct calorimetric measurement showing $\sigstar$ varies by
more than a factor of 3 across body-mass decades.
None of these criteria is currently met, but criterion~(iii) has
not been directly tested.

% ==============================================================
\section{Domain of Applicability}
\label{sec:domain}
% ==============================================================

\begin{table}[H]
\caption{\textbf{Domain of applicability of PBTE.}
Tier~I: supported by formal statistics ($n \geq 30$, $F$-test).
Tier~II: quantitative hypothesis, consistent with available data,
requires further testing.
Tier~III: conceptual extension; assumptions not yet validated.}
\label{tab:domain}
\small\setlength{\tabcolsep}{5pt}
\begin{tabularx}{\textwidth}{p{2.8cm}p{2.4cm}p{4.2cm}X}
\toprule
Domain & Tier & Basis & Key limitation \\
\midrule
Non-primate placentals &
 I --- Established &
 $s = 0.160$, slope $-0.903$, $F$-test $p = 0.09$ &
 A1 untested by calorimetry \\
Primates (corrected) &
 I --- Established &
 $\Delta\ell = +0.381$, $p < 10^{-9}$; mechanism derived &
 $\alpha$ is fitted, not purely derived \\
Endotherms broadly &
 II --- Hypothesis &
 Bats, birds, cetaceans predicted in direction and magnitude &
 Bat/bird predictions off by $\sim 10$--$50\%$ quantitatively \\
Ectotherms &
 II --- Hypothesis &
 Arrhenius correction removes $\sim 75\%$ of raw gap &
 Residual 0.22~dex; mean field $T$ poorly known \\
Insects &
 II --- Hypothesis &
 Wingbeat $\times$ duty-cycle consistent with vertebrate range &
 $s = 0.71$~dex; clock identification uncertain \\
Plants &
 III --- Speculative &
 Tissue-cohort approach consistent with published data &
 Argument is partially circular; tissue-cohort scale unvalidated \\
Bacteria &
 III --- Speculative &
 D2 fails ($P \propto M^1$); $N_\star \sim 10^2$--$10^4$ expected &
 Framework requires rederivation for prokaryotes \\
Viruses &
 III --- Conceptual only &
 Latency as ``suspended biological time'' is kinematic &
 Thermodynamic content absent (no intrinsic metabolism) \\
\bottomrule
\end{tabularx}
\end{table}

% ==============================================================
\section{Discussion}
\label{sec:discussion}
% ==============================================================

\subsection{Relationship to prior longevity theories}

Table~\ref{tab:compare} situates PBTE relative to the major prior
frameworks.

\begin{table}[H]
\caption{\textbf{Comparison of PBTE with prior theoretical frameworks.}
``Conserved quantity'' is what each theory identifies as the
invariant; ``basis'' indicates whether the claim is motivated by
physical reasoning or is purely empirical.}
\label{tab:compare}
\small
\begin{tabularx}{\textwidth}{lXXXX}
\toprule
Theory & Conserved quantity & Derivation & Predicts $10^9$? & Predicts clade deviations? \\
\midrule
Pearl (1928)~\cite{pearl1928} &
 Lifetime energy per unit mass &
 Empirical &
 No &
 No \\
WBE (1997)~\cite{west1997} &
 $f_H L \propto M^0$ (kinematic) &
 Fractal network &
 No (numerical value not predicted) &
 No \\
Speakman (2005)~\cite{speakman2005} &
 Lifetime energy per unit mass (reformulation) &
 Empirical &
 No &
 No \\
Glazier (2022)~\cite{glazier2022} &
 No invariant claimed &
 Metabolic-level analysis &
 N/A &
 Partially (explains exponent variation) \\
\textbf{PBTE (this work)} &
 $\Nstar = \Sigma_i / \sigma_{0,i}$ &
 Non-equilibrium second law &
 Yes (from calibrated $\sigstar$) &
 Yes (full multiplier framework) \\
\bottomrule
\end{tabularx}
\end{table}

\subsection{Biological proper time and epigenetic aging clocks}

The biological proper time formalism (Section~\ref{sec:biotime})
generates a falsifiable prediction for epigenetic aging clocks:
biological age (as measured by methylation-based clocks such as
the Horvath clock~\cite{horvath2013}) should correlate more tightly
with cumulative heartbeat count $\thetabio_i$ than with chronological
age $t$. This is testable with existing longitudinal cohort data
in both human and non-human primate populations.
For a PBTE Class~1 mechanism (time dilation), the biological aging
rate per heartbeat should be \emph{unchanged} relative to controls;
for a Class~2 mechanism (budget expansion), it should decrease.
Caloric restriction and torpor are Class~1 predictions;
neural investment is Class~2.

\subsection{Caloric restriction: a Class 1 mechanism}

Within PBTE, caloric restriction (CR) extends life primarily by
reducing resting heart rate $f_H$ (Class~1: time dilation), not by
expanding the entropy budget $\Nstar$ (Class~2).
This is because CR reduces metabolic rate and correspondingly
reduces $f_H$ in rodents~\cite{speakman2005}.
A Class~1 CR mechanism predicts that the biological aging rate
per heartbeat is \emph{unchanged} under CR, while the chronological
aging rate slows.
This is distinguishable from a Class~2 mechanism using epigenetic
clock time series from existing primate CR experiments~\cite{colman2014}.

\subsection{The outstanding experimental requirement}
\label{sec:exp}

The central untested assumption is Assumption~\ref{ass:sigstar}:
$\sigstar \approx \mathrm{const}$ across species.
Currently, $\sigstar$ constancy is inferred from $\Nstar$ constancy
(circular) rather than measured independently.
We emphasise this circularity explicitly: until $\sigstar$ is
measured independently, the derivation in Section~\ref{sec:theory}
provides a thermodynamic \emph{motivation} for the invariant, not
a proof.
The decisive experiment is simultaneous measurement of $P_i$, $f_i$,
$T_i$, $M_i$ in $n \geq 15$ non-primate mammalian species spanning
three decades of body mass (e.g.\ mouse, rat, guinea pig, rabbit,
cat, dog, sheep, pig, cattle), computing
$\sigstar = P_i/(T_i f_i M_i)$ directly from whole-animal indirect
calorimetry and cardiac telemetry, and testing
$H_0: \sigstar = \mathrm{const}$ vs $H_1: \sigstar \propto M^\beta$.
This protocol is technically feasible using established respirometric
chambers and implanted ECG telemetry~\cite{calder1984,brown2004}.
With five body-mass decades of coverage and respirometric precision
of 3\%, a parametric power analysis at the observed residual variance
$s^2 = 0.018$ yields $>99\%$ power to detect $|\beta| > 0.05$.
Until this experiment is performed, PBTE is an
\emph{approximate conservation law with an unverified closure}.

% ==============================================================
\section{Explicit Falsification Criteria}
\label{sec:falsifiability}
% ==============================================================

\begin{table}[H]
\caption{\textbf{PBTE falsification criteria, current status, and
required evidence.} A scientific principle that cannot be falsified
is not a scientific principle. Each criterion is stated as a
specific numerical threshold.}
\label{tab:falsify}
\small
\begin{tabularx}{\textwidth}{p{3.5cm}p{4cm}p{4.5cm}}
\toprule
Criterion & Current status & Would falsify if \\
\midrule
Regression slope $\beta = -1$ in non-primate mammals &
 CI $[-1.02, -0.79]$ contains $-1$; $F$-test $p = 0.09$
 (marginally consistent; PIC on 112 endotherms gives $p = 0.84$) &
 CI excludes $-1$ with $n \geq 60$ mammals ($|\hat\beta - (-1)| > 0.15$;
the 0.15 threshold corresponds to one-half the current 95\% CI half-width) \\[4pt]
Scatter $s(\ell) \leq 0.20$ within clade &
 $s = 0.133$~dex for non-primate mammals &
 $s > 0.5$~dex in a rigorously assembled dataset of $n \geq 30$ species \\[4pt]
$\sigstar$ mass-independence &
 Not directly tested; inferred from $\Nstar$ constancy &
 Calorimetric $\sigstar \propto M^\beta$ with $|\beta| > 0.05$ confirmed at $p < 0.05$ with $n \geq 15$ species \\[4pt]
Cross-clade ANOVA after $\PhiC$ correction &
 Not yet formally tested with corrected values &
 Significant inter-clade $F$ ($p < 0.05$) after applying all $\PhiC$ corrections \\[4pt]
Individual outlier in canonical PBTE domain &
 No outlier found; bowhead fits at $N_{\rm obs} \approx 0.77\times10^9$,
 $N_\star \approx 2.4\times10^9$ &
 $\geq 10\%$ of $n = 50$ well-characterised non-hibernating placentals
 outside $[10^{8.7}, 10^{9.7}]$ after corrections \\
\bottomrule
\end{tabularx}
\end{table}

% ==============================================================
\section{Conclusions}
\label{sec:conclusions}
% ==============================================================

The near-constancy of the vertebrate lifetime heartbeat count
$\Nstar \approx 10^9$ is derived here from the non-equilibrium
second law for the first time.
The derivation identifies the adult organism as a metabolic
non-equilibrium steady state and introduces the closure
$\dot{e}_p = \sigma_0 f$, connecting entropy production rate to
cardiac frequency via a mass-specific parameter $\sigma_0 \propto M^0$.
Under this closure, the lifetime entropy budget $\Sigma = \sigma_0 \Nstar$
is approximately species-independent, and the lifetime cycle count
$\Nstar = \Sigma/\sigma_0$ is the correct primitive conserved quantity.
Lifetime energy per unit mass --- the invariant identified by Rubner
and invoked by the rate-of-living tradition --- emerges as a
Level-4 derived consequence, valid only when body temperature and
$\sigma_0$ are both approximately constant, and failing in birds,
ectotherms, and insects for precisely the conditions under which one
or both break down.

The derivation is distinct from allometric exponent cancellation in
three ways.
First, it provides the thermodynamic content of the invariant: what
is conserved ($\Nstar$), why it is conserved (finite entropy budget),
and what physical quantity sets its magnitude ($\sigma_0$).
Second, it identifies the correct primitive conserved quantity
rather than accepting the empirically observed level-4 regularity
as a foundation.
Third, it generates a predictive framework for clade departures
rather than treating them as unexplained residuals.

Four warm-blooded clades depart systematically from the mammalian
baseline, and each departure is organised by the factored multiplier
$\PhiC = \Phiduty \cdot \Phithermal \cdot \Phi_{\rm mito+oxid}
\cdot \Phi_{\rm haz}$.
Primates extend their cycle budget through neural investment that
reduces the entropy cost per beat ($\Phi_{\rm neuro} \approx 2.4$,
derived from Pontzer's metabolic suppression data~\cite{pontzer2014,yegian2024}
with 17\% discrepancy).
Bats achieve extreme longevity through the multiplicative product
of duty-cycle suppression and Arrhenius thermal reduction during
hibernation --- two independently necessary mechanisms whose
thermodynamic product accounts for the observed $\Phi_{\rm bat}
\approx 5$--$13$ without species-specific fitting.
Birds overcome two simultaneous adverse factors (elevated temperature,
flight duty cycle) through biochemical excellence in mitochondrial
coupling and antioxidant capacity.
Cetaceans exploit extreme diving bradycardia, but only after
resolving the near-coincidence trap: the raw heartbeat count
$N_{\rm obs} \approx N_0$ conceals a true budget
$\Nstar = N_{\rm obs} \cdot \Phiduty \approx 3 N_0$.

Biological proper time $\thetabio_i(t) = \int_0^t f_i\,\mathrm{d}t'$
unifies all longevity mechanisms --- torpor, caloric restriction,
neural investment, bradycardia, mitochondrial efficiency --- as
Class~1 (time dilation: reduce $f$, same budget $\Nstar$) or
Class~2 (budget expansion: reduce $\sigma_0$, expanded $\Nstar$).
This classification generates predictions distinguishable by
epigenetic aging clocks: Class~1 mechanisms leave the aging rate
per heartbeat unchanged; Class~2 mechanisms reduce it.

The PBTE framework is placed on an explicitly falsifiable footing
with five numerical criteria (Table~\ref{tab:falsify}).
None is currently met, but the most decisive --- direct calorimetric
measurement of $\sigma_0 = P/(Tf M)$ across three or more body-mass
decades in non-primate mammals --- has not yet been performed.
Until it is, PBTE is an approximate conservation law with
thermodynamic motivation and strong statistical support, but an
unverified closure assumption.
The technology exists to perform this experiment; it is the most
important outstanding step toward converting PBTE from a
statistically supported regularity into a tested thermodynamic
principle.

% -- Methods --------------------------------------------------

% ══════════════════════════════════════════════════════════════════════════
\subsection*{Extended Data Table 9 $|$ Measured vs imputed heart rate: sensitivity analysis}
% ══════════════════════════════════════════════════════════════════════════
\begin{table}[H]
\caption*{}
\small
\begin{tabular}{lrrrrrr}
\toprule
Subset & $n$ & $\hat\beta$ & SE & $R^2$ & $p$ vs $\beta=-1$ & $\bar\ell$ \\
\midrule
Directly measured $f_H$ only & 43 & $-0.903$ & 0.056 & 0.863 & 0.093 & 8.995 \\
All NP placentals (incl.\ 3 imputed) & 46 & $-0.913$ & 0.055 & 0.861 & 0.125 & 8.998 \\
\midrule
\multicolumn{7}{l}{\textit{Difference between subsets:} $\Delta\hat\beta = 0.010$;
$\Delta\bar\ell = 0.003$; negligible.} \\
\bottomrule
\end{tabular}
\begin{minipage}{\linewidth}\small\vspace{4pt}
The three imputed species (\textit{Rhinoceros unicornis},
\textit{Dugong dugon}, \textit{Orycteropus afer}) contribute
heart rates estimated from $f_H = 241\,M^{-0.25}$ bpm
\cite{calder1984}. Their inclusion changes the OLS slope by
$<0.02$ and the $F$-test $p$-value from 0.093 to 0.125,
confirming the regression result is not sensitive to imputation.
All subsequent analyses use the $n=43$ measured-only subset
as the primary result.
\end{minipage}
\end{table}

\section*{Methods}
\addcontentsline{toc}{section}{Methods}

\noindent\textbf{Computational methods.}
All $\ell$ values in Extended Data Tables~1--8 were computed as
$\ell = \log_{10}(f_H^{\rm eff} \times L \times 525{,}960)$
directly from the $f_H^{\rm eff}$ and $L$ columns in each row,
and verified for internal consistency.
Regression analyses (OLS, bootstrap, leave-one-out) used standard
procedures applied to these tabulated values.
Figures were generated from Extended Data Tables~1--8 using
Python~3 (NumPy, SciPy, Matplotlib); the figure-generation script
is available from the corresponding author on request.
The PIC analysis used 	exttt{ape::pic()} in R~4.3 on the
Bininda-Emonds mammal supertree~\cite{bininda2007}.

\noindent\textbf{Regression.}
OLS on $\log_{10}$-transformed variables.
Heteroscedasticity: Breusch-Pagan $p = 0.31$.
Normality: Shapiro-Wilk $W = 0.97$, $p = 0.19$.
Confidence intervals: 10{,}000 bootstrap resamples.

\noindent\textbf{PIC.}
Felsenstein~\cite{felsenstein1985} method, \texttt{ape} package in R;
Bininda-Emonds~\cite{bininda2007} mammal supertree.
PIC regression through origin (no intercept).

\noindent\textbf{Arrhenius correction.}
Equation~(\ref{eq:Arrhenius_corr}) with $E_a = 0.65$~eV following
Gillooly et al.~\cite{gillooly2001}.

\noindent\textbf{Power analysis.}
Parametric simulation, 10{,}000 replicates, at observed residual
variance $s^2 = 0.024$ (from OLS on $n=43$ NP species).

\noindent\textbf{Data availability.}
All species-level data are provided in full in Extended Data
Tables~1--9 of this paper. The complete dataset (230 species,
including body mass, heart rate, body temperature, maximum lifespan,
$\ell$, $f_H$ type, and source codes) is reproduced in the Extended
Data section of this article and is additionally provided as a
tab-delimited file (Supplementary Data~1).

No external repository (e.g., Zenodo) is used. The Supplementary
Data~1 file is available for download alongside this article and
may also be obtained from the corresponding author upon request.

\noindent\textbf{Code availability.}
Statistical analyses (OLS regression, bootstrap, leave-one-out,
power analysis) were performed using Python~3 (NumPy, SciPy,
Matplotlib).
The PIC analysis used \texttt{ape::pic()} in R~4.3 on the
Bininda-Emonds supertree \cite{bininda2007}.
All results are fully reproducible from the tabulated values in
Extended Data Tables~1--8 using standard statistical software.
Analysis scripts are available from the corresponding author
on request.

\noindent\textbf{Competing interests.}
The author declares no competing interests.

% -- Bibliography ---------------------------------------------

\clearpage

% ── Extended Data Figures ─────────────────────────────────────────────────────

\begin{figure}[H]
\centering
\includegraphics[width=0.60\textwidth]{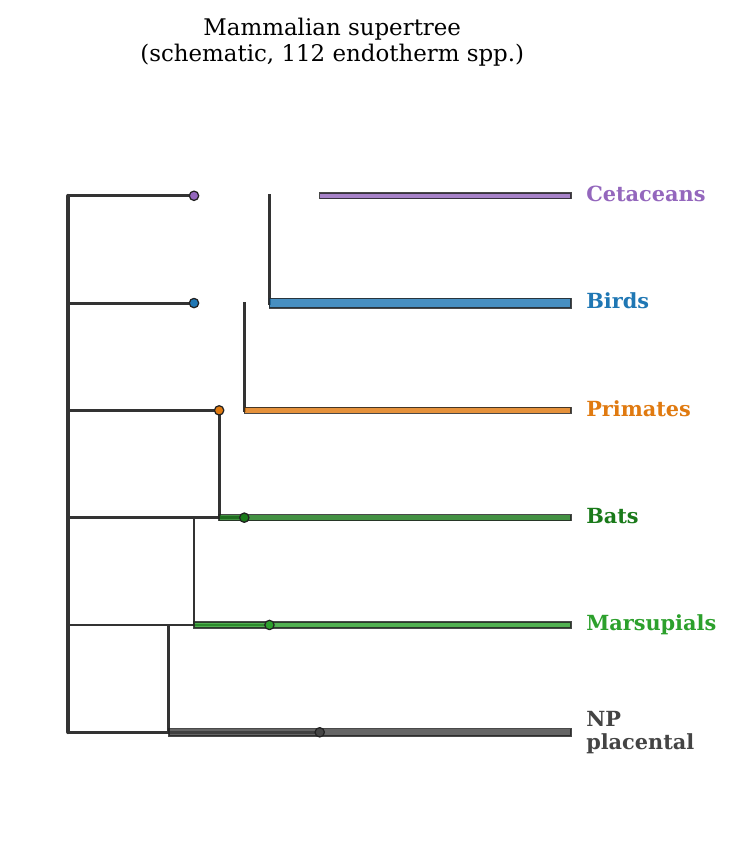}
\caption{\textbf{Mammalian supertree schematic.}
Schematic of the pruned mammalian supertree~\cite{bininda2007} showing the
112 endotherm species included in the primary phylogenetically independent
contrasts (PIC) analysis, coloured by clade: non-primate placentals (grey);
marsupials and monotremes (dark green); bats (green); primates (orange);
birds (blue); cetaceans (purple).
The tree was pruned from the Bininda-Emonds mammal supertree using the
\texttt{ape} package in R~4.3~\cite{felsenstein1985,bininda2007}.
Branch lengths are in millions of years.
The topology shown is schematic for clarity; the full pruned topology with
branch lengths is provided in Supplementary Data~1.
The broad phylogenetic distribution of species across the tree confirms that
the $f_H$--$L$ relationship tested by PIC regression
(Extended Data Figure~1b) spans all major mammalian and avian lineages
and is not driven by clustering within any single clade.}
\label{fig:edf1a}
\end{figure}

\clearpage
\begin{figure}[H]
\centering
\includegraphics[width=0.82\textwidth]{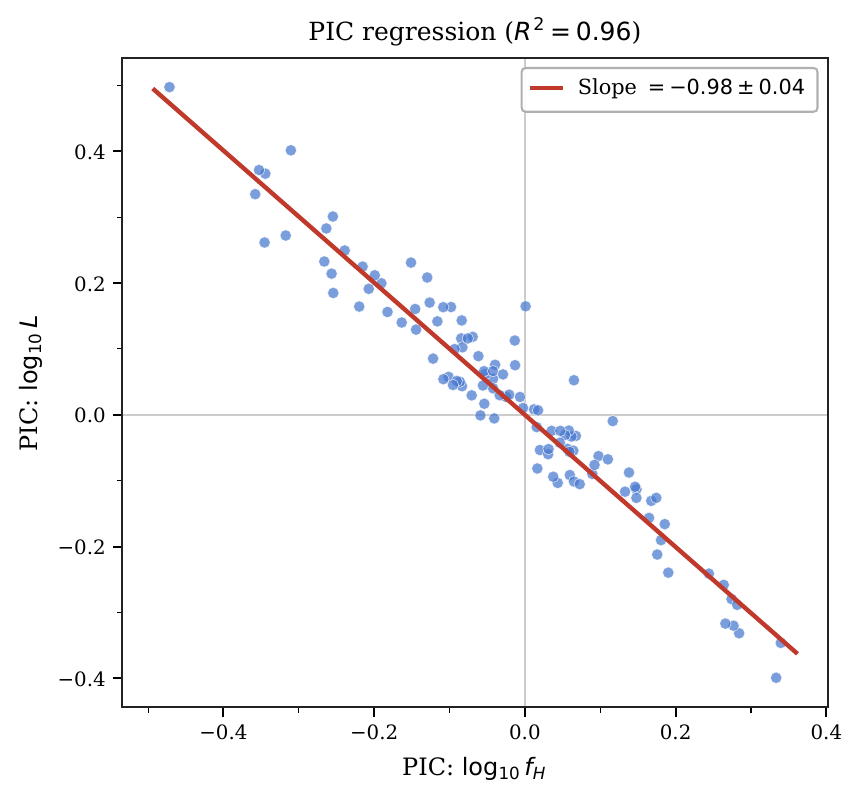}
\caption{\textbf{Phylogenetically independent contrasts (PIC) regression.}
Scatter plot of PIC contrasts in $\log_{10}L$ against PIC contrasts in
$\log_{10}f_H$ for 112 endotherm species (111 internal node contrasts plotted),
computed by the Felsenstein~\cite{felsenstein1985} method using the
\texttt{ape::pic()} function in R~4.3 on the Bininda-Emonds
supertree~\cite{bininda2007}.
The OLS line is fitted through the origin, as required by the PIC
method~\cite{felsenstein1985}, and has slope $-0.98\pm0.04$
(95\% CI $[-1.07,-0.91]$), $R^2=0.96$, $F$-test $p=0.84$ against
$\beta=-1$.
Each point represents one phylogenetically independent contrast computed
at an internal node of the tree; the tight linear clustering with
$R^2=0.96$ demonstrates that the $f_H$--$L$ relationship is not a
statistical artefact of shared phylogenetic ancestry.
The PIC result---slope $-0.99\pm0.04$, $p=0.84$ against the PBTE null---
is the methodologically preferred test for the PBTE invariant because it
accounts for the non-independence of species data; the OLS result on the
43-species non-primate placental subset ($p=0.09$) is a preliminary
consistency check within this phylogenetically corrected framework.
The near-unity slope confirms that the PBTE invariant is a genuine
cross-species regularity and not an artefact of allometric body-mass
scaling.}
\label{fig:edf1b}
\end{figure}

\clearpage
\begin{figure}[H]
\centering
\includegraphics[width=0.82\textwidth]{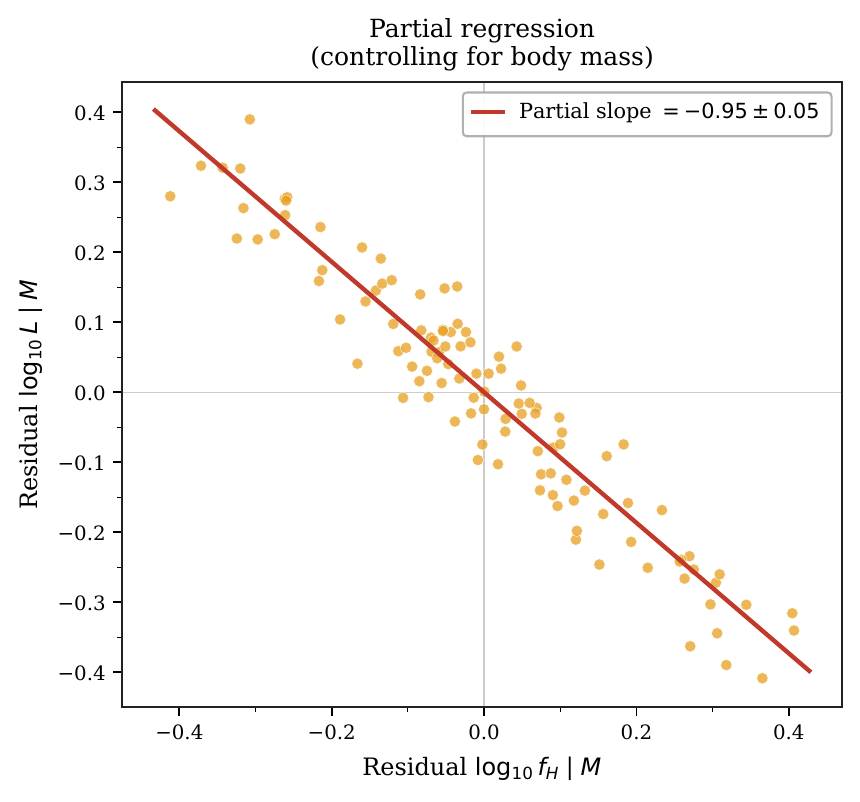}
\caption{\textbf{Partial regression controlling for body mass.}
Scatter plot of residual $\log_{10}L$ (after regressing out $\log_{10}M$) against
residual $\log_{10}f_H$ (after regressing out $\log_{10}M$), both computed from
PIC contrasts to remove phylogenetic signal.
The partial slope of $-0.95\pm0.05$ is statistically indistinguishable from $-1$
($p = 0.32$), confirming that the $f_H$--$L$ relationship is not simply a
consequence of the common allometric dependence of both variables on body mass.
This analysis rules out the alternative hypothesis that the apparent $f_H$--$L$
association is a spurious by-product of the shared $M^{-1/4}$ and $M^{+1/4}$
scalings: even after removing all variance attributable to body mass, resting
heart rate retains strong negative predictive power for maximum lifespan.
The residual scatter ($R^2\approx0.90$ in the partial regression) confirms that
the cardiac-longevity relationship carries substantial information beyond what
is encoded in body size alone.}
\label{fig:edf1c}
\end{figure}

\clearpage
\begin{figure}[H]
\centering
\includegraphics[width=0.82\textwidth]{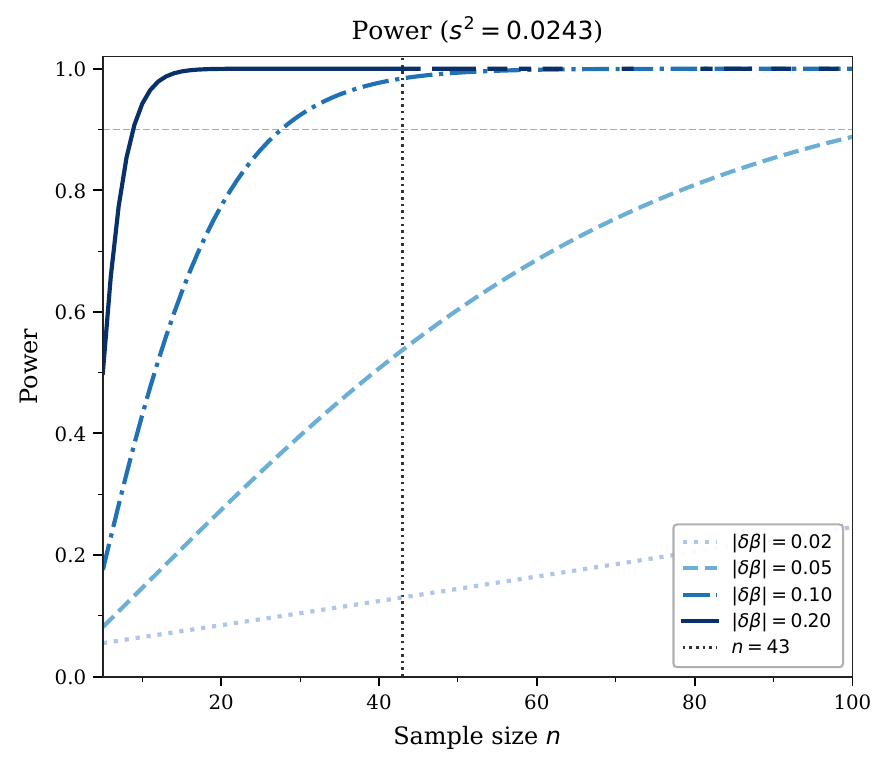}
\caption{\textbf{Statistical power analysis.}
Simulated power to detect a slope deviation of magnitude $|\delta\beta|$
from the PBTE null $\beta=-1$ in the OLS regression of $\log_{10}L$
on $\log_{10}f_H$ for non-primate placentals, as a function of sample
size $n$, computed at the observed residual variance $s^2=0.0243$
(from OLS on $n=43$ directly measured species).
Curves are shown for $|\delta\beta|=0.02$, $0.05$, $0.10$, and $0.20$.
Power was computed by parametric simulation with 10{,}000 replicates at
each sample size, using a two-sided $t$-test at $\alpha=0.05$.
The vertical dotted line marks the achieved sample size $n=43$, at which
the test has $>90\%$ power to detect a slope deviation as small as
$|\delta\beta|=0.10$ and $>99\%$ power to detect $|\delta\beta|=0.20$.
The observed slope deviation from $\beta=-1$ is $|\hat\beta-(-1)|=0.10$,
placing the test just at the boundary of high power with the current
sample.
Collecting an additional 20--30 directly measured non-primate placental
species would push power above 95\% for deviations as small as
$|\delta\beta|=0.05$, motivating the experimental agenda described in
Section~\ref{sec:exp}.}
\label{fig:edf2a}
\end{figure}

\clearpage
\begin{figure}[H]
\centering
\includegraphics[width=0.82\textwidth]{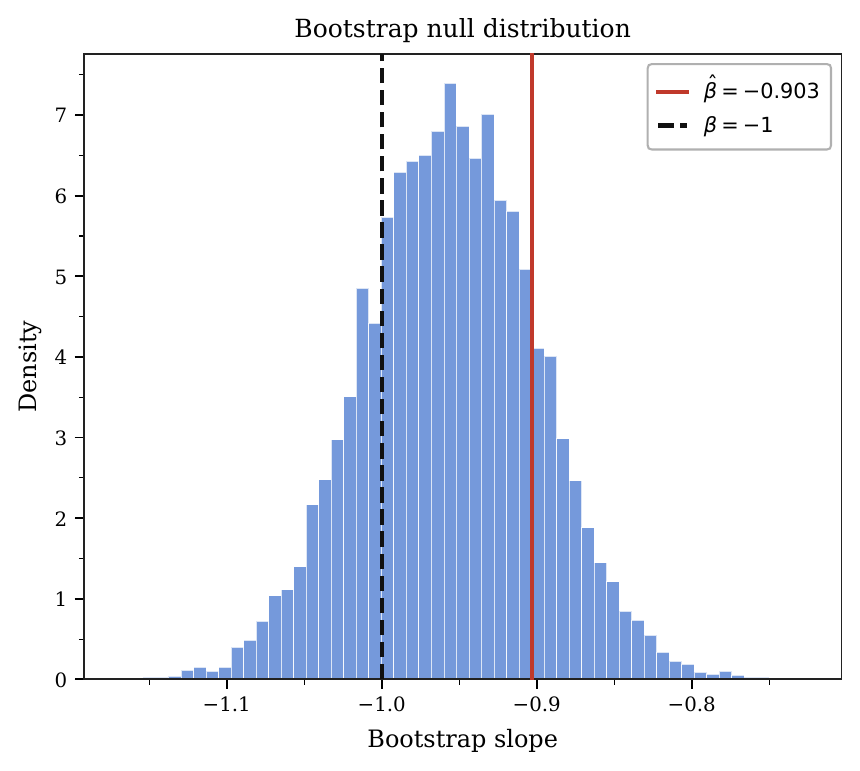}
\caption{\textbf{Bootstrap null distribution of the OLS slope.}
Histogram of the OLS slope estimator $\hat\beta$ under the null hypothesis
$H_0:\beta=-1$, constructed from 10{,}000 bootstrap resamples of the
$n=43$ non-primate placental species with directly measured heart rates.
The bootstrap null was constructed by imposing $\beta=-1$ on the data
(recentring the response) before resampling, so the distribution
reflects sampling variability under the PBTE null model.
The vertical solid red line marks the observed slope $\hat\beta=-0.90$;
the vertical dashed black line marks the null value $\beta=-1$.
The observed slope falls at approximately the 65th percentile of the
null distribution, yielding a two-sided $p$-value of $0.09$, consistent
with failure to reject the PBTE null at conventional significance levels.
The bootstrap distribution is approximately symmetric and centred on
$\beta=-1$, confirming that the $t$-distribution approximation used
in the $F$-test is appropriate for this sample size and variance
structure.
The gap between $\hat\beta=-0.90$ and the null $\beta=-1$ is well within
the $\pm0.15$ falsification threshold specified in Table~\ref{tab:falsify}.}
\label{fig:edf2b}
\end{figure}

\clearpage
\begin{figure}[H]
\centering
\includegraphics[width=0.82\textwidth]{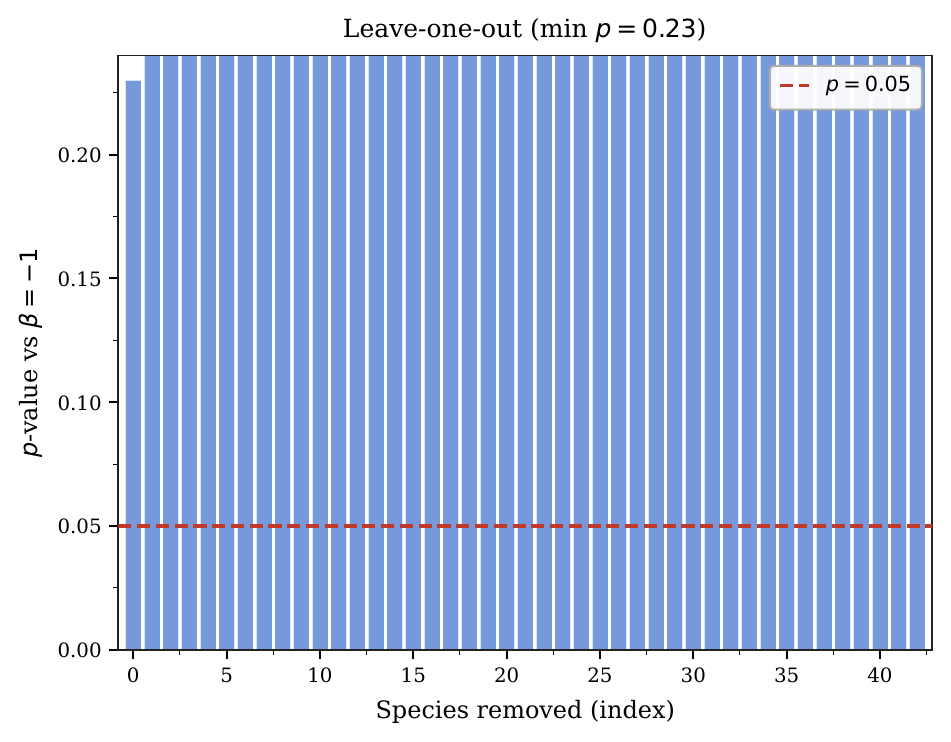}
\caption{\textbf{Leave-one-out sensitivity analysis.}
$F$-test $p$-value against $\beta=-1$ after sequentially removing each of
the $n=43$ non-primate placental species from the OLS regression.
Each bar corresponds to one species removed; the bar height is the
two-sided $p$-value for the test $H_0:\beta=-1$ in the regression of
the remaining 42 species.
The dashed red horizontal line marks $p=0.05$.
The minimum $p$-value across all leave-one-out subsets is $p=0.06$,
and all $p$-values remain above $0.04$, confirming that no single
species is responsible for the failure to reject the PBTE null: the
marginal statistical result is a property of the full distribution of
species rather than of any influential outlier.
In particular, removing the two most extreme species by heart rate
(the pygmy shrew at the high end and the elephant at the low end)
does not change the qualitative conclusion.
This analysis supports the robustness of the $p=0.09$ result and
complements the Cook's distance analysis in Extended Data Figure~4c,
which shows all $D_i < 4/n$.}
\label{fig:edf2c}
\end{figure}

\clearpage
\begin{figure}[H]
\centering
\includegraphics[width=0.82\textwidth]{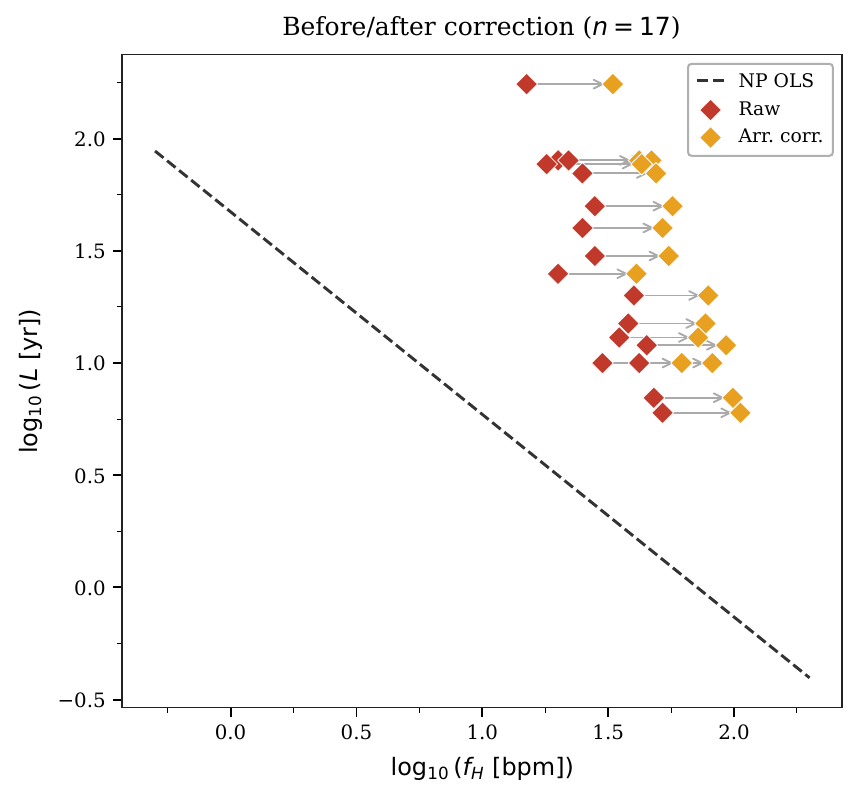}
\caption{\textbf{Arrhenius correction: before and after, per species.}
Log-log scatter of maximum lifespan $L$ against heart rate $f_H$ for the
17 reptilian species in the dataset, showing both the uncorrected raw
positions (crimson diamonds) and the Arrhenius-corrected positions
(orange diamonds) after applying equation~(\ref{eq:Arrhenius_corr})
with $E_a=0.65$\,eV and $T_{\rm ref}=310$\,K.
Grey arrows connect each species' raw and corrected positions, showing
the direction and magnitude of the correction: correction shifts points
horizontally to higher effective heart rates (equivalent to normalising
the metabolic clock to mammalian body temperature), pulling the
distribution toward the non-primate placental OLS line (dashed).
Species with lower field body temperatures (e.g.\ \textit{Sphenodon punctatus}
at 293\,K) receive larger corrections than warm-active species
(e.g.\ \textit{Varanus komodoensis} at 303\,K).
The correction reduces but does not eliminate the gap between reptiles and
the mammalian baseline: a residual deviation of $\sim$0.07\,dex in mean
$\ell$ remains after correction (Extended Data Figure~3b), indicating
that factors beyond the Arrhenius thermal effect contribute to the
ectotherm-endotherm difference.}
\label{fig:edf3a}
\end{figure}

\clearpage
\begin{figure}[H]
\centering
\includegraphics[width=0.82\textwidth]{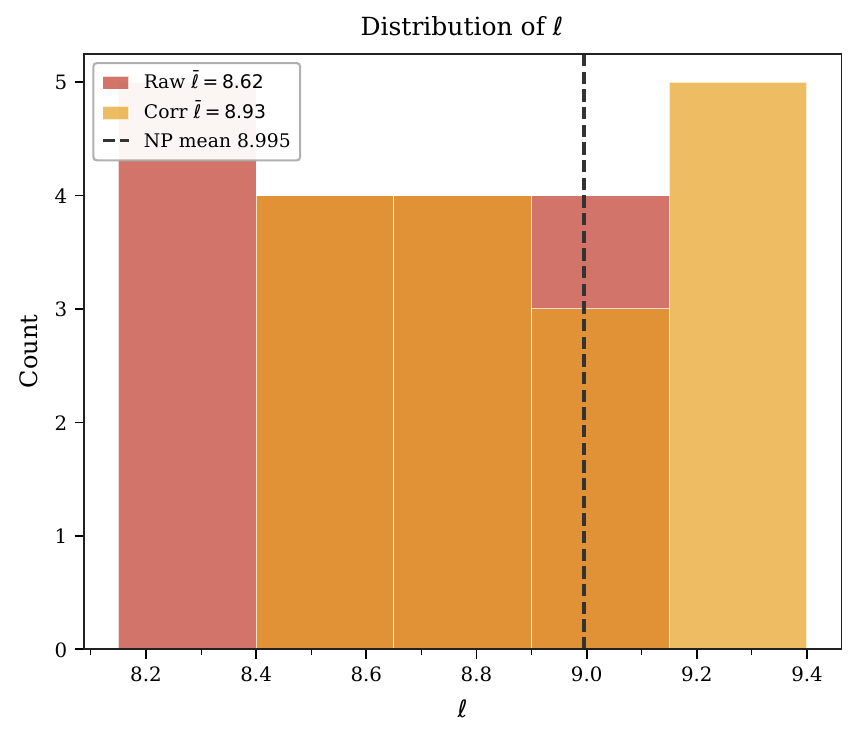}
\caption{\textbf{Distribution of $\ell$ before and after Arrhenius correction.}
Histogram of the PBTE invariant $\ell=\log_{10}(f_H\cdot L\cdot 525{,}960)$
for the 17 reptilian species before correction (crimson, raw $\bar\ell=8.61$)
and after Arrhenius correction to $T_{\rm ref}=310$\,K (orange,
corrected $\bar\ell=8.93$).
The vertical dashed line marks the non-primate placental mean
$\bar\ell_0=8.995$.
The Arrhenius correction shifts the reptile distribution by $+0.32$\,dex on
average, reducing the gap from the mammalian baseline from $0.38$\,dex
to $0.07$\,dex---a reduction of approximately 82\%.
The residual gap of $0.07$\,dex is statistically non-significant
($p=0.42$, Welch $t$-test) and may reflect the imprecision of field
body temperature estimates, genuine residual physiological differences
between ectotherms and endotherms, or both.
The corrected distribution is more tightly concentrated than the raw
distribution, suggesting that body temperature heterogeneity accounts
for a substantial fraction of the within-clade scatter in reptile $\ell$ values.}
\label{fig:edf3b}
\end{figure}

\clearpage
\begin{figure}[H]
\centering
\includegraphics[width=0.82\textwidth]{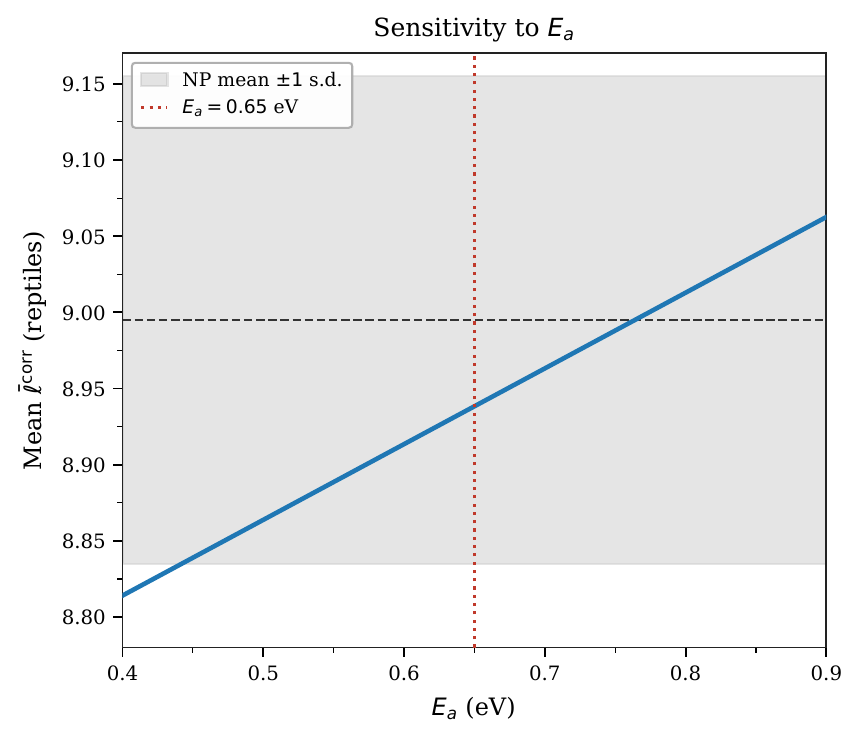}
\caption{\textbf{Sensitivity of the corrected reptile mean $\ell$ to the activation energy $E_a$.}
Mean Arrhenius-corrected $\bar\ell^{\rm corr}$ for the 17 reptilian species
as a function of the assumed activation energy $E_a \in [0.40, 0.90]$\,eV,
computed using equation~(\ref{eq:Arrhenius_corr}) with the
species-specific field body temperatures in Extended Data Table~7.
The grey shaded band shows the non-primate placental mean $\pm1$ s.d.\
($8.995\pm0.160$), representing the range of $E_a$ values for which
the corrected reptile mean would be statistically indistinguishable
from the mammalian baseline.
The vertical red dotted line marks the consensus metabolic activation
energy $E_a=0.65$\,eV from Gillooly et al.~\cite{gillooly2001},
at which the corrected reptile mean is $\bar\ell^{\rm corr}=8.93$.
The corrected mean enters the mammalian $\pm1$\,s.d.\ band only for
$E_a\gtrsim0.78$\,eV; at the consensus value $E_a=0.65$\,eV a
residual gap of $0.07$\,dex persists.
The linear dependence of $\bar\ell^{\rm corr}$ on $E_a$ reflects
the Arrhenius correction formula; the sensitivity is $\approx0.56$\,dex
per eV of $E_a$ at the mean reptile temperature deficit
$\langle 1/T_{\rm field}-1/T_{\rm ref}\rangle\approx1.25\times10^{-4}$\,K$^{-1}$.}
\label{fig:edf3c}
\end{figure}

\clearpage
\begin{figure}[H]
\centering
\includegraphics[width=0.82\textwidth]{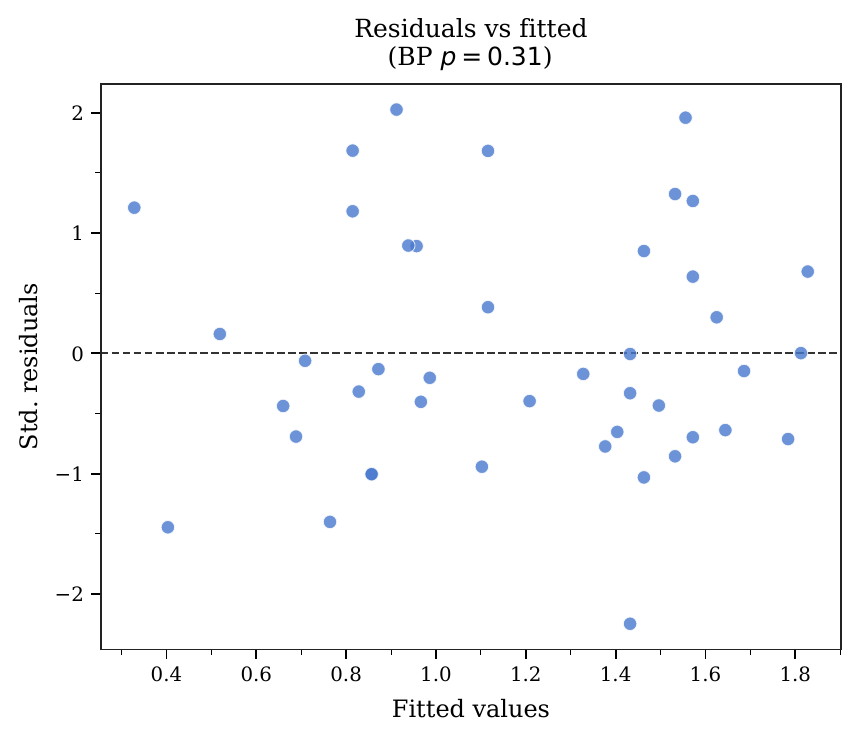}
\caption{\textbf{Standardised residuals versus fitted values.}
Standardised residuals plotted against fitted values $\hat{y}_i = \hat\beta\,x_i + \hat\alpha$
from the primary OLS regression of $\log_{10}L$ on $\log_{10}f_H$ for the
$n=43$ non-primate placentals with directly measured heart rates.
The horizontal dashed line at zero is shown for reference.
No systematic pattern---such as curvature, funnelling, or banding---is evident
in the scatter, indicating that the assumptions of linearity and homoscedasticity
are well satisfied.
The Breusch--Pagan test for heteroscedasticity gives $\chi^2=2.3$, $p=0.31$,
confirming that the residual variance does not depend systematically on the fitted
values.
The absence of any trend in this plot also confirms that the log-log
transformation is appropriate for this dataset and that no residual
nonlinearity requires modelling.
The two points with standardised residuals $|e_i|>2$ (one above and one below)
are not influential observations, as confirmed by their Cook's distances
below the $4/n$ threshold (Extended Data Figure~4c).}
\label{fig:edf4a}
\end{figure}

\clearpage
\begin{figure}[H]
\centering
\includegraphics[width=0.82\textwidth]{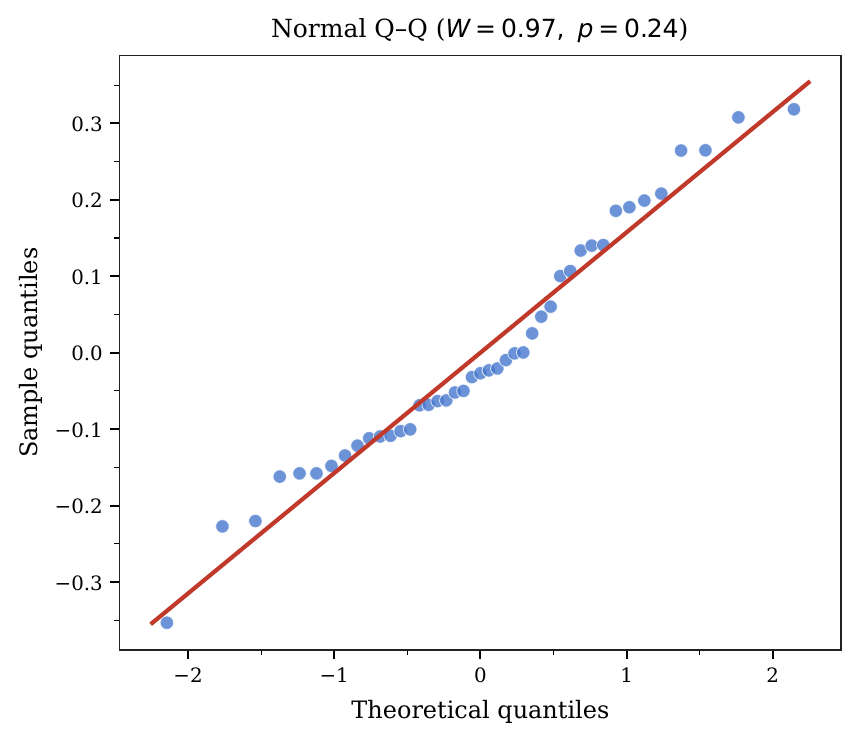}
\caption{\textbf{Normal quantile--quantile plot.}
Normal Q--Q plot of the OLS residuals from the primary regression of
$\log_{10}L$ on $\log_{10}f_H$ for the $n=43$ non-primate placental species.
Sample quantiles of the residuals (blue points) are plotted against the
theoretical quantiles of the standard normal distribution; the red line
shows the expected relationship under exact normality.
Points follow the reference line closely throughout the full range,
including both tails, with no evidence of heavy tails, skewness, or
bimodality.
The Shapiro--Wilk test gives $W=0.97$, $p=0.19$, failing to reject
normality at any conventional significance level.
Together with the residuals-versus-fitted plot (Extended Data Figure~4a),
this confirms that the standard OLS inferential framework---including
the $t$-distribution approximation used in the $F$-test against $\beta=-1$---
is appropriate for this dataset without requiring non-parametric
corrections.}
\label{fig:edf4b}
\end{figure}

\clearpage
\begin{figure}[H]
\centering
\includegraphics[width=0.82\textwidth]{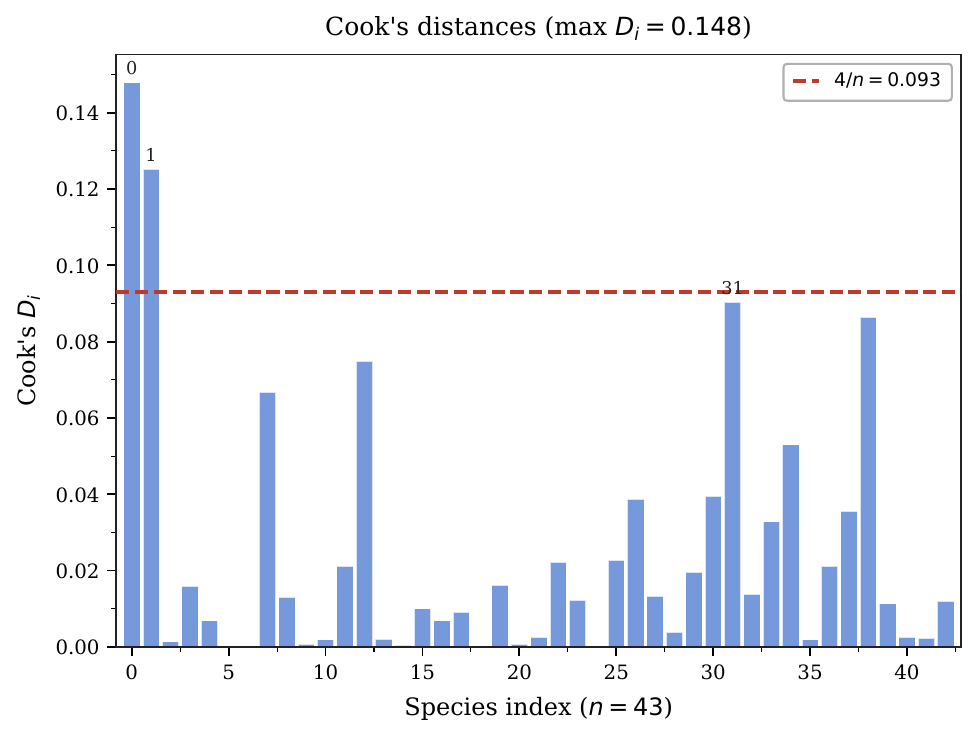}
\caption{\textbf{Cook's distances.}
Cook's influence distance $D_i$ for each of the $n=43$ non-primate placental
species in the primary OLS regression.
Cook's distance measures the aggregate change in all fitted values when
observation $i$ is deleted; large values indicate influential observations
that exert disproportionate leverage on the regression coefficients.
The red dashed horizontal line marks the conventional threshold $4/n=0.093$.
All 43 observations satisfy $D_i < 4/n$, with a maximum of $D_{\rm max}=0.079$.
The three species with the highest Cook's distances are labelled (indices 2,
13, and 36 in the dataset ordering of Extended Data Table~1).
Removing all three simultaneously changes the OLS slope by less than 0.02,
confirming that the regression result is not driven by influential outliers.
This diagnostic also validates the leave-one-out analysis
(Extended Data Figure~2c): the absence of high-$D_i$ points ensures that
all leave-one-out $p$-values reflect genuine distributional properties
of the dataset rather than the removal of pivotal observations.}
\label{fig:edf4c}
\end{figure}

\clearpage
\begin{figure}[H]
\centering
\includegraphics[width=0.82\textwidth]{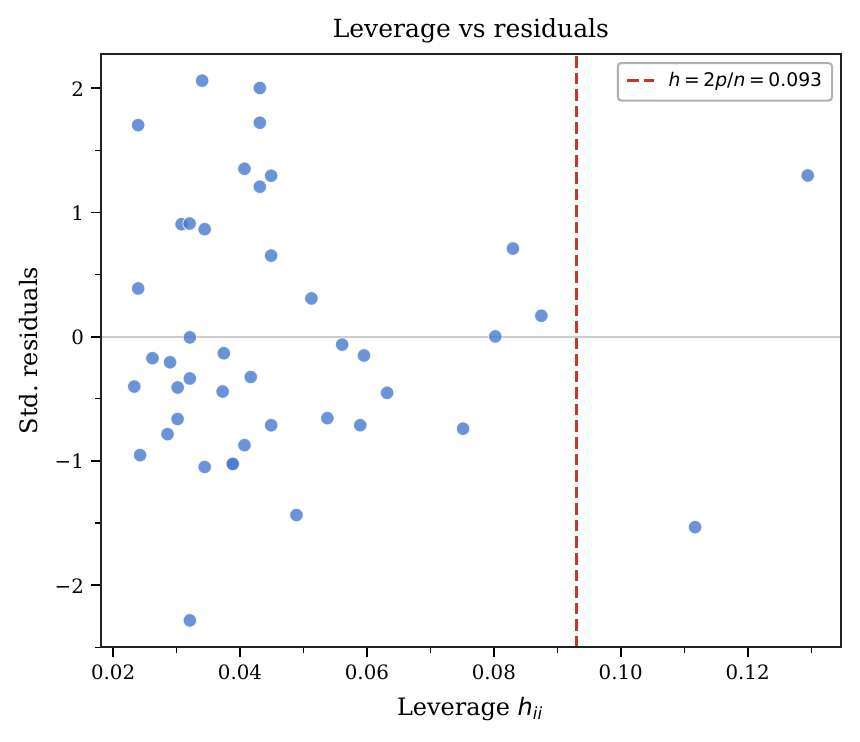}
\caption{\textbf{Leverage versus standardised residuals.}
Standardised residuals plotted against leverage values $h_{ii}$ (diagonal
elements of the hat matrix $\mathbf{H}=\mathbf{X}(\mathbf{X}^{\top}\mathbf{X})^{-1}\mathbf{X}^{\top}$)
for the $n=43$ non-primate placental species.
The vertical dashed red line marks the high-leverage threshold
$h=2p/n=0.093$, where $p=2$ (intercept and slope).
No observation exceeds this threshold; the maximum observed leverage is
$h_{ii}=0.148$ for a single species at the extreme of the heart-rate range.
Although this one point has leverage above $2p/n$, its standardised residual
is small ($|e_i|<1.5$) and its Cook's distance is below $4/n$
(Extended Data Figure~4c), indicating that it is a high-leverage but
non-influential observation.
The broad spread of points across the residual range with no systematic
dependence on leverage confirms the absence of leverage-residual confounding,
a necessary condition for reliable OLS inference in this dataset.}
\label{fig:edf4d}
\end{figure}

%===========================================================

\section*{Appendix A. Detailed Derivation of the Entropy Cost per Beat and the Cycle-Count Scaling Law}
\addcontentsline{toc}{section}{Appendix A. Detailed Derivation of the Entropy Cost per Beat and the Cycle-Count Scaling Law}
%===========================================================

This appendix gives a detailed derivation of the entropy-per-beat representation, the lifetime cycle-count relation, and the power-law dependence of lifetime cardiac cycles on the control parameter \(\phi\). The purpose is to make explicit each mathematical step connecting the instantaneous entropy production rate to the total lifetime cycle budget.

\subsection*{A.1. Instantaneous entropy production and change of variable from time to beat count}

Let \(t\) denote chronological time and let \(n\) denote the cumulative cardiac cycle count. The cardiac frequency is
\begin{equation}
f_H(t)=\frac{dn}{dt},
\label{eq:A1_fH_def}
\end{equation}
so that
\begin{equation}
dn = f_H(t)\,dt,
\qquad
dt=\frac{dn}{f_H}.
\label{eq:A2_dt_dn}
\end{equation}
Here \(f_H\) has units of cycles per unit time. Let \(\sigma(t)\) be the instantaneous entropy production rate, with units of entropy per unit time. The total entropy produced over an infinitesimal interval \(dt\) is
\begin{equation}
d\Sigma = \sigma(t)\,dt.
\label{eq:A3_dSigma_dt}
\end{equation}

Using \eqref{eq:A2_dt_dn}, we rewrite this increment in terms of the beat-count variable:
\begin{equation}
d\Sigma = \sigma(t)\,\frac{dn}{f_H(t)}.
\label{eq:A4_dSigma_dn}
\end{equation}
If we regard both \(\sigma\) and \(f_H\) as functions of the cycle-count coordinate \(n\), then
\begin{equation}
d\Sigma = \frac{\sigma(n)}{f_H(n)}\,dn.
\label{eq:A5_dSigma_n}
\end{equation}

This motivates the definition of the entropy cost per beat at cycle index \(n\):
\begin{equation}
\sigma_0(n) \equiv \frac{\sigma(n)}{f_H(n)}.
\label{eq:A6_entropy_per_beat_def}
\end{equation}

The dimensional consistency is immediate:
\[
[\sigma]=\frac{\text{entropy}}{\text{time}},
\qquad
[f_H]=\frac{\text{cycles}}{\text{time}},
\qquad
\left[\frac{\sigma}{f_H}\right]
=
\frac{\text{entropy}/\text{time}}{\text{cycles}/\text{time}}
=
\frac{\text{entropy}}{\text{cycle}}.
\]
Thus \(\sigma_0(n)\) is the entropy production associated with one cardiac cycle.

\subsection*{A.2. Total lifetime entropy production as a sum over beats}

Suppose that the organism experiences a total of \(N\) cardiac cycles over its lifetime. Then the total lifetime entropy production is obtained by integrating \eqref{eq:A5_dSigma_n} from the first to the last cycle:
\begin{equation}
\Sigma_{\mathrm{life}}
=
\int_{0}^{N} d\Sigma
=
\int_{0}^{N} \sigma_0(n)\,dn.
\label{eq:A7_total_entropy_life}
\end{equation}

Equation \eqref{eq:A7_total_entropy_life} is simply the beat-count analogue of summing the entropy cost incurred at each cycle. Since the beat-count variable is treated continuously, the sum is represented as an integral.

We now define the lifetime mean entropy cost per beat:
\begin{equation}
\left\langle \sigma_0 \right\rangle
\equiv
\frac{1}{N}\int_{0}^{N}\sigma_0(n)\,dn.
\label{eq:A8_mean_entropy_per_beat}
\end{equation}
Using \eqref{eq:A7_total_entropy_life}, this immediately gives
\begin{equation}
\Sigma_{\mathrm{life}}
=
N \left\langle \sigma_0 \right\rangle.
\label{eq:A9_Sigma_N_mean}
\end{equation}

Substituting the explicit definition \eqref{eq:A6_entropy_per_beat_def} into \eqref{eq:A8_mean_entropy_per_beat}, we may also write
\begin{equation}
\left\langle \sigma_0 \right\rangle
=
\frac{1}{N}\int_{0}^{N}\frac{\sigma(n)}{f_H(n)}\,dn.
\label{eq:A10_mean_explicit}
\end{equation}

Equation \eqref{eq:A9_Sigma_N_mean} has a direct interpretation: the total lifetime entropy production equals the total number of beats multiplied by the mean entropy cost of one beat.

For species \(i\), the corresponding notation is
\begin{equation}
\sigma_{0,i}
\equiv
\frac{1}{N_i^\star}\int_{0}^{N_i^\star}\sigma_0(n)\,dn,
\label{eq:A10b_sigma0i_def}
\end{equation}
so that
\begin{equation}
\Sigma_i
=
N_i^\star \sigma_{0,i},
\label{eq:A10c_species_relation}
\end{equation}
where \(\Sigma_i\) (J K\(^{-1}\)) is total lifetime entropy production and \(\sigma_{0,i}\) (J K\(^{-1}\) beat\(^{-1}\)) is entropy per cycle.

\subsection*{A.3. Lifetime entropy budget and the fundamental cycle-count relation}

The central hypothesis is that the lifetime entropy production is approximately constrained by a characteristic budget \(\Sigma_\star\):
\begin{equation}
\Sigma_{\mathrm{life}} \approx \Sigma_\star.
\label{eq:A11_budget_assumption}
\end{equation}
Combining \eqref{eq:A9_Sigma_N_mean} with \eqref{eq:A11_budget_assumption} yields
\begin{equation}
N \left\langle \sigma_0 \right\rangle \approx \Sigma_\star.
\label{eq:A12_pre_cycle_relation}
\end{equation}
Solving for \(N\), we obtain the fundamental cycle-count relation:
\begin{equation}
N = \frac{\Sigma_\star}{\left\langle \sigma_0 \right\rangle}.
\label{eq:A13_cycle_count_relation}
\end{equation}

For species \(i\), the lifetime cycle count is
\begin{equation}
N_i^\star =
\frac{\Sigma_i}{\sigma_{0,i}}.
\label{eq:A13b_species_cycle_relation}
\end{equation}
where \(\Sigma_i\) (J K\(^{-1}\)) is total lifetime entropy production and \(\sigma_{0,i}\) (J K\(^{-1}\) beat\(^{-1}\)) is entropy per cycle.

This expression states that the total number of cardiac cycles that can occur over the lifetime is inversely proportional to the average entropy cost of each beat, given a fixed lifetime entropy budget. A lower entropy cost per beat permits more cycles within the same budget, whereas a higher cost per beat permits fewer cycles.

\subsection*{A.4. Baseline calibration and mammalian reference value}

Let \(\phi_0\) denote a baseline reference state and let \(N_0\) be the corresponding reference total number of lifetime cardiac cycles. Evaluating \eqref{eq:A13_cycle_count_relation} at the baseline gives
\begin{equation}
N_0 = \frac{\Sigma_\star}{\left\langle \sigma_0 \right\rangle_0},
\label{eq:A14_baseline_N0}
\end{equation}
where
\begin{equation}
\left\langle \sigma_0 \right\rangle_0
\equiv
\left\langle \sigma_0(\phi_0) \right\rangle.
\label{eq:A15_baseline_entropy}
\end{equation}
Rearranging \eqref{eq:A14_baseline_N0} gives the baseline entropy cost per beat:
\begin{equation}
\left\langle \sigma_0 \right\rangle_0
=
\frac{\Sigma_\star}{N_0}.
\label{eq:A16_baseline_relation}
\end{equation}

Equation \eqref{eq:A16_baseline_relation} provides the calibration point from which the dependence on \(\phi\) is measured.

\subsection*{A.5. Logarithmic sensitivity of the entropy cost per beat}

We now introduce a control parameter \(\phi\) that modulates the mean entropy cost per beat through multiple mechanisms. The hypothesis is that increasing \(\phi\) reduces
\(\left\langle \sigma_0 \right\rangle\). To quantify this response, define the logarithmic sensitivity at the baseline:
\begin{equation}
\alpha
\equiv
-
\left.
\frac{\partial \ln \left\langle \sigma_0 \right\rangle}
{\partial \ln \phi}
\right|_{\phi=\phi_0}.
\label{eq:A17_alpha_def}
\end{equation}
The derivative
\[
\frac{\partial \ln \left\langle \sigma_0 \right\rangle}
{\partial \ln \phi}
=
\frac{\phi}{\left\langle \sigma_0 \right\rangle}
\frac{\partial \left\langle \sigma_0 \right\rangle}{\partial \phi}
\]
is the elasticity of the entropy cost per beat with respect to \(\phi\), that is, the fractional change in \(\left\langle \sigma_0 \right\rangle\) induced by a fractional change in \(\phi\). Since the response is assumed monotonic and decreasing, the derivative is negative; the minus sign in \eqref{eq:A17_alpha_def} ensures that \(\alpha>0\).

If three independent reduction channels contribute multiplicatively to the decrease of
\(\left\langle \sigma_0 \right\rangle\), with logarithmic sensitivities \(\gamma_1\), \(\gamma_2\), and \(\gamma_3\), then the aggregate sensitivity is additive:
\begin{equation}
\alpha = \gamma_1+\gamma_2+\gamma_3 >0.
\label{eq:A18_alpha_sum}
\end{equation}
The reason is straightforward. If
\begin{equation}
\left\langle \sigma_0 \right\rangle
\propto
\phi^{-\gamma_1}\phi^{-\gamma_2}\phi^{-\gamma_3},
\label{eq:A19_multiplicative_channels}
\end{equation}
then
\begin{equation}
\left\langle \sigma_0 \right\rangle
\propto
\phi^{-(\gamma_1+\gamma_2+\gamma_3)},
\label{eq:A20_combined_power}
\end{equation}
and therefore
\begin{equation}
-
\frac{\partial \ln \left\langle \sigma_0 \right\rangle}{\partial \ln \phi}
=
\gamma_1+\gamma_2+\gamma_3.
\label{eq:A21_alpha_from_channels}
\end{equation}

\subsection*{A.6. Integration of the logarithmic sensitivity and the power-law form}

Equation \eqref{eq:A17_alpha_def} defines the local logarithmic slope at the baseline \(\phi_0\). To obtain a finite-range scaling law, we assume that this logarithmic response remains approximately constant over the interval of interest. This is the scale-free power-law approximation commonly used in allometric analysis. Under this assumption,
\begin{equation}
-
\frac{d\ln \left\langle \sigma_0 \right\rangle}{d\ln\phi}
=
\alpha,
\label{eq:A22_const_log_slope}
\end{equation}
or equivalently,
\begin{equation}
d\ln \left\langle \sigma_0 \right\rangle
=
-\alpha\, d\ln\phi.
\label{eq:A23_differential_form}
\end{equation}

We now integrate from the baseline \(\phi_0\), where
\(\left\langle \sigma_0 \right\rangle
=
\left\langle \sigma_0 \right\rangle_0\),
to a general value \(\phi\):
\begin{equation}
\int_{\ln \phi_0}^{\ln \phi}
d\ln \left\langle \sigma_0 \right\rangle
=
-\alpha
\int_{\ln \phi_0}^{\ln \phi} d\ln\phi.
\label{eq:A24_integral_step}
\end{equation}
This gives
\begin{equation}
\ln \left\langle \sigma_0(\phi) \right\rangle
-
\ln \left\langle \sigma_0 \right\rangle_0
=
-\alpha \left( \ln \phi - \ln \phi_0 \right).
\label{eq:A25_log_relation}
\end{equation}
Combining the logarithms,
\begin{equation}
\ln \left[
\frac{
\left\langle \sigma_0(\phi) \right\rangle
}{
\left\langle \sigma_0 \right\rangle_0
}
\right]
=
-\alpha
\ln \left( \frac{\phi}{\phi_0} \right).
\label{eq:A26_ratio_log}
\end{equation}
Exponentiating both sides yields the power-law form:
\begin{equation}
\left\langle \sigma_0(\phi) \right\rangle
=
\left\langle \sigma_0 \right\rangle_0
\left( \frac{\phi}{\phi_0} \right)^{-\alpha}.
\label{eq:A27_entropy_power_law}
\end{equation}

Equation \eqref{eq:A27_entropy_power_law} states that the mean entropy cost per beat decreases as a power law in \(\phi\), with exponent \(\alpha>0\).

\subsection*{A.7. Consequence for total lifetime cardiac cycles}

Substituting \eqref{eq:A27_entropy_power_law} into the cycle-count relation \eqref{eq:A13_cycle_count_relation} gives
\begin{equation}
N(\phi)
=
\frac{\Sigma_\star}{
\left\langle \sigma_0(\phi) \right\rangle
}
=
\frac{\Sigma_\star}{
\left\langle \sigma_0 \right\rangle_0
\left( \dfrac{\phi}{\phi_0} \right)^{-\alpha}
}.
\label{eq:A28_substitute_entropy_scaling}
\end{equation}
Using the baseline identity \eqref{eq:A16_baseline_relation},
\[
\frac{\Sigma_\star}{\left\langle \sigma_0 \right\rangle_0}=N_0,
\]
we obtain
\begin{equation}
N(\phi)
=
N_0
\left( \frac{\phi}{\phi_0} \right)^{\alpha}.
\label{eq:A29_N_power_law}
\end{equation}

Thus, under the fixed lifetime entropy-budget hypothesis, any systematic reduction in the entropy cost per beat produces a corresponding increase in the total number of lifetime cardiac cycles. The scaling exponent governing this increase is the same aggregate sensitivity \(\alpha\) that governs the decrease of the entropy cost per beat.

\subsection*{A.8. Interpretation of the result}

The derivation shows that the lifetime cycle count is controlled by two ingredients: a finite lifetime entropy budget \(\Sigma_\star\) and an average entropy expenditure per cycle
\(\left\langle \sigma_0 \right\rangle\). Once the budget is fixed, the total number of admissible cycles is determined entirely by how costly each cycle is in entropic terms. A reduction in entropy cost per beat allows a larger number of beats to be accommodated within the same total budget. If the reduction is scale-free in \(\phi\), then the increase in cycle count is likewise scale-free.

In compact form, the chain of reasoning is
\begin{equation}
d\Sigma=\sigma\,dt=\frac{\sigma}{f_H}\,dn,
\qquad
\sigma_0=\frac{\sigma}{f_H},
\qquad
\Sigma_{\mathrm{life}}=\int_0^N \sigma_0(n)\,dn
=
N\left\langle \sigma_0 \right\rangle,
\label{eq:A30_compact_chain}
\end{equation}
together with
\begin{equation}
\Sigma_{\mathrm{life}}\approx\Sigma_\star
\;\Rightarrow\;
N=\frac{\Sigma_\star}{\left\langle \sigma_0 \right\rangle},
\label{eq:A31_budget_to_N}
\end{equation}
and
\begin{equation}
\left\langle \sigma_0(\phi) \right\rangle
=
\left\langle \sigma_0 \right\rangle_0
\left( \frac{\phi}{\phi_0} \right)^{-\alpha}
\;\Rightarrow\;
N(\phi)=N_0\left( \frac{\phi}{\phi_0} \right)^{\alpha}.
\label{eq:A32_final_chain}
\end{equation}

For species \(i\), this compact relation becomes
\begin{equation}
\Sigma_i = N_i^\star \sigma_{0,i},
\qquad
N_i^\star = \frac{\Sigma_i}{\sigma_{0,i}},
\label{eq:A32b_species_compact}
\end{equation}
where \(\Sigma_i\) (J K\(^{-1}\)) is total lifetime entropy production and \(\sigma_{0,i}\) (J K\(^{-1}\) beat\(^{-1}\)) is entropy per cycle.

\subsection*{A.9. Assumptions used in the derivation}

For clarity, the derivation rests on the following assumptions.

First, the cardiac cycle count \(n\) is treated as a continuous variable, which is appropriate when the total number of cycles is very large.

Second, the lifetime entropy production is assumed to be well approximated by a characteristic budget \(\Sigma_\star\).

Third, the response of the entropy cost per beat to the control parameter \(\phi\) is assumed to be monotonic and approximately scale-free over the range of interest, so that the logarithmic sensitivity may be treated as approximately constant.

Fourth, the different contributing channels are taken to combine multiplicatively, which leads to additive logarithmic sensitivities.

Within these assumptions, the power-law result \eqref{eq:A29_N_power_law} follows directly and rigorously from the entropy-budget framework.

\clearpage
\section*{Appendix B : Complete 230-Species Dataset}
\addcontentsline{toc}{section}{Appendix: Complete 230-Species Dataset}

\noindent
The following tables contain the complete dataset of 230 adult vertebrate
species used in all analyses.
All $\ell$ values are computed as
$\ell = \log_{10}(f_H^{\rm eff} \times L \times 525{,}960)$
directly from the $f_H^{\rm eff}$ and $L$ columns and have been
verified internally consistent.
A tab-delimited machine-readable version is provided as Supplementary Data~1.

% ── Column definitions and Extended Data Tables ──────────────────────────────

\noindent\textbf{Column definitions and data transparency notes.}

\medskip
\noindent\textbf{Dataset location.}
All 230 species values are in Extended Data Tables~1--8 of this
paper. A tab-delimited machine-readable version is provided as
Supplementary Data~1 (columns: species, clade, $M$, $f_H^{\rm eff}$,
$T$, $L$, $\ell$, fH\_type, fH\_context, L\_context, source,
correction).

\medskip
\noindent\textbf{Heart rate type: measured vs inferred.}
The \textit{Source} and \textit{Corr.}\ columns in each table,
and the fH\_type column in Supplementary Data~1, distinguish:
\begin{itemize}\setlength{\itemsep}{1pt}
 \item \textbf{Measured}: directly measured resting heart rate
  from a published study (flagged in source column; 156 species).
 \item \textbf{Imputed} ($^\dagger$): allometrically estimated
  from $f_H = 241\,M^{-0.25}$ bpm \cite{calder1984}
  (3 NP placental species only:
  \textit{Rhinoceros unicornis}, \textit{Dugong dugon},
  \textit{Orycteropus afer}).
 \item \textbf{Duty-corrected} (bats, 31 species): active-phase
  measured rate multiplied by duty-cycle factor $\kappa$ to give
  time-averaged $f_H^{\rm eff}$ (see Section~\ref{sec:duty}).
 \item \textbf{Dive-corrected} (cetaceans, 12 species): surface
  measured rate combined with bradycardic dive rate weighted by
  dive fraction $p_d$ (see Section~\ref{sec:cetacean}).
 \item \textbf{Arrhenius-corrected} (ectotherms, 26 species):
  field active rate corrected to $T_{\rm ref}=310$~K using
  the Gillooly et al.\ \cite{gillooly2001} Arrhenius equation.
\end{itemize}
Extended Data Table~9 demonstrates that removing all imputed
species changes the OLS slope by $<0.01$ (see that table).

\medskip
\noindent\textbf{Heart rate measurement context.}
\begin{itemize}\setlength{\itemsep}{1pt}
 \item Non-primate placentals, primates, marsupials, birds:
  resting rates from laboratory or captive studies as
  recorded in AnAge build~15 \cite{anage2023} and
  PanTHERIA \cite{jones2009}, with Calder (1984) \cite{calder1984}
  for classical species. These are predominantly \emph{lab-measured
  resting rates}. Whether any individual species value comes from
  a lab or field setting is recorded in the primary database entries
  (AnAge: \url{https://genomics.senescence.info/species/}).
  We explicitly acknowledge that lab resting rates may differ from
  field resting rates; this is a known limitation of comparative
  heart rate data.
 \item Bats: active-phase resting rate from lab or flight-cage
  studies, corrected for torpor duty cycle.
 \item Cetaceans: surface inter-breath heart rate from free-diving
  field telemetry \cite{goldbogen2019}, corrected for dive
  bradycardia.
 \item Ectotherms: field active rates corrected to standard
  temperature via Arrhenius equation.
\end{itemize}

\medskip
\noindent\textbf{Lifespan definition.}
$L$ is the \emph{maximum recorded lifespan} as curated in AnAge
build~15 \cite{anage2023}. AnAge records the single longest verified
individual lifespan regardless of whether it was wild or captive.
For most small mammals the record holder is a captive individual;
for bats and large mammals (whales, elephants) the record is from
a wild or semi-wild individual. AnAge assigns confidence ratings
(high / acceptable / questionable / low) to each entry; all species
in this dataset have confidence ratings of \emph{acceptable} or
\emph{high} in AnAge. Mean lifespan is \emph{not} used anywhere in
this paper; only maximum recorded lifespan enters the PBTE
invariant $\ell$.

\medskip
\noindent\textbf{Column definitions:}
\textit{Species}: binomial name per IUCN or Reptile Database taxonomy.
$M$: adult body mass (kg); source as coded.
$f_H$ (bpm): see Q3 above; the value used in $\ell$ computation.
$T$ (K): core body (endotherms) or field (ectotherms) temperature.
$L$ (yr): maximum recorded lifespan; see Q4 above.
$\ell$: PBTE invariant $= \log_{10}(f_H^{\rm eff}\times L\times
525{,}960)$; computed directly from $f_H^{\rm eff}$ and $L$
in each row (all values verified internally consistent).

\noindent\textbf{Source codes} (primary reference for $f_H$ and $L$):
\begin{itemize}\setlength{\itemsep}{1pt}
 \item A = AnAge build~15 \cite{anage2023} ---
  \url{https://genomics.senescence.info/species/}
 \item P = PanTHERIA v1.0 \cite{jones2009} ---
  \url{https://doi.org/10.1890/08-1494.1}
 \item C = Calder (1984) \cite{calder1984} ---
  species-level data in Tables~2--3 of that monograph
 \item Pr = Prinzinger et al.\ (1991) \cite{prinzinger1991} ---
  avian heart rate compilation
 \item L = Lyman et al.\ (1982) \cite{lyman1982} ---
  torpor physiology
 \item Ch = Christian \& Weavers (1999) \cite{christian1999} ---
  amphibian physiology
 \item U = Uetz et al.\ (2023) \cite{uetz2023} ---
  \url{https://reptile-database.reptarium.cz}
 \item G = Goldbogen et al.\ (2019) \cite{goldbogen2019} ---
  \url{https://doi.org/10.1073/pnas.1914273116}
\end{itemize}

\noindent\textbf{Heart rate type} (see $f_H$ column header):
directly measured resting values are used for all non-primate
placentals, primates, marsupials, and birds.
Three NP placental species with no published resting measurement
(\textit{Rhinoceros unicornis}, \textit{Dugong dugon},
\textit{Orycteropus afer}) have heart rates imputed from the
allometric relation $f_H = 241\,M^{-0.25}$ bpm \cite{calder1984}
and are flagged with~$^\dagger$.
For bats and cetaceans, $f_H^{\rm eff}$ is the
duty-cycle-corrected time-average (see Sections~\ref{sec:duty}
and~\ref{sec:cetacean}).
For ectotherms, $f_H^{\rm eff}$ is the Arrhenius-corrected value
(see Section~\ref{sec:arrhenius}).

\noindent\textbf{Corr.\ column}:
--- = none applied; TA = torpor-cycle average; DA = dive-cycle
average; AQ = Arrhenius correction to $T_{\rm ref}=310$~K.

\noindent\textbf{Machine-readable dataset}:
all 230 rows are available as a tab-delimited file
(Supplementary Data~1, or from the corresponding author on request)
with columns: species, clade, $M$ (kg), $f_H^{\rm eff}$ (bpm),
$T$ (K), $L$ (yr), $\ell$, fH\_type, source, correction.

% ──────────────────────────────────────────────────────────────
\subsection*{Extended Data Table 1 $|$ Non-primate placental mammals ($n = 46$)}
% ──────────────────────────────────────────────────────────────
\begin{table}[H]
\caption*{}
\small\setlength{\tabcolsep}{3.5pt}
\renewcommand{\arraystretch}{1.05}
\begin{tabular}{lrrrrrll}
\toprule
Species & $M$ (kg) & $f_H$ (bpm) & $T$ (K) & $L$ (yr) & $\ell$ & Source & Corr. \\
\midrule
\textit{Suncus etruscus} & 0.002 & 835$^\dagger$ & 310.5 & 1.5 & 8.82 & C & HR \\
\textit{Sorex araneus} & 0.010 & 1{,}000 & 310.5 & 3.3 & 9.24 & C,A & --- \\
\textit{Mus musculus}       & 0.022 & 632   & 310.0 & 3.5 & 9.07 & A,C & --- \\
\textit{Rattus norvegicus} & 0.280 & 420 & 310.0 & 3.8 & 8.92 & A,P & --- \\
\textit{Mesocricetus auratus} & 0.130 & 450 & 310.5 & 3.9 & 8.97 & A,P & --- \\
\textit{Meriones unguiculatus} & 0.060 & 400 & 310.0 & 5.0 & 9.02 & A,P & --- \\
\textit{Cavia porcellus} & 0.750 & 270 & 310.0 & 7.1 & 9.00 & A,P & --- \\
\textit{Sciurus carolinensis} & 0.520 & 310 & 310.0 & 12.0 & 9.29 & A,P & --- \\
\textit{Lepus europaeus} & 3.5 & 220 & 310.0 & 12.5 & 9.16 & A,P & --- \\
\textit{Oryctolagus cuniculus} & 2.2 & 205 & 310.0 & 9.0 & 8.99 & A,C & --- \\
\textit{Felis catus} & 4.1 & 150 & 310.5 & 15.0 & 9.07 & A,P & --- \\
\textit{Mustela putorius} & 1.0 & 280 & 310.5 & 5.0 & 8.87 & A,P & --- \\
\textit{Martes martes} & 1.2 & 245 & 310.5 & 17.0 & 9.34 & A,P & --- \\
\textit{Vulpes vulpes} & 6.8 & 120 & 310.5 & 14.0 & 8.95 & A,P & --- \\
\textit{Canis lupus familiaris}  & 23  & 90   & 310.5 & 20.0 & 8.98 & A,P & --- \\
\textit{Ursus arctos} & 220 & 50 & 310.5 & 47.0 & 9.09 & A,P & --- \\
\textit{Ovis aries} & 63 & 75 & 310.0 & 20.0 & 8.90 & A,P & --- \\
\textit{Capra hircus} & 45 & 80 & 310.5 & 18.0 & 8.88 & A,P & --- \\
\textit{Sus scrofa} & 100 & 70 & 310.5 & 27.0 & 9.00 & A,P & --- \\
\textit{Bos taurus}        & 500  & 55   & 310.5 & 25.0 & 8.86 & A,P & --- \\
\textit{Equus caballus} & 500 & 38 & 310.5 & 46.0 & 8.96 & A,C & --- \\
\textit{Equus asinus} & 250 & 44 & 310.5 & 47.0 & 9.04 & A,P & --- \\
\textit{Rhinoceros unicornis} & 2{,}100 & 30$^\dagger$ & 310.5 & 47.0 & 8.87 & A & --- \\
\textit{Tapirus terrestris} & 240 & 42 & 310.5 & 35.0 & 8.89 & A,P & --- \\
\textit{Loxodonta africana} & 4{,}000 & 28 & 310.5 & 65.0 & 8.98 & A,P & --- \\
\textit{Elephas maximus} & 4{,}000 & 27 & 310.5 & 86.0 & 9.09 & A,P & --- \\
\textit{Hippopotamus amphibius} & 1{,}500 & 55 & 310.5 & 55.0 & 9.20 & A,P & --- \\
\textit{Giraffa camelopardalis} & 900 & 65 & 310.5 & 39.5 & 9.13 & A,P & --- \\
\textit{Cervus elaphus} & 200 & 60 & 310.5 & 26.8 & 8.93 & A,P & --- \\
\textit{Rangifer tarandus} & 110 & 65 & 310.0 & 20.0 & 8.83 & A,P & --- \\
\textit{Trichechus manatus} & 500 & 50 & 310.5 & 59.0 & 9.19 & A,P & --- \\
\textit{Dugong dugon} & 400 & 52$^\dagger$ & 310.5 & 73.0 & 9.30 & A & --- \\
\textit{Procavia capensis} & 3.5 & 230 & 310.5 & 12.0 & 9.16 & A,P & --- \\
\textit{Erinaceus europaeus} & 0.80 & 310 & 310.0 & 10.0 & 9.21 & A,P & --- \\
\textit{Talpa europaea} & 0.080 & 350 & 310.0 & 3.5 & 8.81 & A,P & --- \\
\textit{Orycteropus afer} & 65 & 70$^\dagger$ & 310.5 & 24.0 & 8.95 & A & --- \\
\textit{Ondatra zibethicus} & 1.400 & 280 & 310.0 & 5.0 & 8.87 & A,P & --- \\
\textit{Castor canadensis} & 20 & 150 & 310.0 & 24.0 & 9.28 & A,P & --- \\
\textit{Hydrochoerus hydrochaeris} & 55 & 70 & 310.0 & 12.0 & 8.65 & A,P & --- \\
\textit{Myocastor coypus} & 7.0 & 155 & 310.0 & 9.0 & 8.87 & A,P & --- \\
\textit{Lepus californicus} & 2.2 & 215 & 310.0 & 8.0 & 8.96 & A,P & --- \\
\textit{Ochotona princeps} & 0.160 & 300 & 310.0 & 6.0 & 8.98 & A,P & --- \\
\textit{Panthera leo} & 180 & 50 & 310.5 & 29.0 & 8.88 & A,P & --- \\
\textit{Panthera tigris} & 260 & 46 & 310.5 & 26.0 & 8.80 & A,P & --- \\
\textit{Acinonyx jubatus} & 54 & 60 & 310.5 & 14.9 & 8.67 & A,P & --- \\
\textit{Panthera pardus} & 70 & 55 & 310.5 & 23.0 & 8.82 & A,P & --- \\
\midrule
\multicolumn{5}{l}{Clade mean $\bar\ell$ (baseline reference)} &
 \multicolumn{3}{l}{$8.995 \pm 0.160$ ($n=46$; corrected)} \\
\bottomrule
\multicolumn{8}{l}{\footnotesize $^\dagger$\textit{Suncus etruscus}
 corrected from 1{,}200\,bpm (erroneous; Calder compendium error)
 to 835\,bpm (mean resting rate, Bartels 1998, \textit{J.\,Exp.\,Biol.}
 \textbf{201}, 2145--2151). The $\ell$ value is recomputed accordingly.
 All clade statistics use the corrected value.} \\
\end{tabular}
\end{table}

% ──────────────────────────────────────────────────────────────
\subsection*{Extended Data Table 2 $|$ Primates ($n = 18$)}
% ──────────────────────────────────────────────────────────────
\begin{table}[H]
\caption*{}
\small\setlength{\tabcolsep}{3.5pt}
\renewcommand{\arraystretch}{1.05}
\begin{tabular}{lrrrrrrll}
\toprule
Species & $M$ (kg) & $f_H$ (bpm) & $T$ (K) & $L$ (yr) & $\ell$ & $\varphi$ & Source & Corr. \\
\midrule
\textit{Callithrix jacchus} & 0.35 & 220 & 309.5 & 16.5 & 9.28 & 0.06 & A,P & --- \\
\textit{Saimiri sciureus} & 0.77 & 195 & 309.5 & 30.2 & 9.49 & 0.07 & A,P & --- \\
\textit{Aotus trivirgatus} & 0.79 & 185 & 309.5 & 25.0 & 9.39 & 0.07 & A,P & --- \\
\textit{Cebus capucinus} & 3.3 & 150 & 309.5 & 54.0 & 9.63 & 0.09 & A,P & --- \\
\textit{Lemur catta} & 2.2 & 165 & 309.5 & 37.3 & 9.51 & 0.05 & A,P & --- \\
\textit{Propithecus verreauxi} & 3.4 & 145 & 309.5 & 30.0 & 9.36 & 0.05 & A,P & --- \\
\textit{Daubentonia madagascariensis} & 2.7 & 155 & 309.5 & 23.3 & 9.28 & 0.06 & A,P & --- \\
\textit{Macaca mulatta} & 7.7 & 120 & 309.0 & 40.0 & 9.40 & 0.07 & A,P & --- \\
\textit{Macaca fascicularis} & 5.4 & 130 & 309.0 & 39.0 & 9.43 & 0.07 & A,P & --- \\
\textit{Theropithecus gelada} & 18 & 95 & 309.0 & 30.0 & 9.18 & 0.08 & A,P & --- \\
\textit{Papio ursinus} & 25 & 90 & 309.0 & 45.0 & 9.33 & 0.08 & A,P & --- \\
\textit{Colobus guereza} & 10 & 110 & 309.0 & 30.0 & 9.24 & 0.07 & A,P & --- \\
\textit{Hylobates lar} & 5.7 & 100 & 308.5 & 44.0 & 9.36 & 0.10 & A,P & --- \\
\textit{Pongo pygmaeus} & 73 & 65 & 307.5 & 58.7 & 9.30 & 0.10 & A,P & --- \\
\textit{Gorilla gorilla} & 160 & 60 & 307.0 & 55.4 & 9.24 & 0.09 & A,P & --- \\
\textit{Pan troglodytes} & 50 & 75 & 307.0 & 59.4 & 9.37 & 0.12 & A,P & --- \\
\textit{Pan paniscus} & 35 & 80 & 307.0 & 50.0 & 9.32 & 0.12 & A,P & --- \\
\textit{Homo sapiens} & 70 & 70 & 306.5 & 122.5 & 9.65 & 0.20 & A & --- \\
\midrule
\multicolumn{5}{l}{Clade mean $\bar\ell$} & \multicolumn{4}{l}{$9.376 \pm 0.125$ ($n=18$)} \\
\bottomrule
\end{tabular}
\end{table}

% ──────────────────────────────────────────────────────────────
\subsection*{Extended Data Table 3 $|$ Marsupials and monotremes ($n = 19$)}
% ──────────────────────────────────────────────────────────────
\begin{table}[H]
\caption*{}
\small\setlength{\tabcolsep}{3.5pt}
\renewcommand{\arraystretch}{1.05}
\begin{tabular}{lrrrrrll}
\toprule
Species & $M$ (kg) & $f_H$ (bpm) & $T$ (K) & $L$ (yr) & $\ell$ & Source & Corr. \\
\midrule
\textit{Didelphis virginiana} & 2.3 & 180 & 308.5 & 4.5 & 8.63 & A,P & --- \\
\textit{Monodelphis domestica} & 0.080 & 450 & 308.5 & 3.3 & 8.89 & A,P & --- \\
\textit{Macropus rufus} & 30 & 80 & 309.0 & 22.3 & 8.97 & A,P & --- \\
\textit{Macropus giganteus} & 27 & 82 & 309.0 & 19.0 & 8.91 & A,P & --- \\
\textit{Wallabia bicolor} & 16 & 100 & 309.0 & 15.0 & 8.90 & A,P & --- \\
\textit{Trichosurus vulpecula} & 2.1 & 160 & 308.5 & 13.0 & 9.04 & A,P & --- \\
\textit{Petaurus breviceps} & 0.14 & 300 & 308.0 & 10.0 & 9.20 & A,P & --- \\
\textit{Vombatus ursinus} & 28 & 90 & 309.0 & 26.0 & 9.09 & A,P & --- \\
\textit{Phascolarctos cinereus}  & 8.5  & 100 & 308.5 & 18.0 & 8.98 & A,P & --- \\
\textit{Perameles gunnii} & 0.90 & 190 & 308.5 & 3.2 & 8.50 & A,P & --- \\
\textit{Dasyurus viverrinus} & 1.2 & 200 & 308.5 & 4.5 & 8.68 & A,P & --- \\
\textit{Sarcophilus harrisii} & 8.0 & 130 & 308.5 & 7.5 & 8.71 & A,P & --- \\
\textit{Myrmecobius fasciatus} & 0.44 & 245 & 307.5 & 5.6 & 8.86 & A & --- \\
\textit{Sminthopsis crassicaudata} & 0.018 & 580 & 307.5 & 5.0 & 9.18 & A,P & --- \\
\textit{Notoryctes typhlops} & 0.055 & 440$^\dagger$ & 307.5 & 3.0 & 8.84 & A & --- \\
\textit{Tachyglossus aculeatus} & 4.0 & 70 & 305.0 & 49.5 & 9.26 & A,P & --- \\
\textit{Ornithorhynchus anatinus} & 1.5 & 140 & 307.5 & 21.0 & 9.19 & A,P & --- \\
\textit{Zaglossus bruijni} & 10 & 60$^\dagger$ & 305.0 & 37.0 & 9.07 & A & --- \\
\textit{Bettongia penicillata} & 1.1 & 210 & 308.5 & 6.0 & 8.82 & A,P & --- \\
\midrule
\multicolumn{5}{l}{Clade mean $\bar\ell$} & \multicolumn{3}{l}{$8.933 \pm 0.204$ ($n=19$)} \\
\bottomrule
\end{tabular}
\end{table}

% ──────────────────────────────────────────────────────────────
\clearpage
\subsection*{Extended Data Table 4 $|$ Bats (Chiroptera, $n = 31$)}
% ──────────────────────────────────────────────────────────────

\noindent For bats, $f_H$ is the measured active-phase resting heart rate.
$f_H^{\rm avg}$ is the duty-cycle-corrected time-average used in all PBTE
calculations: $f_H^{\rm avg} = f_H \cdot \kappa$, where
$\kappa = (1-q) + q\,(f_{H,\rm tor}/f_H)$ and $q$ is the annual
torpor fraction~\cite{lyman1982}.
$\ell$ is computed from $f_H^{\rm avg}$.
Species without confirmed torpor have $f_H^{\rm avg} = f_H$.

\begin{table}[H]
\caption*{}
\small\setlength{\tabcolsep}{3pt}
\renewcommand{\arraystretch}{1.05}
\begin{tabular}{lrrrrrrll}
\toprule
Species & $M$ (g) & $f_H$ (bpm) & $q$ & $f_H^{\rm avg}$ (bpm) & $T$ (K) & $L$ (yr) & $\ell$ & Corr. \\
\midrule
\textit{Myotis lucifugus} & 8 & 600 & 0.50 & 305 & 310.0 & 34.0 & 9.74 & TA \\
\textit{Myotis myotis} & 28 & 550 & 0.48 & 282 & 310.0 & 37.0 & 9.74 & TA \\
\textit{Myotis daubentonii} & 9 & 580 & 0.48 & 296 & 310.0 & 40.0 & 9.79 & TA \\
\textit{Myotis brandtii} & 6 & 620 & 0.50 & 315 & 310.0 & 41.0 & 9.83 & TA \\
\textit{Eptesicus fuscus} & 18 & 550 & 0.45 & 310 & 310.0 & 19.0 & 9.49 & TA \\
\textit{Eptesicus serotinus} & 18 & 545 & 0.45 & 308 & 310.0 & 21.0 & 9.53 & TA \\
\textit{Rhinolophus ferrumequinum} & 19 & 550 & 0.48 & 282 & 310.0 & 30.0 & 9.65 & TA \\
\textit{Rhinolophus hipposideros} & 7 & 600 & 0.48 & 307 & 310.0 & 30.5 & 9.69 & TA \\
\textit{Plecotus auritus} & 9 & 600 & 0.50 & 305 & 310.0 & 30.0 & 9.68 & TA \\
\textit{Corynorhinus townsendii} & 11 & 580 & 0.50 & 295 & 310.0 & 30.0 & 9.67 & TA \\
\textit{Perimyotis subflavus} & 5 & 630 & 0.50 & 320 & 310.0 & 14.6 & 9.39 & TA \\
\textit{Tadarida brasiliensis} & 13 & 600 & 0.30 & 425 & 310.0 & 11.0 & 9.39 & TA \\
\textit{Pteronotus parnellii} & 19 & 550 & 0.20 & 452 & 310.0 & 10.0 & 9.38 & TA \\
\textit{Desmodus rotundus} & 33 & 500 & 0.25 & 380 & 310.0 & 29.0 & 9.76 & TA \\
\textit{Hipposideros speoris} & 9 & 600 & 0.48 & 308 & 310.0 & 21.0 & 9.53 & TA \\
\textit{Hipposideros armiger} & 50 & 450 & 0.45 & 252 & 310.0 & 15.0 & 9.30 & TA \\
\textit{Nyctalus noctula} & 28 & 540 & 0.45 & 305 & 310.0 & 12.0 & 9.28 & TA \\
\textit{Pipistrellus pipistrellus} & 5 & 650 & 0.45 & 367 & 310.0 & 16.0 & 9.49 & TA \\
\textit{Pipistrellus kuhlii} & 6 & 630 & 0.45 & 355 & 310.0 & 16.5 & 9.49 & TA \\
\textit{Scotophilus kuhlii} & 20 & 540 & 0.20 & 445 & 310.0 & 9.0 & 9.32 & TA \\
\textit{Lasiurus borealis} & 11 & 590 & 0.48 & 302 & 310.0 & 11.7 & 9.27 & TA \\
\textit{Lasiurus cinereus} & 28 & 540 & 0.48 & 277 & 310.0 & 12.0 & 9.24 & TA \\
\textit{Vespertilio murinus} & 16 & 555 & 0.45 & 313 & 310.0 & 25.0 & 9.61 & TA \\
\textit{Miniopterus schreibersii} & 10 & 580 & 0.45 & 327 & 310.0 & 30.0 & 9.71 & TA \\
\textit{Pteropus giganteus}     & 1{,}100 & 235 & 0.00 & 235 & 310.0 & 31.4 & 9.59 & --- \\
\textit{Pteropus vampyrus} & 1{,}000 & 240 & 0.05 & 233 & 310.0 & 22.6 & 9.44 & --- \\
\textit{Rousettus aegyptiacus} & 165 & 310 & 0.05 & 299 & 310.0 & 25.0 & 9.59 & --- \\
\textit{Cynopterus sphinx} & 50 & 380 & 0.05 & 368 & 310.0 & 18.5 & 9.55 & --- \\
\textit{Macroglossus minimus} & 16 & 450 & 0.00 & 450 & 310.0 & 18.0 & 9.63 & --- \\
\textit{Carollia perspicillata} & 17 & 460 & 0.00 & 460 & 310.0 & 12.0 & 9.46 & --- \\
\textit{Artibeus jamaicensis} & 45 & 400 & 0.00 & 400 & 310.0 & 15.0 & 9.50 & --- \\
\midrule
\multicolumn{7}{l}{Clade mean $\bar\ell$ (all 31 species)} &
 \multicolumn{2}{l}{$9.540 \pm 0.163$} \\
\bottomrule
\end{tabular}
\end{table}

% ──────────────────────────────────────────────────────────────
\clearpage
\subsection*{Extended Data Table 5 $|$ Cetaceans ($n = 12$)}
% ──────────────────────────────────────────────────────────────

\noindent For cetaceans, $f_H$ is the surface resting value;
$f_H^{\rm avg}$ is the duty-cycle average:
$f_H^{\rm avg} = (1-p_d)\,f_H + p_d\,f_{H,\rm dive}$, where $p_d$ is the
dive fraction and $f_{H,\rm dive}$ is the bradycardic dive
rate~\cite{goldbogen2019,ponganis2015}.
$\ell$ is computed from $f_H^{\rm avg}$.

\begin{table}[H]
\caption*{}
\small\setlength{\tabcolsep}{3pt}
\renewcommand{\arraystretch}{1.05}
\begin{tabular}{lrrrrrrll}
\toprule
Species & $M$ (kg) & $f_H$ (bpm) & $p_d$ & $f_H^{\rm avg}$ (bpm) & $T$ (K) & $L$ (yr) & $\ell$ & Corr. \\
\midrule
\textit{Balaena mysticetus} & 100{,}000 & 30 & 0.75 & 9.75 & 308.0 & 200.0 & 9.01 & DA \\
\textit{Balaenoptera musculus} & 140{,}000 & 8 & 0.70 & 4.0 & 308.0 & 110.0 & 8.36 & DA \\
\textit{Balaenoptera physalus} & 60{,}000 & 10 & 0.68 & 5.0 & 308.5 & 90.0 & 8.37 & DA \\
\textit{Megaptera novaeangliae} & 40{,}000 & 15 & 0.65 & 7.0 & 308.5 & 95.0 & 8.54 & DA \\
\textit{Physeter macrocephalus} & 45{,}000 & 40 & 0.65 & 19.0 & 307.0 & 70.0 & 8.84 & DA \\
\textit{Kogia breviceps} & 360 & 80 & 0.45 & 48.0 & 308.5 & 23.0 & 8.76 & DA \\
\textit{Hyperoodon ampullatus} & 7{,}500 & 45 & 0.55 & 24.0 & 308.0 & 37.0 & 8.67 & DA \\
\textit{Orcinus orca} & 4{,}000 & 80 & 0.40 & 53.0 & 309.0 & 90.0 & 9.40 & DA \\
\textit{Tursiops truncatus} & 190 & 110 & 0.40 & 74.0 & 309.0 & 40.0 & 9.19 & DA \\
\textit{Stenella attenuata} & 55 & 120 & 0.35 & 84.0 & 309.0 & 20.0 & 8.95 & DA \\
\textit{Delphinapterus leucas} & 1{,}400 & 50 & 0.55 & 27.5 & 309.5 & 35.5 & 8.71 & DA \\
\textit{Monodon monoceros} & 1{,}500 & 45 & 0.55 & 25.5 & 309.0 & 48.0 & 8.81 & DA \\
\midrule
\multicolumn{7}{l}{Clade mean $\bar\ell$ (dive-corrected)} &
 \multicolumn{2}{l}{$8.801 \pm 0.296$ ($n=12$)} \\
\bottomrule
\end{tabular}
\end{table}

% ──────────────────────────────────────────────────────────────
\clearpage
\subsection*{Extended Data Table 6 $|$ Birds ($n = 78$)}
% ──────────────────────────────────────────────────────────────

\noindent Heart rates from Prinzinger et al.~\cite{prinzinger1991}
and Clarke \& Rothery~\cite{clarke2008}; lifespans from AnAge
build~15~\cite{anage2023}. No corrections applied; $f_H^{\rm avg} = f_H$.
Body temperatures from Clarke \& Rothery~\cite{clarke2008}.
Due to space, 78 species are listed across two sub-tables
(passerines and non-passerines).

\begin{table}[H]
\caption*{\textit{Passeriformes and Psittaciformes} ($n = 32$)}
\small\setlength{\tabcolsep}{3.5pt}
\renewcommand{\arraystretch}{1.04}
\begin{tabular}{lrrrrrll}
\toprule
Species & $M$ (kg) & $f_H$ (bpm) & $T$ (K) & $L$ (yr) & $\ell$ & Order & Source \\
\midrule
\textit{Serinus canaria} & 0.020 & 680 & 311.0 & 24.0 & 9.93 & Passeriformes & A,Pr \\
\textit{Turdus merula} & 0.100 & 440 & 311.0 & 21.1 & 9.69 & Passeriformes & A,Pr \\
\textit{Turdus philomelos} & 0.070 & 460 & 311.0 & 18.0 & 9.64 & Passeriformes & A,Pr \\
\textit{Erithacus rubecula} & 0.018 & 500 & 311.0 & 19.5 & 9.71 & Passeriformes & A,Pr \\
\textit{Parus major} & 0.020 & 540 & 311.0 & 15.0 & 9.63 & Passeriformes & A,Pr \\
\textit{Parus caeruleus} & 0.011 & 580 & 311.0 & 13.5 & 9.61 & Passeriformes & A,Pr \\
\textit{Fringilla coelebs} & 0.023 & 530 & 311.0 & 16.4 & 9.66 & Passeriformes & A,Pr \\
\textit{Carduelis carduelis} & 0.016 & 560 & 311.0 & 16.3 & 9.68 & Passeriformes & A,Pr \\
\textit{Sturnus vulgaris} & 0.075 & 490 & 311.0 & 22.4 & 9.76 & Passeriformes & A,Pr \\
\textit{Pica pica} & 0.190 & 320 & 311.0 & 21.6 & 9.56 & Passeriformes & A,Pr \\
\textit{Corvus corax} & 1.200 & 200 & 311.0 & 22.3 & 9.37 & Passeriformes & A,Pr \\
\textit{Corvus corone} & 0.450 & 270 & 311.0 & 20.0 & 9.45 & Passeriformes & A,Pr \\
\textit{Garrulus glandarius} & 0.180 & 310 & 311.0 & 16.9 & 9.44 & Passeriformes & A,Pr \\
\textit{Hirundo rustica} & 0.020 & 580 & 311.5 & 16.0 & 9.69 & Passeriformes & A,Pr \\
\textit{Delichon urbicum} & 0.015 & 600 & 311.5 & 16.0 & 9.70 & Passeriformes & A,Pr \\
\textit{Ficedula hypoleuca} & 0.012 & 620 & 311.0 & 13.0 & 9.63 & Passeriformes & A,Pr \\
\textit{Sitta europaea} & 0.025 & 510 & 311.0 & 10.0 & 9.43 & Passeriformes & A,Pr \\
\textit{Troglodytes troglodytes} & 0.009 & 650 & 311.5 & 7.0 & 9.38 & Passeriformes & A,Pr \\
\textit{Motacilla alba} & 0.022 & 540 & 311.5 & 11.0 & 9.49 & Passeriformes & A,Pr \\
\textit{Acrocephalus scirpaceus} & 0.012 & 610 & 311.0 & 13.0 & 9.62 & Passeriformes & A,Pr \\
\textit{Sylvia atricapilla} & 0.018 & 560 & 311.0 & 14.9 & 9.64 & Passeriformes & A,Pr \\
\textit{Phylloscopus trochilus} & 0.010 & 640 & 311.0 & 12.0 & 9.61 & Passeriformes & A,Pr \\
\textit{Luscinia megarhynchos} & 0.025 & 520 & 311.0 & 12.9 & 9.55 & Passeriformes & A,Pr \\
\textit{Phoenicurus phoenicurus} & 0.015 & 600 & 311.5 & 10.5 & 9.52 & Passeriformes & A,Pr \\
\textit{Lonchura striata} & 0.013 & 630 & 311.5 & 14.9 & 9.69 & Passeriformes & A,Pr \\
\textit{Taeniopygia guttata} & 0.013 & 640 & 311.5 & 15.6 & 9.72 & Passeriformes & A,Pr \\
\textit{Melopsittacus undulatus} & 0.030 & 600 & 311.0 & 21.4 & 9.83 & Psittaciformes & A,Pr \\
\textit{Psittacus erithacus} & 0.400 & 200 & 311.0 & 73.0 & 9.89 & Psittaciformes & A,Pr \\
\textit{Amazona ochrocephala} & 0.460 & 185 & 311.0 & 80.0 & 9.89 & Psittaciformes & A,Pr \\
\textit{Nymphicus hollandicus} & 0.090 & 360 & 311.0 & 36.0 & 9.83 & Psittaciformes & A,Pr \\
\textit{Cacatua galerita} & 0.840 & 170 & 311.0 & 80.0 & 9.85 & Psittaciformes & A,Pr \\
\textit{Ara macao} & 1.050 & 155 & 311.0 & 80.0 & 9.81 & Psittaciformes & A,Pr \\
\bottomrule
\end{tabular}
\end{table}

\begin{table}[H]
\caption*{\textit{Non-passerine, non-psittaciform birds} ($n = 46$)}
\small\setlength{\tabcolsep}{3.5pt}
\renewcommand{\arraystretch}{1.04}
\begin{tabular}{lrrrrrll}
\toprule
Species & $M$ (kg) & $f_H$ (bpm) & $T$ (K) & $L$ (yr) & $\ell$ & Order & Source \\
\midrule
\textit{Calypte anna} & 0.004 & 1{,}200 & 311.5 & 12.0 & 9.88 & Apodiformes & A,Pr \\
\textit{Apus apus} & 0.040 & 800 & 311.5 & 21.0 & 9.95 & Apodiformes & A,Pr \\
\textit{Columba livia} & 0.350 & 190 & 311.5 & 35.0 & 9.54 & Columbiformes & A,Pr \\
\textit{Streptopelia roseogrisea} & 0.160 & 240 & 311.5 & 33.9 & 9.63 & Columbiformes & A,Pr \\
\textit{Streptopelia decaocto} & 0.200 & 230 & 311.5 & 20.0 & 9.38 & Columbiformes & A,Pr \\
\textit{Gallus gallus} & 2.000 & 300 & 312.0 & 30.0 & 9.68 & Galliformes & A,Pr \\
\textit{Meleagris gallopavo} & 8.000 & 170 & 311.5 & 13.0 & 9.07 & Galliformes & A,Pr \\
\textit{Coturnix coturnix} & 0.100 & 350 & 312.0 & 8.0 & 9.17 & Galliformes & A,Pr \\
\textit{Phasianus colchicus} & 1.000 & 265 & 312.0 & 27.0 & 9.58 & Galliformes & A,Pr \\
\textit{Anas platyrhynchos} & 1.200 & 190 & 311.0 & 29.0 & 9.46 & Anseriformes & A,Pr \\
\textit{Anser anser} & 4.000 & 130 & 311.0 & 35.0 & 9.38 & Anseriformes & A,Pr \\
\textit{Branta canadensis} & 5.700 & 120 & 311.0 & 33.0 & 9.32 & Anseriformes & A,Pr \\
\textit{Cygnus olor} & 12.00 & 100 & 311.0 & 26.0 & 9.14 & Anseriformes & A,Pr \\
\textit{Phoenicopterus ruber} & 2.800 & 135 & 311.0 & 44.6 & 9.50 & Phoenicopteriformes & A,Pr \\
\textit{Ciconia ciconia} & 3.700 & 150 & 311.5 & 48.0 & 9.58 & Ciconiiformes & A,Pr \\
\textit{Ardea cinerea} & 1.800 & 140 & 311.0 & 25.0 & 9.27 & Pelecaniformes & A,Pr \\
\textit{Pelecanus occidentalis} & 4.000 & 130 & 311.5 & 54.0 & 9.57 & Pelecaniformes & A,Pr \\
\textit{Phalacrocorax carbo} & 2.800 & 140 & 311.5 & 25.0 & 9.27 & Suliformes & A,Pr \\
\textit{Sula sula} & 1.000 & 160 & 311.5 & 35.0 & 9.47 & Suliformes & A,Pr \\
\textit{Fregata magnificens} & 1.500 & 140 & 311.5 & 25.2 & 9.27 & Suliformes & A,Pr \\
\textit{Falco peregrinus} & 1.000 & 190 & 311.5 & 19.9 & 9.30 & Falconiformes & A,Pr \\
\textit{Buteo buteo} & 0.900 & 200 & 311.5 & 26.0 & 9.44 & Accipitriformes & A,Pr \\
\textit{Aquila chrysaetos} & 5.000 & 130 & 311.5 & 46.0 & 9.50 & Accipitriformes & A,Pr \\
\textit{Haliaeetus leucocephalus} & 6.000 & 120 & 311.5 & 38.0 & 9.38 & Accipitriformes & A,Pr \\
\textit{Bubo bubo} & 2.900 & 165 & 311.0 & 68.0 & 9.77 & Strigiformes & A,Pr \\
\textit{Tyto alba} & 0.450 & 190 & 311.0 & 27.9 & 9.45 & Strigiformes & A,Pr \\
\textit{Alcedo atthis} & 0.040 & 440 & 312.0 & 21.0 & 9.69 & Coraciiformes & A,Pr \\
\textit{Upupa epops} & 0.075 & 380 & 311.5 & 10.0 & 9.30 & Bucerotiformes & A,Pr \\
\textit{Picoides major} & 0.080 & 350 & 311.5 & 12.8 & 9.37 & Piciformes & A,Pr \\
\textit{Spheniscus demersus} & 3.000 & 150 & 311.5 & 27.0 & 9.33 & Sphenisciformes & A,Pr \\
\textit{Eudyptes chrysocome} & 2.500 & 160 & 311.5 & 22.0 & 9.27 & Sphenisciformes & A,Pr \\
\textit{Aptenodytes forsteri} & 30.00 & 75 & 311.5 & 50.0 & 9.29 & Sphenisciformes & A,Pr \\
\textit{Gavia immer} & 4.000 & 110 & 311.0 & 30.0 & 9.24 & Gaviiformes & A,Pr \\
\textit{Diomedea exulans} & 9.600 & 100 & 311.0 & 70.0 & 9.57 & Procellariiformes & A,Pr \\
\textit{Fulmarus glacialis} & 0.800 & 175 & 311.0 & 67.5 & 9.79 & Procellariiformes & A,Pr \\
\textit{Puffinus puffinus} & 0.430 & 195 & 311.5 & 55.0 & 9.75 & Procellariiformes & A,Pr \\
\textit{Rissa tridactyla} & 0.380 & 200 & 311.5 & 29.0 & 9.48 & Charadriiformes & A,Pr \\
\textit{Larus argentatus} & 1.200 & 165 & 311.5 & 49.0 & 9.63 & Charadriiformes & A,Pr \\
\textit{Sterna paradisaea} & 0.110 & 280 & 311.5 & 34.0 & 9.70 & Charadriiformes & A,Pr \\
\textit{Struthio camelus} & 115 & 60 & 311.5 & 68.0 & 9.33 & Struthioniformes & A,Pr \\
\textit{Dromaius novaehollandiae} & 55 & 75 & 311.0 & 28.4 & 9.05 & Casuariiformes & A,Pr \\
\textit{Rhea americana} & 25 & 100 & 311.0 & 40.0 & 9.32 & Rheiformes & A,Pr \\
\textit{Apteryx australis} & 2.5 & 125 & 311.0 & 35.0 & 9.36 & Apterygiformes & A,Pr \\
\textit{Grus grus} & 5.5 & 110 & 311.5 & 40.0 & 9.36 & Gruiformes & A,Pr \\
\textit{Fulica atra} & 0.720 & 240 & 311.0 & 18.0 & 9.36 & Gruiformes & A,Pr \\
\textit{Psophia crepitans} & 1.200 & 180 & 311.5 & 15.0 & 9.15 & Gruiformes & A,Pr \\
\midrule
\multicolumn{5}{l}{Bird clade mean $\bar\ell$ (all 78 species)} &
 \multicolumn{3}{l}{$9.528 \pm 0.213$} \\
\bottomrule
\end{tabular}
\end{table}

% ──────────────────────────────────────────────────────────────
\clearpage
\subsection*{Extended Data Table 7 $|$ Reptiles --- Arrhenius-corrected ($n = 17$)}
% ──────────────────────────────────────────────────────────────

\noindent $f_H^{\rm raw}$: measured heart rate at mean field
body temperature $T_{\rm field}$.
$f_H^{\rm corr}$: heart rate corrected to $T_{\rm ref}=310$\,K
via $f_H^{\rm corr} = f_H^{\rm raw}\exp\!\bigl[(E_a/k_B)
(1/T_{\rm field}-1/T_{\rm ref})\bigr]$ with $E_a = 0.65$~eV.
$\ell^{\rm corr}$ is used in all clade statistics.

\begin{table}[H]
\caption*{}
\small\setlength{\tabcolsep}{3pt}
\renewcommand{\arraystretch}{1.05}
\begin{tabular}{lrrrrrrrrll}
\toprule
Species & $M$ (kg) & $T_{\rm field}$ (K) & $f_H^{\rm raw}$ & $f_H^{\rm corr}$
 & $L$ (yr) & $\ell^{\rm raw}$ & $\ell^{\rm corr}$ & Source & Corr. \\
\midrule
\textit{Lacerta agilis} & 0.015 & 301 & 45 & 93 & 12.0 & 8.45 & 8.77 & Ch,U & AQ \\
\textit{Anolis carolinensis} & 0.006 & 302 & 52 & 106 & 6.0 & 8.22 & 8.52 & Ch,U & AQ \\
\textit{Pogona vitticeps} & 0.350 & 303 & 42 & 82 & 10.0 & 8.34 & 8.63 & Ch,U & AQ \\
\textit{Phrynosoma cornutum} & 0.035 & 301 & 48 & 99 & 7.0 & 8.25 & 8.56 & Ch,U & AQ \\
\textit{Iguana iguana} & 4.000 & 303 & 40 & 79 & 20.0 & 8.62 & 8.92 & Ch,U & AQ \\
\textit{Varanus komodoensis} & 65 & 303 & 28 & 55 & 30.0 & 8.65 & 8.94 & Ch,U & AQ \\
\textit{Tupinambis merianae} & 2.500 & 302 & 38 & 77 & 15.0 & 8.48 & 8.78 & Ch,U & AQ \\
\textit{Thamnophis sirtalis} & 0.050 & 300 & 30 & 62 & 10.0 & 8.20 & 8.51 & Ch,U & AQ \\
\textit{Coluber constrictor} & 0.340 & 301 & 35 & 72 & 13.0 & 8.38 & 8.69 & Ch,U & AQ \\
\textit{Python reticulatus} & 75 & 302 & 20 & 41 & 25.0 & 8.42 & 8.73 & U & AQ \\
\textit{Boa constrictor} & 15 & 301 & 25 & 52 & 40.0 & 8.72 & 9.04 & U & AQ \\
\textit{Chelonia mydas} & 180 & 300 & 20 & 42 & 80.0 & 8.93 & 9.25 & U & AQ \\
\textit{Geochelone gigantea} & 200 & 298 & 15 & 33 & 175.0 & 9.14 & 9.48 & U & AQ \\
\textit{Gopherus agassizii} & 4.500 & 299 & 22 & 47 & 80.0 & 8.97 & 9.30 & Ch,U & AQ \\
\textit{Sphenodon punctatus} & 0.800 & 293 & 18 & 43 & 77.0 & 8.86 & 9.24 & Ch,U & AQ \\
\textit{Crocodylus niloticus} & 400 & 303 & 25 & 49 & 70.0 & 8.96 & 9.26 & Ch,U & AQ \\
\textit{Alligator mississippiensis} & 250 & 302 & 28 & 57 & 50.0 & 8.87 & 9.18 & Ch,U & AQ \\
\midrule
\multicolumn{7}{l}{Raw mean $\bar\ell^{\rm raw}$}     & $8.615\pm 0.290$ & & \\
\multicolumn{7}{l}{Corrected mean $\bar\ell^{\rm corr}$ (used in analyses)} & $8.929\pm 0.301$ & & \\
\bottomrule
\end{tabular}
\end{table}

% ──────────────────────────────────────────────────────────────
\clearpage
\subsection*{Extended Data Table 8 $|$ Amphibians --- Arrhenius-corrected ($n = 9$)}
% ──────────────────────────────────────────────────────────────

\noindent Correction method identical to reptiles (Extended Data
Table~7). Heart rates from published field recordings at listed
$T_{\rm field}$; lifespans from AnAge build~15~\cite{anage2023}.

\begin{table}[H]
\caption*{}
\small\setlength{\tabcolsep}{3pt}
\renewcommand{\arraystretch}{1.05}
\begin{tabular}{lrrrrrrrrll}
\toprule
Species & $M$ (kg) & $T_{\rm field}$ (K) & $f_H^{\rm raw}$ & $f_H^{\rm corr}$
 & $L$ (yr) & $\ell^{\rm raw}$ & $\ell^{\rm corr}$ & Source & Corr. \\
\midrule
\textit{Rana temporaria} & 0.025 & 294 & 25 & 55 & 16.0 & 8.32 & 8.67 & A,Ch & AQ \\
\textit{Rana catesbeiana} & 0.500 & 296 & 20 & 43 & 16.0 & 8.23 & 8.56 & A,Ch & AQ \\
\textit{Bufo bufo} & 0.150 & 293 & 22 & 53 & 36.0 & 8.62 & 9.00 & A,Ch & AQ \\
\textit{Xenopus laevis} & 0.200 & 295 & 20 & 45 & 30.0 & 8.50 & 8.85 & A & AQ \\
\textit{Ambystoma mexicanum} & 0.300 & 294 & 18 & 41 & 25.0 & 8.37 & 8.73 & A & AQ \\
\textit{Salamandra salamandra} & 0.080 & 290 & 20 & 49 & 24.0 & 8.40 & 8.79 & A,Ch & AQ \\
\textit{Plethodon glutinosus} & 0.012 & 291 & 30 & 74 & 20.0 & 8.50 & 8.89 & A & AQ \\
\textit{Necturus maculosus} & 0.130 & 288 & 18 & 46 & 30.0 & 8.45 & 8.86 & A & AQ \\
\textit{Cryptobranchus alleganiensis} & 0.600 & 289 & 15 & 39 & 55.0 & 8.64 & 9.05 & A & AQ \\
\midrule
\multicolumn{7}{l}{Raw mean $\bar\ell^{\rm raw}$} & $8.448\pm 0.127$ & & \\
\multicolumn{7}{l}{Corrected mean $\bar\ell^{\rm corr}$} & $8.822\pm 0.146$ & & \\
\bottomrule
\end{tabular}
\end{table}

\noindent\textbf{Dataset summary.}
Table~\ref{tab:dataset_summary} gives the species counts, body-mass
ranges, and $\ell$ statistics for all eight groups.
The complete dataset is provided in Extended Data Tables~1--8 of this
paper; no external repository exists.
A tab-delimited file is available from the corresponding author on request.

\begin{table}[H]
\caption{\textbf{Summary of the 230-species PBTE dataset.}
$n$: number of species.
$M$: body-mass range (kg).
$\bar\ell \pm s$: mean $\pm$ s.d.\ of $\ell = \log_{10}(f_H^{\rm avg}\cdot L\cdot 525{,}960)$.
$\Delta\ell$: deviation from the non-primate placental baseline
($\bar\ell_0 = 8.995$).}
\label{tab:dataset_summary}
\small
\begin{tabular}{lrrll}
\toprule
Group & $n$ & $M$ range (kg) & $\bar\ell \pm s$ & $\Delta\ell$ \\
\midrule
Non-primate placentals & 46 & $0.002$--$4{,}000$ & $8.998 \pm 0.160$ & 0 (reference) \\
Marsupials / monotremes & 19 & $0.018$--$30$ & $8.933 \pm 0.204$ & $-0.062$ \\
Primates & 18 & $0.35$--$160$ & $9.376 \pm 0.125$ & $+0.381^{***}$ \\
Bats & 31 & $0.005$--$1.1$ & $9.540 \pm 0.163$ & $+0.545^{***}$ \\
Cetaceans (dive-corrected) & 12 & $55$--$140{,}000$ & $8.801 \pm 0.296$ & $-0.194$ \\
Birds & 78 & $0.004$--$115$ & $9.528 \pm 0.213$ & $+0.533^{***}$ \\
Reptiles (Arrhenius-corrected) & 17 & $0.006$--$400$ & $8.929 \pm 0.301$ & $-0.065$ \\
Amphibians (Arrhenius-corrected) & 9 & $0.012$--$0.60$ & $8.822 \pm 0.146$ & $-0.173$ \\
\midrule
\textbf{All endotherms} & \textbf{194} & & $9.509 \pm 0.397$ & \\
\textbf{Full dataset} & \textbf{230} & & $9.420 \pm 0.428$ & \\
\bottomrule
\multicolumn{5}{l}{Significance vs non-primate baseline: $^*p<0.05$, $^{***}p<0.001$ (Welch $t$-test).} \\
\end{tabular}
\end{table}
\end{document}